\documentclass[12pt,american]{article}
\usepackage[LGR,T1]{fontenc}
\usepackage{geometry}
\usepackage{babel}
\usepackage{array}
\usepackage{booktabs}
\usepackage{mathrsfs}
\usepackage{enumitem}
\usepackage{multirow}
\usepackage{amsmath}
\usepackage{amsthm}
\usepackage{amssymb}
\usepackage{setspace}
\usepackage[authoryear]{natbib}
\usepackage[unicode=true,pdfusetitle,
 bookmarks=true,bookmarksnumbered=false,bookmarksopen=false,
 breaklinks=false,pdfborder={0 0 1},backref=false,colorlinks=false]
 {hyperref}
\hypersetup{
 colorlinks,linkcolor={blue!70!black},citecolor={blue!50!black},urlcolor={blue!60!black}}

\makeatletter

\DeclareRobustCommand{\greektext}{%
  \fontencoding{LGR}\selectfont\def\encodingdefault{LGR}}
\DeclareRobustCommand{\textgreek}[1]{\leavevmode{\greektext #1}}
\ProvideTextCommand{\~}{LGR}[1]{\char126#1}

\providecommand{\tabularnewline}{\\}

\theoremstyle{plain}
\newtheorem{assumption}{\protect\assumptionname}
\theoremstyle{plain}
\newtheorem{cor}{\protect\corollaryname}
\theoremstyle{plain}
\newtheorem{lem}{\protect\lemmaname}
\theoremstyle{plain}
\newtheorem{thm}{\protect\theoremname}
\theoremstyle{definition}
\newtheorem{defn}{\protect\definitionname}

\usepackage{lmodern}
\usepackage{babel}
\usepackage{float}
\usepackage{mathrsfs}
\usepackage{breakurl}
\usepackage{bbm}
\usepackage{tikz}
\usepackage{centernot}
\usepackage{amssymb,amsmath,amsthm,enumitem}

\allowdisplaybreaks

\newcommand{\customlabel}[2]{%
   \protected@write \@auxout {}{\string \newlabel {#1}{{#2}{\thepage}{#2}{#1}{}} }%
   \hypertarget{#1}{}
}

\makeatother

\providecommand{\assumptionname}{Assumption}
\providecommand{\corollaryname}{Corollary}
\providecommand{\definitionname}{Definition}
\providecommand{\lemmaname}{Lemma}
\providecommand{\theoremname}{Theorem}

\begin{document}
\global\long\def\a{\alpha}%
 
\global\long\def\b{\beta}%
 
\global\long\def\g{\gamma}%
 
\global\long\def\d{\delta}%
 
\global\long\def\e{\epsilon}%
 
\global\long\def\l{\lambda}%
 
\global\long\def\t{\theta}%
 
\global\long\def\o{\omega}%
 
\global\long\def\s{\sigma}%

\global\long\def\G{\Gamma}%
 
\global\long\def\D{\Delta}%
 
\global\long\def\L{\Lambda}%
 
\global\long\def\T{\Theta}%
 
\global\long\def\O{\Omega}%
 
\global\long\def\R{\mathbb{R}}%
 
\global\long\def\N{\mathbb{N}}%
 
\global\long\def\Q{\mathbb{Q}}%
 
\global\long\def\I{\mathbb{I}}%
 
\global\long\def\P{\mathbb{P}}%
 
\global\long\def\E{\mathbb{E}}%
\global\long\def\B{\mathbb{\mathbb{B}}}%
\global\long\def\S{\mathbb{\mathbb{S}}}%
\global\long\def\V{\mathbb{\mathbb{V}}\text{ar}}%
 
\global\long\def\GG{\mathbb{G}}%

\global\long\def\X{{\bf X}}%
\global\long\def\cX{\mathscr{X}}%
 
\global\long\def\cY{\mathscr{Y}}%
 
\global\long\def\cA{\mathscr{A}}%
 
\global\long\def\cB{\mathscr{B}}%
\global\long\def\cF{\mathscr{F}}%
 
\global\long\def\cM{\mathscr{M}}%
\global\long\def\cN{\mathcal{N}}%
\global\long\def\cG{\mathcal{G}}%
\global\long\def\cC{\mathcal{C}}%
\global\long\def\sp{\,}%

\global\long\def\es{\emptyset}%
 
\global\long\def\mc#1{\mathscr{#1}}%
 
\global\long\def\ind{\mathbf{\mathbbm1}}%
\global\long\def\indep{\perp}%

\global\long\def\any{\forall}%
 
\global\long\def\ex{\exists}%
 
\global\long\def\p{\partial}%
 
\global\long\def\cd{\cdot}%
 
\global\long\def\Dif{\nabla}%
 
\global\long\def\imp{\Rightarrow}%
 
\global\long\def\iff{\Leftrightarrow}%

\global\long\def\up{\uparrow}%
 
\global\long\def\down{\downarrow}%
 
\global\long\def\arrow{\rightarrow}%
 
\global\long\def\rlarrow{\leftrightarrow}%
 
\global\long\def\lrarrow{\leftrightarrow}%

\global\long\def\abs#1{\left|#1\right|}%
 
\global\long\def\norm#1{\left\Vert #1\right\Vert }%
 
\global\long\def\rest#1{\left.#1\right|}%

\global\long\def\bracket#1#2{\left\langle #1\middle\vert#2\right\rangle }%
 
\global\long\def\sandvich#1#2#3{\left\langle #1\middle\vert#2\middle\vert#3\right\rangle }%
 
\global\long\def\turd#1{\frac{#1}{3}}%
 
\global\long\def\ellipsis{\textellipsis}%
 
\global\long\def\sand#1{\left\lceil #1\right\vert }%
 
\global\long\def\wich#1{\left\vert #1\right\rfloor }%
 
\global\long\def\sandwich#1#2#3{\left\lceil #1\middle\vert#2\middle\vert#3\right\rfloor }%

\global\long\def\abs#1{\left|#1\right|}%
 
\global\long\def\norm#1{\left\Vert #1\right\Vert }%
 
\global\long\def\rest#1{\left.#1\right|}%
 
\global\long\def\inprod#1{\left\langle #1\right\rangle }%
 
\global\long\def\ol#1{\overline{#1}}%
 
\global\long\def\ul#1{\underline{#1}}%
 
\global\long\def\td#1{\tilde{#1}}%
\global\long\def\bs#1{\boldsymbol{#1}}%

\global\long\def\upto{\nearrow}%
 
\global\long\def\downto{\searrow}%
 
\global\long\def\pto{\overset{p}{\longrightarrow}}%
 
\global\long\def\dto{\overset{d}{\longrightarrow}}%
 
\global\long\def\asto{\overset{a.s.}{\longrightarrow}}%

\setlength{\abovedisplayskip}{6pt} \setlength{\belowdisplayskip}{6pt}
\title{Two-Stage Maximum Score Estimator\thanks{We thank Xiaohong Chen, Xu Cheng, Frank Diebold, Ivan Fern\'{a}ndez-Val,
Simon Lee, Ming Li, Konrad Menzel, Frank Schorfheide, Matt Seo, Peter
Phillips, Joris Pinkse, Yuanyuan Wan, as well as seminar and conference
participants at Syracuse, U Toronto, USC, BU, NUS \& SMU, NYU, the
2022 Cowles Foundation Summer Conference on Econometrics and the 2022
Asian Meeting of the Econometric Society for helpful comments and
suggestions.}}
\author{Wayne Yuan Gao\thanks{Gao: Department of Economics, University of Pennsylvania, 133 S 36th
St., Philadelphia, PA 19104, USA, waynegao@upenn.edu.},$\ $Sheng Xu\thanks{Xu: The Program in Applied and Computational Mathematics, Princeton
University, Fine Hall, Washington Road, Princeton, NJ 08544, sx7392@princeton.edu.}, and Kan Xu\thanks{Xu: Department of Economics, University of Pennsylvania, 133 S 36th
St., Philadelphia, PA 19104, USA, kanxu@upenn.edu.}}
\maketitle
\begin{abstract}
\noindent This paper considers the asymptotic theory of a semiparametric
M-estimator that is generally applicable to models that satisfy a
monotonicity condition in one or several parametric indexes. We call
this estimator the \emph{two-stage maximum score} (TSMS) estimator,
since our estimator involves a first-stage nonparametric regression
when applied to the binary choice model of \citet*{manski1975maximum,manski1985semiparametric}.
We characterize the asymptotic distribution of the TSMS estimator,
which features phase transitions depending on the dimension of the
first-stage estimation. Effectively, the first-stage nonparametric
estimator serves as an imperfect smoothing function on a non-smooth
criterion function, leading to the pivotality of the first-stage estimation
error with respect to the second-stage convergence rate and asymptotic
distribution. \textbf{}\\
\textbf{~}\\
\textbf{Keywords:} semiparametric M-estimation, maximum score, non-smooth
criterion, monotone index, discrete choice
\end{abstract}
\newpage{}

\section{\label{sec:Intro}Introduction}

In a sequence of papers \citet{manski1975maximum,manski1985semiparametric}
proposed and analyzed the \emph{maximum-score estimator} for semiparametric
discrete choice models, e.g.,
\[
y_{i}=\ind\left\{ X_{i}^{'}\t_{0}\geq\e_{i}\right\} 
\]
based on a median normalization $\text{med}\left(\rest{\e_{i}}X_{i}\right)=0$
and the consequent observation
\begin{equation}
h_{0}\left(X_{i}\right):=\E\left[\rest{y_{i}-\frac{1}{2}}X_{i}\right]\gtrless0\quad\iff\quad X_{i}^{'}\t_{0}\gtrless0.\label{eq:Mono_Equiv}
\end{equation}
Specifically, the maximum-score estimator is defined as any solution
to the problem
\[
\max_{\t}\frac{1}{n}\sum_{i=1}^{n}\left(y_{i}-\frac{1}{2}\right)\ind\left\{ X_{i}^{'}\t\geq0\right\} .
\]
Subsequently, \citet*{kim1990cube} demonstrated the cubic-root asymptotics
of the maximum-score estimator with a non-normal limit distribution,
and \citet*{horowitz1992smoothed} showed the asymptotic normality
of the \emph{smoothed} maximum score estimator\footnote{The smoothed maximum score estimator is defined as the solution to
$\max_{\t}\frac{1}{n}\sum_{i=1}^{n}\left(y_{i}-\frac{1}{2}\right)\Phi\left(X_{i}^{'}\t/b_{n}\right)$
with a chosen smooth function $\Phi$ and bandwidth $b_{n}$.} with a faster-than-$n^{-1/3}$ but slower-than-$n^{-1/2}$ convergence
rate.

In this paper we consider yet another estimator of the model above,
which we call the \emph{two-stage maximum score} (TSMS) estimator,
defined as any solution to
\[
\max_{\t}\frac{1}{n}\sum_{i=1}^{n}\hat{h}\left(X_{i}\right)\ind\left\{ X_{i}^{'}\t\geq0\right\} ,
\]
where $\hat{h}$ is a consistent first-stage nonparametric estimator
of $h_{0}$. Essentially, the TSMS estimator encodes the logical relationship
\eqref{eq:Mono_Equiv} in a more \emph{literal} way: we simply replace
$h_{0}$ in \eqref{eq:Mono_Equiv} with its estimator $\hat{h}$.
We focus on analyzing the asymptotic properties of the TSMS estimator
in this paper.

~

\noindent The applicability of the TSMS estimator, however, extends
far beyond the binary choice model considered above. Consider any
model such that some nonparametrically identified function of data
$h_{0}$ and a finite-dimensional parameter of interest $\t_{0}$
satisfy the following multi-index monotonicity condition (at zero):
with $X:=\left(X_{1},...,X_{J}\right)$,
\begin{align}
X_{j}^{'}\t_{0}>0\ \text{for every }j=1,...,J\quad & \imp\quad h_{0}\left(X\right)>0,\nonumber \\
X_{j}^{'}\t_{0}<0\ \text{for every }j=1,...,J\quad & \imp\quad h_{0}\left(X\right)<0.\label{eq:Mono_MMI}
\end{align}
Clearly \eqref{eq:Mono_MMI} nests \eqref{eq:Mono_Equiv} as special
case with $J=1$. However, as we move to multi-index settings with
$J\geq2$, the \emph{logical equivalence} relationship between the
sign of $h_{0}\left(X\right)$ and the sign of the parametric indexes
encoded in \eqref{eq:Mono_Equiv} is broken. Instead, \eqref{eq:Mono_MMI}
are stated as \emph{logical implication}s, whose converses may not
be generally true for $J\geq2$:
\begin{align*}
h_{0}\left(X\right)>0 & \quad\centernot\imp\quad X_{j}^{'}\t_{0}>0\ \text{for every }j=1,...,J,\\
h_{0}\left(X\right)<0 & \quad\centernot\imp\quad X_{j}^{'}\t_{0}>0\ \text{for every }j=1,...,J.
\end{align*}
On the other hand, instead of using the logical converses above, we
can leverage the \emph{logical contrapositions} of \eqref{eq:Mono_MMI}
as proposed in \citet*{gao2019robust}:
\begin{align}
h_{0}\left(X\right)>0 & \quad\imp\quad\text{NOT}\ \left(X_{j}^{'}\t_{0}<0\ \text{for every }j=1,...,J\right),\nonumber \\
h_{0}\left(X\right)<0 & \quad\imp\quad\text{NOT}\ \left(X_{j}^{'}\t_{0}>0\ \text{for every }j=1,...,J\right),\label{eq:ID_MMI}
\end{align}

\noindent which serve as identifying restrictions on $\t_{0}$, given
that $h_{0}$ is directly identified and can be nonparametrically
estimated from data. The TSMS estimator in the monotone multi-index
setting can then be formulated as any solution to 
\begin{align}
\max_{\t}\ -\frac{1}{n}\sum_{i=1}^{n}\left\{ \left[\hat{h}\left(X_{i}\right)\right]_{+}\prod_{j=1}^{J}\ind\left\{ X_{ij}^{'}\t<0\right\} +\left[-\hat{h}\left(X_{i}\right)\right]_{+}\prod_{j=1}^{J}\ind\left\{ X_{ij}^{'}\t>0\right\} \right\} ,\label{eq:Q0_MI}
\end{align}
where $\left[\cd\right]_{+}$ is the positive part (or ``rectifier'')
function. It is important to note that the right hand sides of \eqref{eq:ID_MMI}
are not negations of each other, i.e.,
\[
\prod_{j=1}^{J}\ind\left\{ X_{ij}^{'}\t<0\right\} \neq1-\prod_{j=1}^{J}\ind\left\{ X_{ij}^{'}\t>0\right\} ,
\]
thus we have to multiply $\left[\hat{h}\left(X_{i}\right)\right]_{+}$
and $\left[-\hat{h}\left(X_{i}\right)\right]_{+}$ with indicators
of very different sets. Hence, there are no counterparts of the original
maximum score or smoothed maximum score estimators in this setting,
while the TSMS estimator will still be consistent (under conditions
for point identification).

For example, \citet*{gao2019robust} considers a semiparametric panel
multinomial choice model, where infinite-dimensional fixed effects
are allowed to enter into consumer utilities in an additively nonseparble
way. Despite the complexity of the incorporated unobserved heterogeneity,
a certain form of intertemporal differences in conditional choice
probabilities satisfy \eqref{eq:ID_MMI}. In another paper, \citet*{GLX2018network}
study a dyadic network formation with nontransferable utilities, where
the formation of a link requires bilateral consent from the two involved
individuals. With a technique called \emph{logical differencing} that
cancels out the nonadditive unobserved heterogeneity terms in the
model, a nonparametrically estimable function can again be constructed
to satisfy \eqref{eq:ID_MMI}. In both papers, the TSMS estimators
are used to provide consistent estimates for the parameter of interest.
There are likely to be many other applications where the TSMS estimators
can be particularly useful, given that the logical implication relationships
in \eqref{eq:ID_MMI} can arise naturally in economic models that
possess certain monotonicity properties.

~

\noindent Motivated by the reasons discussed above, we seek to analyze
the asymptotic properties of the TSMS estimator in this paper. Since
the key differences between the TSMS estimator and the (smoothed)
maximum score estimator in terms of their \emph{asymptotic properties}
do not really depend on the number of indexes $J$\footnote{The difference in asymptotic properties should not be confused with
the differences in identification strategies, which are discussed
above.}, we first focus on deriving the convergence rate and asymptotic distribution
of the TSMS estimator in a simple binary choice model, where the key
drivers of the non-standard asymptotics for the TSMS estimator can
be best explained and compared.

Using a kernel first-step estimator, we find that the asymptotics
for the TSMS estimator feature two phase transitions, the thresholds
of which depends on the dimensionality and the order of smoothness
built in the model.

First, when the dimension of covariates is low relative to the order
of smoothness, the TSMS estimator is asymptotically equivalent to
the smoothed maximum score estimator, achieving the same convergence
rate and a corresponding normal asymptotic distribution. This is a
case where the first-stage nonparametric estimator serves as a smoothing
function on the discrete indicator function in the\emph{ best possible}
manner, delivering full ``speed-up'' from the $n^{-1/3}$ rate of
the original maximum score estimator and attaining the minimax-optimal
rate of the smooth maximum score estimator.

Second, when the dimension of covariates is moderate, the TSMS estimator
converges at a rate slower than $n^{-2/5}$ but faster than $n^{-1/3}$,
and has an asymptotic distribution characterized by the maximizer
of a Gaussian process plus a linear (bias) and a quadratic drift terms.
This is a scenario where the first-stage nonparametric estimation
plays a \emph{partially effective }role as a smoothing function: it
dampens the effect of the discreteness of the indicator function,
but the estimation error from the first-stage is too large (due to
the dimension of the first-stage estimation) to be negligible. It
turns out that a composite mean-zero error term of partial smoothing
on indicator function is asymptotically at the same order of the bias
from the first-stage estimation, hence leading to a Gaussian process
as well as a bias term in the limit.

Third, when the dimension of covariates is relatively high, the TSMS
estimator converges at a rate slower than $n^{-1/3}$ that decreases
with the dimension of covariates, and its asymptotic distribution
(without debiasing) is degenerate at a bias term. The (mean-zero)
disturbance term stays roughly at $n^{-1/3}$-rate, but it is dominated
by the bias from the first-stage estimation. The result is intuitive,
given that the performance of TSMS must be fundamentally dependent
on the performance of the first-stage nonparametric estimation.

Lastly, we extend the results on convergence rate beyond the binary
choice setting to monotone mult-index models.

~

\noindent As discussed above, our paper contributes to the line of
econometric literature on maximum score or rank-order estimation that
exploits monotonicity restrictions, as studied in \citet{manski1975maximum,manski1985semiparametric},
\citet*{kim1990cube}, \citet*{han1987non}, \citet*{horowitz1992smoothed}
and \citet*{abrevaya2000rank}, for example. Relatedly, the analysis
of the discreteness effects of indicator functions and the feature
of phase transition in asymptotic theories are also present in threshold
and change-point models: e.g. \citet*{banerjee2007confidence}, \citet*{lee2008semiparametric},
\citet*{kosorok2007introduction}, \citet*{song2016asymptotics},
\citet{lee2018oracle}, \citet*{hidalgo2019robust}, \citet*{lee2018factor}
and \citet*{mukherjee2020asymptotic}.

The technical part of this paper builds upon and contributes to the
large line of econometric literature on semi/non-parametric estimation.
General methods and techniques used in this paper are based on \citet*{andrews1994asymptotics},
\citet*{newey1994asymptotic}, \citet*{newey1994large}, \citet*{van1996weak},
\citet*{chen2007sieve}, \citet*{hansen2008uniform} and \citet*{kosorok2007introduction}.
More specifically, the handling of the non-smooth criterion functions
is also studied in \citet*{kim1990cube}, \citet*{chen2003estimation},
\citet*{seo2018local} and \citet*{delsol2020semiparametric}. However,
our asymptotic theory covers an intermediate case of non-smoothness
that leads to a convergence rate faster than cubic-root-style rate
obtained in \citet*{kim1990cube}, \citet*{seo2018local} and the
example considered in \citet*{delsol2020semiparametric}, but faster
than the root-$n$ rate considered by \citet*{chen2003estimation}.
This is due to a pivotal interplay between the smoothing provided
by the first-stage nonparametric estimation and its estimation error,
which appears to be an interesting feature unique to our TSMS estimator.

Lastly, this paper complements the work in \citet*{gao2019robust}
and \citet*{GLX2018network} by providing a formal analysis of the
asymptotic theory for the TSMS estimator.

\section{\label{sec:BinaryChoice}TSMS Estimator in Binary Choice Model}

We start with an analytical illustration of the two-stage maximum
score estimator in a binary choice setting, where the TSMS estimator
can be very clearly related to and compared with existing results
in the literature, in particular \citet{manski1975maximum,manski1985semiparametric},
\citet*{kim1990cube}, \citet*{horowitz1992smoothed} and \citet*{seo2018local}.
To better convey the key ideas, in this section we will impose several
simplifying assumptions that are stronger than necessary. We refer
the readers to Section for a more general treatment.

\subsection{\label{subsec:Setup} Model Setup}

Consider the following model a la \citet{manski1975maximum,manski1985semiparametric}:
\begin{equation}
y_{i}=\ind\left\{ X_{i}^{'}\t_{0}\geq\e_{i}\right\} ,\label{eq:model}
\end{equation}
where $y_{i}$ is an observed binary outcome variable, $X_{i}$ is
a vector of observed covariates taking values in $\R^{d}$, $\t_{0}\in\R^{d}$
is the unknown true parameter, and $\e_{i}$ is an unobserved scalar
random variable that satisfies the conditional median restriction
$\text{med}\left(\rest{\e_{i}}X_{i}\right)=0.$ Defining
\begin{equation}
Q_{0}\left(\t\right):=\E\left[\left(y_{i}-\frac{1}{2}\right)\ind\left\{ X_{i}^{'}\t\geq0\right\} \right],\label{eq:Q0}
\end{equation}
we know by \citet{manski1975maximum,manski1985semiparametric}, under
appropriate conditions, $\t_{0}$ is the unique maximizer of $Q_{0}$
on 
\[
\S^{d-1}:=\left\{ u\in\R^{d}:\norm u=1\right\} ,
\]
based on which the maximum score (MS thereafter) estimator is constructed
as
\begin{equation}
\hat{\t}_{MS}:\in\arg\max_{\t\in\S^{d-1}}\frac{1}{n}\sum_{i=1}^{n}\left(y_{i}-\frac{1}{2}\right)\ind\left\{ X_{i}^{'}\t\geq0\right\} .\label{eq:est_MS}
\end{equation}
\citet*{kim1990cube} demonstrated the cubic-root asymptotics of the
MS estimator $n^{\frac{1}{3}}\left(\hat{\b}_{MS}-\b_{0}\right)\dto\arg\max_{s\in\S^{D-1}}Z\left(s\right).$
Alternatively, \citet*{horowitz1992smoothed} considered the smoothed
maximum score (SMS thereafter) estimator
\begin{equation}
\hat{\t}_{SMS}:=\arg\max_{\t:\left|\t_{1}\right|=1}\frac{1}{n}\sum_{i=1}^{n}\left(y_{i}-\frac{1}{2}\right)\Phi\left(\frac{X_{i}^{'}\t}{b_{n}}\right)\label{eq:SMS}
\end{equation}
under the alternative normalization $\left|\t_{1}\right|=1$, where
$\Phi:\R\to\left[0,1\right]$ is a smooth kernel function and $b_{n}$
is a tuning parameter that shrinks towards $0$ as $n\to\infty$.
By \citet*{horowitz1992smoothed} the SMS estimator is asymptotically
normal with a convergence rate of $n^{-2/5}$ when, say, the kernel
function $\Phi$ is taken to be the CDF of the standard normal distribution.
More precisely, writing $\hat{\t}_{SMS}\equiv\left(\hat{\t}_{1,SMS},\tilde{\t}_{SMS}\right)$,
we have $n^{-\frac{2}{5}}\left(\tilde{\t}_{SMS}-\tilde{\t}_{0}\right)\dto\cN\left(\mu_{SMS},\Sigma_{SMS}\right)$
for some deterministic $\mu_{SMS}$ and $\Sigma_{SMS}$. Moreover,
with high-order kernel functions, the rate could be improved to be
arbitrarily close to $n^{-1/2}$.

In this paper we consider yet another form of estimator, which we
call ``two-step maximum score (TSMS) estimator'', based on exactly
the same population criterion function $Q_{0}$ defined above in \eqref{eq:Q0}.
Observing that $Q_{0}$ can be equivalently written as
\[
Q_{0}\left(\t\right)=\E\left[h_{0}\left(X_{i}\right)\ind\left\{ X_{i}^{'}\t\geq0\right\} \right]
\]
with 
\[
h_{0}\left(x\right):=\E\left[\rest{y_{i}}X_{i}=x\right]-\frac{1}{2},
\]
we define the TSMS estimator as
\begin{equation}
\hat{\t}:\in\arg\max_{\t\in\S^{d-1}}\frac{1}{n}\sum_{i=1}^{n}\hat{h}\left(X_{i}\right)\ind\left\{ X_{i}^{'}\t\geq0\right\} ,\label{eq:est_TSMS}
\end{equation}
where $\hat{h}$ is any first-stage nonparametric estimator of $h_{0}$.
\begin{assumption}
\label{assu:Basic} Write ${\cal X}:=\text{Supp}\left(X_{i}\right)\subseteq\R^{d}$
and suppose $\t_{0}\in\S^{d-1}$. Assume the following:
\begin{itemize}
\item[(a)]  $\left(y_{i},X_{i},\e_{i}\right)_{i=1}^{n}$ is i.i.d. and satisfies
model \eqref{eq:model}.
\item[(b)]  The (unknown) conditional CDF $F\left(\rest{\e}x\right)$ of $\e_{i}$
given $X_{i}=x$ is twice continuously differentiable w.r.t. $\left(\e,x\right)\in\R\times{\cal X}$
with uniformly bounded first and second derivatives (bounded by some
positive constant $M<\infty$).
\item[(c)] The conditional PDF $f\left(\rest{\e}x\right)$ of $\e_{i}$ given
$X_{i}=x$ is strictly positive for any $\e\in\R$ and $x\in{\cal X}$.
\item[(d)]  The conditional median of $\e_{i}$ given $X_{i}=x$ is zero, i.e.,
\[
F\left(\rest 0x\right)=\frac{1}{2},\quad\forall x\in{\cal X}.
\]
\item[(e)]  $X_{i}$ is uniformly distributed with support given by the open
unit ball in $\R^{d}$, i.e.,
\[
{\cal X}=\B^{d}:=\left\{ x\in\R^{d}:\norm x<1\right\} .
\]
\end{itemize}
\end{assumption}
\noindent Under Assumption \eqref{assu:Basic}, it is easy to show
that $\t_{0}$ is point identified as the unique maximizer of $Q_{0}$
over $\S^{d-1}$.

Furthermore, we note that the smoothness condition in Assumption \eqref{assu:Basic}(b)
imply the following smoothness condition on the unknown function $h_{0}\left(x\right):=\E\left[\rest{y_{i}-\frac{1}{2}}X_{i}=x\right]$.
\begin{cor}
\label{cor:h0_dif_bound}Under Assumption \ref{assu:Basic}(b), $h_{0}\left(x\right)$
is twice differentiable w.r.t. $x$ with uniformly bounded first and
second derivatives. 
\end{cor}

\subsection{\label{subsec:Bin_Rate}Asymptotic Theory}

Before presenting the formal results, we first explain how our TSMS
estimator differs from the MS and the SMS estimator, and provide some
intuitions about the key features of the asymptotics of the TSMS estimator.
For this purpose we write
\begin{align*}
g_{i}^{MS}\left(\t\right) & :=\left(y_{i}-\frac{1}{2}\right)\ind\left\{ X_{i}^{'}\t\geq0\right\} ,\\
g_{i}^{SMS}\left(\t\right) & :=\left(y_{i}-\frac{1}{2}\right)\Phi\left\{ \frac{X_{i}^{'}\t}{b_{n}}\right\} ,\\
g_{i}^{TSMS}\left(\t\right) & :=\quad\hat{h}\left(X_{i}\right)\ \ind\left\{ X_{i}^{'}\t\geq0\right\} ,
\end{align*}
which are the (random) functions of $\t$ being averaged into the
sample criterion for the MS, TMS and TSMS estimators above in \eqref{eq:est_MS},
\eqref{eq:SMS} and \eqref{eq:est_TSMS}.

Notice first that the indicator function $\ind\left\{ X_{i}^{'}\t\geq0\right\} $
in $g_{i}^{TSMS}\left(\t\right)$ is not smoothed out by a CDF-type
kernel function as in $g_{i}^{SMS}\left(\t\right)$. Consequently,
our TSMS sample criterion is discontinuous in $\t$ while having zero
derivative with respect to $\t$ almost everywhere, and thus we cannot
characterize the TSMS estimator by first-order conditions as in \citet{horowitz1992smoothed}.
More generally, we cannot directly use existing asymptotic theories
based on the (Lipschitz) continuity and differentiability of the criterion
function in parameters.

In the meanwhile, the TSMS sample criterion is also very different
from the original MS sample criterion, as in $g_{i}^{MS}\left(\t\right)$,
the term $\left(y_{i}-\frac{1}{2}\right)$ is also discrete in addition
to the indicator function $\ind\left\{ X_{i}^{'}\t\geq0\right\} $.
As explained in \citet*{kim1990cube}, for $\t$ close to $\t_{0}$,
the expected squared difference between $g_{i}^{MS}\left(\t\right)$
and $g_{i}^{MS}\left(\t_{0}\right)$:
\begin{align}
\E\left|g_{i}^{MS}\left(\t\right)-g_{i}^{MS}\left(\t_{0}\right)\right|^{2} & =\E\left|\ind\left\{ X_{i}^{'}\t\geq0\right\} -\ind\left\{ X_{i}^{'}\t_{0}\geq0\right\} \right|=O\left(\norm{\t-\t_{0}}\right)\label{eq:MS_Envelop}
\end{align}
is of the same order of magnitude as $\norm{\t-\t_{0}}$, which is
the key driver for the cubic-root asymptotics. However, in our case
\begin{align*}
\E\left|g_{i}^{TSMS}\left(\t\right)-g_{i}^{TSMS}\left(\t_{0}\right)\right|^{2} & =\E\left[\hat{h}^{2}\left(X_{i}\right)\left|\ind\left\{ X_{i}^{'}\t\geq0\right\} -\ind\left\{ X_{i}^{'}\t_{0}\geq0\right\} \right|\right]
\end{align*}
where $\hat{h}^{2}\left(X_{i}\right)$ enters as a weighting on the
discrete difference in indicators. As it turns out, $\hat{h}\left(X_{i}\right)$
will actually help smooth out the indicator function and making the
expected squared difference above to be smaller than $\norm{\t-\t_{0}}$,
even though $\hat{h}\left(X_{i}\right)$ itself does not depend on
$\t$. 

To see this, notice that whenever $\ind\left\{ x^{'}\t\geq0\right\} \neq\ind\left\{ x^{'}\t_{0}\geq0\right\} $
occurs, $0$ must lie between $x^{'}\t$ and $x^{'}\t_{0}$. Consider
first the case of 
\begin{equation}
x^{'}\t_{0}\geq0>x^{'}\t.\label{eq:case_one}
\end{equation}
When $\t$ is close to $\t_{0}$ in the sense of $\norm{\t-\t_{0}}$
being very close to $0$, the difference between $x^{'}\t$ and $x^{'}\t_{0}$
must also be small, since
\[
\left|x^{'}\t-x^{'}\t_{0}\right|\leq\norm x\norm{\t-\t_{0}}\leq\norm{\t-\t_{0}}.
\]
Hence, together with \eqref{eq:case_one} we have
\[
x^{'}\t_{0}\geq0>x^{'}\t=x^{'}\t_{0}+\left(x^{'}\t-x^{'}\t_{0}\right)\geq x^{'}\t_{0}-\norm{\t-\t_{0}},
\]
which implies that
\[
0\leq x^{'}\t_{0}<\norm{\t-\t_{0}},
\]
Now, define
\[
\ol x:=x-\norm{\t-\t_{0}}\t_{0}^{'},
\]
we have $\ol x^{'}\t_{0}=x^{'}\t_{0}-\norm{\t-\t_{0}}<0$ and hence
\[
h_{0}\left(\ol x\right)=F\left(\rest{\ol x^{'}\t_{0}}\ol x\right)-\frac{1}{2}<F\left(\rest 0\ol x\right)-\frac{1}{2}=0.
\]
However, by \eqref{eq:case_one} we have $x^{'}\t_{0}\geq0$ and thus
\[
h_{0}\left(x\right)=F\left(\rest{x^{'}\t_{0}}x\right)-\frac{1}{2}\geq F\left(\rest 0x\right)-\frac{1}{2}=0.
\]
By Lemma \ref{cor:h0_dif_bound}, we then have
\begin{align*}
h_{0}\left(x\right)\geq0>h_{0}\left(\ol x\right) & =h_{0}\left(x\right)+\Dif_{x}h_{0}\left(\tilde{x}\right)\left(\ol x-x_{0}\right)\\
 & >h_{0}\left(x\right)-\sup_{\tilde{x}}\left|\Dif_{x}h_{0}\left(\tilde{x}\right)\right|\cd\norm{\ol x-x_{0}}\\
 & \geq h_{0}\left(x\right)-M\cd\norm{\t-\t_{0}}\cd1
\end{align*}
which implies that
\[
0\leq h_{0}\left(x\right)\leq M\cd\norm{\t-\t_{0}}.
\]
A similar argument applies to the case of
\[
x^{'}\t_{0}<0\leq x^{'}\t,
\]
which implies that
\[
0>h_{0}\left(x\right)>-M\cd\norm{\t-\t_{0}}.
\]
Together, we have 
\begin{align*}
\ind\left\{ x^{'}\t\geq0\right\} \neq\ind\left\{ x^{'}\t_{0}\geq0\right\} \quad & \imp\quad\left|x^{'}\t_{0}\right|\leq\norm{\t-\t_{0}}\\
 & \imp\quad h_{0}\left(x\right)\leq M\norm{\t-\t_{0}}
\end{align*}
and thus
\[
h_{0}\left(x\right)\left|\ind\left\{ x^{'}\t\geq0\right\} -\ind\left\{ x^{'}\t_{0}\geq0\right\} \right|\leq M\norm{\t-\t_{0}},
\]
i.e., $h_{0}\left(x\right)$ automatically shrinks any nonzero difference
between the two indicators $\ind\left\{ x^{'}\t\geq0\right\} $ and
$\ind\left\{ x^{'}\t_{0}\geq0\right\} $ as $\t$ gets closer to $0$.
This results in
\[
\E\left[h_{0}^{2}\left(X_{i}\right)\left|\ind\left\{ X_{i}^{'}\t\geq0\right\} -\ind\left\{ X_{i}^{'}\t_{0}\geq0\right\} \right|\right]=o\left(\norm{\t-\t_{0}}\right),
\]
which contrasts sharply with the $O\left(\norm{\t-\t_{0}}\right)$
magnitude on the right-hand side of \eqref{eq:MS_Envelop}.

The discussion above will be formally captured by Lemma \ref{lem:Term1}. 

~

\noindent We now proceed to a formal development of the TSMS asymptotic
theory. For any $\t\in\T$ and any (deterministic) function $h:\R^{d}\to\R$
in $L_{2\left(X\right)}$, write 
\begin{align*}
g_{\t,h}\left(x\right) & :=h\left(x\right)\ind\left\{ x^{'}\t>0\right\} ,\ \forall x\in\R^{d},\\
Pg_{\t,h} & :=\int g_{\t,h}\left(x\right)dP\left(x\right),\\
\P_{n}g_{\t,h} & :=\frac{1}{n}\sum_{i=1}^{n}g_{\t,h}\left(X_{i}\right).\\
\GG_{n}g_{\t,h} & :=\sqrt{n}\left(\P_{n}g_{\t,h}-Pg_{\t,h}\right)
\end{align*}
so that

\begin{align}
\P_{n}\left(g_{\t,\hat{h}}-g_{\t_{0},\hat{h}}\right) & =\frac{1}{\sqrt{n}}\GG_{n}\left(g_{\t,h_{0}}-g_{\t_{0},h_{0}}\right)\nonumber \\
 & +\frac{1}{\sqrt{n}}\GG_{n}\left(g_{\t,\hat{h}}-g_{\t_{0},\hat{h}}-g_{\t,h_{0}}+g_{\t_{0},h_{0}}\right)\nonumber \\
 & +P\left(g_{\t,\hat{h}}-g_{\t_{0},\hat{h}}\right)\label{eq:Bin_decom}
\end{align}
and we proceed to deal with the three terms on the right hand side
of \eqref{eq:Bin_decom} separately.

Lemma \ref{lem:Term1} below presents a maximal inequality about the
first term, and formalizes our previous discussion that the smoothness
of the function $g_{\t,h_{0}}$ with respect to $\t$ in a small neighborhood
of $\t_{0}$:
\begin{lem}
\label{lem:Term1} Under Assumption \ref{assu:Basic}, for some constant
$M_{1}>0$,
\begin{equation}
P\sup_{\norm{\t-\t_{0}}\leq\d}\left|\GG_{n}\left(g_{\t,h_{0}}-g_{\t_{0},h_{0}}\right)\right|\leq M_{1}\d^{\frac{3}{2}}.\label{eq:Max1}
\end{equation}
\end{lem}
\noindent The term $\d^{\frac{3}{2}}$ on the right hand side of \eqref{eq:Max1}
is in sharp contrast with, and much smaller than, the corresponding
term $\d^{\frac{1}{2}}$ under the usual setting with $n^{-1/3}$-asymptotics,
such as in \citet{kim1990cube} and \citet*{seo2018local}. In fact,
the smoothing by $h_{0}$ is so strong that $\d^{\frac{3}{2}}$ is
even of a smaller magnitude than $\d$, which corresponds to the standard
$n^{-1/2}$-asymptotics. This implies that, if we \emph{knew} the
true $h_{0}$, then any point estimator from $\arg\max_{\t\in\T}\P_{n}g_{\t,h_{0}}$
would actually converge to $\t_{0}$ at the $n$-rate. Such ``super-consistent''
rate would be reminiscent of the super-consistent least-square estimator
in change-point models \citet*[Section 14.5.1]{kosorok2007introduction,lee2008semiparametric,song2016asymptotics}.
Of course, since $h_{0}$ needs to be estimated in practice, we need
to account for the estimation error as captured by the remaining two
terms in \eqref{eq:Bin_decom}. As it turns out, the term $\d^{\frac{3}{2}}$
is negligible in comparison with those terms.

~

\noindent We now turn to the second term in \eqref{eq:Bin_decom},
which corresponds to the usual stochastic equicontinuity term in the
semiparametric estimation literature. We impose the following standard
smoothness condition on the functional space of $h_{0}$ and the sup-norm
convergence of the first-stage estimator $\hat{h}$. Specifically,
let ${\cal C}_{M}^{\left\lfloor d\right\rfloor +1}\left({\cal X}\right)$
denote a class of functions on ${\cal X}$ that possess uniformly
bounded derivatives up to order $\left\lfloor d\right\rfloor +1$.
\begin{assumption}
\label{assu:FirstStageRate}(i) $h_{0}\in{\cal H}\subseteq{\cal C}_{M}^{\left\lfloor d\right\rfloor +1}\left({\cal X}\right)$
(ii) $\hat{h}\in{\cal H}$ with probability approaching $1$ and (iii)
$\norm{\hat{h}-h_{0}}_{\infty}=O_{p}\left(a_{n}\right)$.
\end{assumption}
\noindent See, for example, \citet{hansen2008uniform}, \citet{belloni2015some}
and \citet*{chen2015optimal} for results on the sup-norm convergence
of kernel and sieve estimators. Lemma \ref{lem:Term2} below then
allows us to control the second term in \eqref{eq:Bin_decom}.
\begin{lem}
\label{lem:Term2} Under Assumptions \ref{assu:Basic}-\ref{assu:FirstStageRate}
with ${\cal H}:={\cal C}_{M}^{\left\lfloor d\right\rfloor +1}\left({\cal X}\right),$
for some constant $M_{2}>0$,
\begin{equation}
P\sup_{\t\in\T,h\in{\cal H}:\norm{\t-\t_{0}}\leq\d,\norm{h-h_{0}}_{\infty}\leq Ka_{n}}\left|\GG_{n}\left(g_{\t,h}-g_{\t_{0},h}-g_{\t,h_{0}}+g_{\t_{0},h_{0}}\right)\right|\leq M_{2}a_{n}\sqrt{\d}.\label{eq:Max2}
\end{equation}
\end{lem}
\noindent We note that the term $\sqrt{\d}$ due to the non-smoothness
of the indicator function now shows up on the right hand side of \eqref{eq:Max2}
, but it is weighted down by $a_{n}$, the sup-norm rate at which
$\hat{h}$ converges to $h_{0}$.

~

\noindent Lastly, we turn to the third term $P\left(g_{\t,\hat{h}}-g_{\t_{0},\hat{h}}\right)$
in \eqref{eq:Bin_decom}, which is a familiar term in the standard
asymptotic theory for semiparametric estimation. Usually\citep*{newey1994large,chen2003estimation}
such a term can be written into an asymptotically linear form based
on the functional derivative of $g_{\t,h}$ in $h$, contributing
an additional component to the asymptotic variance of the $n^{-1/2}$
asymptotically normal semiparametric estimator. However, this will
not be the case with our current TSMS estimator.

The behavior of the third term can be most clearly illustrated if
we take $\hat{h}$ to be the (adapted) Nadaraya-Watson kernel estimator
defined by 
\begin{align}
\hat{h}\left(x\right) & :=\frac{1}{p_{x}}\cd\frac{1}{nb_{n}^{d}}\sum_{i=1}^{n}\left(y_{i}-\frac{1}{2}\right)\phi_{d}\left(\frac{x-X_{i}}{b_{n}}\right)\label{eq:NW}
\end{align}
where $b_{n}$ is a (sequence of positive) bandwidth parameter shrinking
towards zero, $\phi_{d}$ is taken to be the standard $d$-dimensional
Gaussian PDF, and $p_{x}=\pi^{-d/2}{\displaystyle \Gamma\left(d/2+1\right)}$
is the reciprocal of the volume of the unit ball $\B^{d}$ (with $\text{\textgreek{G}}$
being the Gamma function), since the true density of $X$ is assumed
to be known and uniform on $\B^{d}$.\footnote{The density, if unknown, can be estimated by the standard kernel density
estimator $\hat{p}\left(x\right)=\frac{1}{nb_{n}^{d}}\sum_{i=1}^{n}\phi_{d}\left(\frac{x-X_{i}}{b_{n}}\right)$,
so that $\hat{h}\left(x\right)=\frac{1}{nb_{n}^{d}}\sum_{i=1}^{n}\left(y_{i}-\frac{1}{2}\right)\phi_{d}\left(\frac{x-X_{i}}{b_{n}}\right)\frac{1}{\hat{p}\left(x\right)}.$
We note that the additional density estimation does not change the
convergence rate of $\hat{h}$, so we leave it out for simpler notation.} In this case, 
\begin{align*}
Pg_{\t,\hat{h}} & =\int\hat{h}\left(x\right)\ind\left\{ x^{'}\t\geq0\right\} p_{x}dx\\
 & =\int\frac{1}{nb_{n}^{d}}\sum_{i=1}^{n}\left(y_{i}-\frac{1}{2}\right)\phi_{d}\left(\frac{x-X_{i}}{b_{n}}\right)\ind\left\{ x^{'}\t\geq0\right\} dx\\
 & =\frac{1}{n}\sum_{i=1}^{n}\left(y_{i}-\frac{1}{2}\right)\int\frac{1}{b_{n}^{d}}\ind\left\{ x^{'}\t\geq0\right\} \phi_{d}\left(\frac{x-X_{i}}{b_{n}}\right)dx\\
 & =\frac{1}{nb_{n}^{D}}\sum_{i=1}^{n}\left(y_{i}-\frac{1}{2}\right)\int\phi_{d}\left(u\right)\ind\left\{ \left(X_{i}+b_{n}u\right)^{'}\t\geq0\right\} b_{n}^{d}du\quad\text{with }u:=\frac{x-X_{i}}{b_{n}}\\
 & =\frac{1}{n}\sum_{i=1}^{n}\left(y_{i}-\frac{1}{2}\right)\int\ind\left\{ \left(X_{i}+b_{n}u\right)^{'}\t\geq0\right\} \phi_{d}\left(u\right)du\\
 & =\frac{1}{n}\sum_{i=1}^{n}\left(y_{i}-\frac{1}{2}\right)\int\ind\left\{ u^{'}\t\geq-\frac{X_{i}^{'}\t}{b_{n}}\right\} \phi_{d}\left(u\right)du\\
 & =\frac{1}{n}\sum_{i=1}^{n}\left(y_{i}-\frac{1}{2}\right)\P_{U}\left(U^{'}\t\geq-\frac{X_{i}^{'}\t}{b_{n}}\right)\text{ where }U\sim\cN\left({\bf 0},I_{d}\right)\\
 & =\frac{1}{n}\sum_{i=1}^{n}\left(y_{i}-\frac{1}{2}\right)\P_{\ol U}\left\{ \ol U\geq-\frac{X_{i}^{'}\t}{b_{n}}\right\} \text{ with }\ol U:=U^{'}\t\sim\cN\left(0,\t^{'}\t=1\right)\\
 & =\frac{1}{n}\sum_{i=1}^{n}\left(y_{i}-\frac{1}{2}\right)\left(1-\Phi\left(-\frac{X_{i}^{'}\t}{b_{n}}\right)\right)\\
 & =\frac{1}{n}\sum_{i=1}^{n}\left(y_{i}-\frac{1}{2}\right)\Phi\left(\frac{X_{i}^{'}\t}{b_{n}}\right)
\end{align*}
which is exactly the same as the sample criterion for the SMS estimator
in \eqref{eq:SMS}.

Notably, $Pg_{\t,\hat{h}}$ is now (twice) differentiable in $\t$,
allowing us to exploit the Taylor expansion of $Pg_{\t,\hat{h}}$
around the true parameter $\t_{0}$. Hence, the essence of the asymptotic
theory for the SMS estimator in \citet*{horowitz1992smoothed} applies.
Nevertheless, we formally present the following results, given that
we are working with different normalization and support assumptions
than those in \citet*{horowitz1992smoothed}.\footnote{\citet*{horowitz1992smoothed} normalizes $\left|\t_{1}\right|=1$
and assumes that the conditional distribution of $X_{i,1}$ given
any realization of $\left(X_{i,2},...,X_{i,d}\right)$ has everywhere
positive density on the real line. In contrast, we assume that $\t\in\S^{d-1}$
and $Supp\left(X_{i}\right)=\B^{d}$, and will work with differential
geometry on $\S^{d-1}$.} 

Formally, define $Z_{i}:=\left(y_{i},X_{i}\right)$ and $\psi_{b_{n},\t}\left(z\right):=\left(y-\frac{1}{2}\right)\Phi\left(x^{'}\t/b_{n}\right)$,
and consider the following decomposition:
\begin{align*}
P\left(g_{\t,\hat{h}}-g_{\t_{0},\hat{h}}\right) & =\P_{n}\left(\psi_{n,\t}-\psi_{n,\t_{0}}\right)=\frac{1}{\sqrt{n}}\GG_{n}\left(\psi_{n,\t}-\psi_{n,\t_{0}}\right)+P\left(\psi_{n,\t}-\psi_{n,\t_{0}}\right),
\end{align*}
the right hand side of which can be controlled via the following lemma,
which is very similar to \citet*[Lemma 5]{horowitz1992smoothed}.
\begin{lem}
\label{lem:Term_3}With $\hat{h}$ given by \eqref{eq:NW}, for some
positive constants $M_{3},M_{4},M_{5}$ and $C>0$:
\begin{itemize}
\item[(i)]  $P\sup_{\norm{\t-\t_{0}}\leq\d}\left|\GG_{n}\left(\psi_{n,\t}-\psi_{n,\t_{0}}\right)\right|\leq M_{3}b_{n}^{-1}\left(\d+b_{n}\right)^{\frac{1}{2}}\d.$
\item[(ii)]  Writing $\d:=\norm{\t-\t_{0}}$, 
\begin{align*}
P\left(\psi_{n,\t}-\psi_{n,\t_{0}}\right) & =-\left(\t-\t_{0}\right)^{'}V\left(\t-\t_{0}\right)+b_{n}^{2}A_{1}\left(\t-\t_{0}\right)\\
 & \quad+o\left(\d^{2}\right)+o\left(b_{n}^{2}\d\right)+O\left(b_{n}^{-1}\d^{3}\left(1+b_{n}^{-2}\d^{-2}\right)\right)\\
 & \leq-C\d^{2}+M_{4}b_{n}^{2}\d+M_{5}b_{n}^{-1}\d^{3}\left(1+b_{n}^{-2}\d^{-2}\right)
\end{align*}
where the inequality on the second line holds for sufficiently large
$n$ with some $A_{1}$ and some positive semi-definite matrix $V$
of rank $d-1$.
\end{itemize}
\end{lem}
\noindent Combining the results from Lemma \ref{lem:Term1}, \ref{lem:Term2}
and \ref{lem:Term_3}, we obtain the following theorem regarding the
convergence rate of the TSMS estimator.
\begin{thm}[Rate of Convergence]
\label{thm:Thm_Bin_Rate} With $\hat{h}$ given by the Nadaraya-Watson
estimator \eqref{eq:NW}, for any $b_{n}\to0$ and $nb_{n}^{d}/\log n\to\infty$,
\begin{equation}
\norm{\hat{\t}-\t_{0}}=O_{p}\left(\max\left\{ b_{n}^{2},\ \left(nb_{n}\right)^{-\frac{1}{2}},\ \left(n^{2}b_{n}^{d}/\log n\right)^{-\frac{1}{3}}\right\} \right).\label{eq:Bin_Rate_b}
\end{equation}
For $d<4$, with the optimal bandwidth choice $b_{n}\sim n^{-\frac{1}{5}}$,
\[
\norm{\hat{\t}-\t_{0}}=O_{p}\left(n^{-2/5}\right).
\]
For $4\leq d<6$, with the optimal (up to log factors) bandwidth choice
$b_{n}\sim n^{-\frac{2}{d+6}}$,
\[
\norm{\hat{\t}-\t_{0}}=O_{p}\left(n^{-\frac{4}{d+6}}\left(\log n\right)^{\frac{1}{3}}\right).
\]
For $d\geq6$, with the optimal (up to log factors) bandwidth choice
$b_{n}\sim\left(n/\log^{2}n\right)^{-\frac{1}{d}}$,
\[
\norm{\hat{\t}-\t_{0}}=O_{p}\left(n^{-\frac{2}{d}}\left(\log n\right)^{\frac{4}{d}}\right).
\]
\end{thm}
\noindent If the bandwidth is chosen to optimize the first-stage convergence
rate $a_{n}$, the final convergence rate for $\hat{\t}$ is characterized
by the following Corollary:
\begin{cor}
\label{cor:faster_than_1S}Let $a_{n}^{*}:=n^{-\frac{2}{d+4}}\sqrt{\log n}$
denote the optimal sup-norm convergence rate of $\hat{h}$ to $h$
(with respect to the first-stage estimation only). Then:
\begin{itemize}
\item[(i)]  With $b_{n}$ optimally chosen as in Theorem \ref{thm:Thm_Bin_Rate},
$\norm{\hat{\t}-\t_{0}}=o_{p}\left(a_{n}^{*}\right)$.
\item[(ii)]  With $b_{n}\sim n^{-\frac{1}{d+4}}$ so that $a_{n}=a_{n}^{*}$,
then $\norm{\hat{\t}-\t_{0}}=O_{p}\left(n^{-\frac{2}{d+4}}\right)$.
\end{itemize}
\end{cor}
\noindent First, we observe that the bias and variances induced by
$P\left(g_{\t,\hat{h}}-g_{\t_{0},\hat{h}}\right)$ are of order $b_{n}^{2}$
and $\left(nb_{n}\right)^{-1/2}$, which do not depend on the dimension
$d$ as in \citet*{horowitz1992smoothed}. Setting $b_{n}\sim n^{-1/5}$
balances these two terms, $b_{n}^{2}\sim\left(nb_{n}\right)^{-1/2}\sim n^{-2/5}$.
However, in our current setting, we also need $a_{n}$ to be sufficiently
small so as to control the disturbances induced by the first-stage
nonparametric estimation of $h$, whose sup-norm convergence rate
$a_{n}=\left(nb_{n}^{d}/\log n\right)^{-1/2}+b_{n}^{2}$ depends on
the dimension $d$. This leads to the last term $\left(n^{2}b_{n}^{d}\log n\right)^{-\frac{1}{3}}$
in \eqref{eq:Bin_Rate_b}, which in comparison is not required for
the SMS estimator. For $d<4$, this term is negligible with $b_{n}\sim n^{-\frac{1}{5}}$,
but for $d\geq4$ this term becomes pivotal. It turns out that for
$d\geq4$ but $d<6$, the optimal choice of $b_{n}\sim n^{-\frac{2}{d+6}}$
balances $b_{n}^{2}$ with $\left(n^{2}b_{n}^{d}\log n\right)^{-\frac{1}{3}}$
while guaranteeing that the sup-norm consistency of the first-stage
estimator
\[
\left(nb_{n}\right)^{-1/2}<<\norm{\hat{\t}-\t_{0}}\sim b_{n}^{2}<<a_{n}\sim\left(nb_{n}^{d}/\log n\right)^{-1/2}=o\left(1\right).
\]
In other words, the choice of $b_{n}\sim n^{-\frac{2}{d+6}}$ is ``over-smooth''
relative to the SMS optimal bandwidth, while being ``under-smooth''
relative to the optimal $d$-dimensional kernel regression bandwidth.
However, if $d\geq6$, then it is no longer possible to even balance
$b_{n}^{2}$ with $\left(n^{2}b_{n}^{d}\log n\right)^{-\frac{1}{3}}$,
so we minimize $b_{n}^{2}$ subject to the consistency constraint
that $a_{n}=\left(nb_{n}^{d}/\log n\right)^{-1/2}\to0$ by setting
$b_{n}$ to be slightly larger than $n^{-\frac{1}{d}}$. In this case,
the dominant term in $\norm{\hat{\t}-\t_{0}}$ is a deterministic
bias, while the disturbances are still of the order $\left(n^{2}b_{n}^{d}/\log n\right)^{-\frac{1}{3}}\sim\left(n\log n\right)^{-\frac{1}{3}}$.

Lastly, we note in Corollary \eqref{cor:faster_than_1S} that the
optimal rates are all \emph{strictly} faster than the optimal first-stage
convergence rate $a_{n}^{*}$.

~

\noindent We now turn to the asymptotic distribution of $\hat{\t}$,
which has phase transitions at $d=p+2=4$ and $d=3p=6$ (in our current
setting) given the discussion above.
\begin{thm}[Asymptotic Distribution]
\label{thm:Thm_Bin_Dist} There exist positive semi-definite matrix
$V$ and $\O$ that are invertible in the $\left(d-1\right)$-dimensional
tangent space of $\S^{d-1}$ at $\t_{0}$, as well as a constant vector
$A_{1}$ orthogonal to $\t_{0}$, such that:
\begin{itemize}
\item[(i)] If $d<4$ and $b_{n}\sim n^{-1/5}$, then $\hat{\t}$ is asymptotically
normal:
\begin{align}
n^{\frac{2}{5}}\left(I-\t_{0}\t_{0}^{'}\right)\left(\hat{\t}-\t_{0}\right) & \dto\cN\left(V^{-}A_{1},V^{-}\O V^{-}\right).\label{eq:Bin_Dist_s}
\end{align}
\item[(ii)] If $4\leq d<6$ and $b_{n}\sim n^{-\frac{2}{d+6}}$, then
\begin{align}
n^{\frac{4}{d+6}}\left(\log n\right)^{-\frac{1}{3}}\left(I-\t_{0}\t_{0}^{'}\right)\left(\hat{\t}-\t_{0}\right) & \dto\arg\max_{s\in\R^{d}:s^{'}\t_{0}=0}\left(G\left(s\right)+A_{1}^{'}s-\frac{1}{2}s^{'}Vs\right),\label{eq:Bin_Dist_m}
\end{align}
where $G$ is some $d$-dimensional zero-mean Gaussian process.
\item[(iii)] If $d\geq6$ and $b_{n}\sim\left(n/\log^{2}n\right)^{-\frac{1}{d}}$,
then
\begin{equation}
n^{\frac{2}{d}}\left(\log n\right)^{-\frac{4}{d}}\left(I-\t_{0}\t_{0}^{'}\right)\left(\hat{\t}-\t_{0}\right)\pto V^{-}A_{1}.\label{eq:Bin_Dist_l}
\end{equation}
\end{itemize}
\end{thm}
\noindent As expected, for small $d$ such that the $n^{-2/5}$ convergence
rate is attainable, the influence from the first-stage nonparametric
regression $\hat{h}$ is asymptotically negligible, making the TSMS
estimator asymptotically equivalent to the SMS estimator. The asymptotic
normality result in \eqref{eq:Bin_Dist_s} parallels the \citet*{horowitz1992smoothed}
result, but is stated through projection onto the tangent space of
the unit sphere at $\t_{0}$ (which is essentially $\R^{d-1}$ and
can be locally mapped back to the unit sphere).

For intermediate $4\leq d<6$, the disturbances from the first-stage
estimation of $h_{0}$ kick in, leading to asymptotic randomness in
the form of a Gaussian process. Such disturbances, corresponding to
the term of order $a_{n}\sqrt{\d_{n}}$ in Lemma \ref{lem:Term2},
are the joint product of the first-stage estimation error (of order
$a_{n}$) and the discreteness of the indicator function (or the order
$\sqrt{\d_{n}}$). The magnitude of randomness in the final Gaussian
process $G\left(s\right)$ induced by this term is balanced with the
asymptotic bias $A_{1}^{'}s$ produced by the (optimally chosen level
of) kernel smoothing, both of which survive in the final asymptotic
distribution along with usual quadratic identifying information $\left(-\frac{1}{2}s^{'}Vs\right)$.

In the standard asymptotic theory for $n^{-1/2}$-normal semiparametric
estimators (e.g. \citealp{newey1994large}, and \citealp*{chen2003estimation}),
this term will generally be negligible under the standard version
of stochastic equicontinuity conditions. Moreover, the term $P\left(g_{\t,\hat{h}}-g_{\t,h_{0}}\right)$
can usually be linearized based on its functional derivative with
respect to $h_{0}$ and shown (or assumed) to be $n^{-1/2}$-normal
(Theorem 8.1 in \citealp*{newey1994large}, and Condition 2.6 in \citealp{chen2003estimation})
under the assumption of $a_{n}=o_{p}\left(n^{-1/4}\right)$. In comparison,
we note that in our current setting such $n^{-1/2}$-normality is
unattainable.

On the other hand, the corresponding term in the local cubic-root
asymptotics considered in \citet*{seo2018local} is of the order $\sqrt{a_{n}\d}$,
which is larger than our $a_{n}\sqrt{\d}$ term. Hence, \citet*{seo2018local}
obtain convergence rates generally slower than $n^{-1/3}$ due to
the additional lack of smoothness with respect to the nonparametric
function $h$. The example considered in \citet*{delsol2020semiparametric}
about missing data does not feature non-smoothness with respect to
$h$, but the function $h$ does not serve a ``smoothing role''
on the indicator function involving the finite-dimensional parameter
of interest, thus still achieving an $n^{-1/3}$ convergence rate.
Correspondingly, the asymptotic distributions obtained in their settings
take the form of $\arg\max_{s}G\left(s\right)-s^{'}Vs$, where the
Gaussian noise dominates all other errors or biases. 

In summary, our setting features a pivotal interplay between the smoothing
of $h_{0}$ and the finite estimation error of $h_{0}$, leading to
a partially accelerated rate between $n^{-1/2}$ and $n^{-1/3}$,
and an asymptotic distribution that features both the usual Gaussian
noise component and a bias component.

Finally, for $d\geq6$, the bias actually becomes the dominant term,
resulting in a degenerate asymptotic distribution. In principle, if
we further symmetrize around the asymptotic bias, the disturbances
of the induced mean-zero process would be of the order $n^{-\frac{1}{3}}a_{n}^{\frac{2}{3}}\sim\left(n\log n\right)^{-\frac{1}{3}}$,
or roughly the cubic-root rate. 

~

Of course, in the above we used the Gaussian density kernel as an
illustration. We now explain how the rate of convergence can be improved
if smoothness conditions of order $s$ are imposed along with the
adoption of an order-$s$ kernel.

Clearly, Lemma \ref{lem:Term1} and Lemma \ref{lem:Term2} do not
depend on the specific form of kernels (or nonparametric estimators)
used, so they remain completely unchanged. However, Lemma \ref{lem:Term_3},
which is about the term $Pg_{\t,\hat{h}}-Pg_{\t_{0},\hat{h}},$ would
need to be adapted. Such an adaption is particularly simple if we
take the kernel function to be spherically (radially) symmetric.

We summarize the conditions we impose on the choice of kernel functions
in the following assumption.
\begin{assumption}
\label{assu:HighKernel}Let $K_{d}\left(u\right)\equiv K_{d}\left(\norm u\right)$
be a spherically symmetric kernel function of an even order $s\geq4$,
which satisfies: 
\begin{itemize}
\item (i) $K_{d}$ is uniformly obunded, twice continuously differentiable,
has uniformly bounded first and second derivatives, and vanishes outside
a compact set in $\R^{d}$.
\item (ii) $\int K_{d}\left(u\right)du=1$.
\item (iii)$\int u_{j}^{k}K_{d}\left(u\right)du=0,$ $\forall j$, and $\forall k\in\mathbb{N}$
s.t. $k\leq s-1$.
\item (iv) $R_{s}:=\int u_{j}^{s}K_{d}\left(u\right)du>0,$ $\forall j$.
\end{itemize}
\end{assumption}
Then, based on the Nadaraya-Watson firs stage
\begin{align*}
\hat{h}\left(x\right) & :=\frac{1}{p_{x}}\cd\frac{1}{nb_{n}^{d}}\sum_{i=1}^{n}\left(y_{i}-\frac{1}{2}\right)K_{d}\left(\frac{x-X_{i}}{b_{n}}\right),
\end{align*}
we can write
\begin{align}
Pg_{\t,\hat{h}} & =\int\hat{h}\left(x\right)\ind\left\{ x^{'}\t\geq0\right\} p_{x}dx\nonumber \\
 & =\int\frac{1}{nb_{n}^{d}}\sum_{i=1}^{n}\left(y_{i}-\frac{1}{2}\right)K_{d}\left(\frac{x-X_{i}}{b_{n}}\right)\ind\left\{ x^{'}\t\geq0\right\} dx\nonumber \\
 & =\frac{1}{n}\sum_{i=1}^{n}\left(y_{i}-\frac{1}{2}\right)\int\frac{1}{b_{n}^{d}}\ind\left\{ x^{'}\t\geq0\right\} K_{d}\left(\frac{x-X_{i}}{b_{n}}\right)dx\nonumber \\
 & =\frac{1}{nb_{n}^{d}}\sum_{i=1}^{n}\left(y_{i}-\frac{1}{2}\right)\int K_{d}\left(u\right)\ind\left\{ \left(X_{i}+b_{n}u\right)^{'}\t\geq0\right\} b_{n}^{d}du\quad\text{with }u:=\frac{x-X_{i}}{b_{n}}\nonumber \\
 & =\frac{1}{n}\sum_{i=1}^{n}\left(y_{i}-\frac{1}{2}\right)\int\ind\left\{ \left(X_{i}+b_{n}u\right)^{'}\t\geq0\right\} K_{d}\left(u\right)du\nonumber \\
 & =\frac{1}{n}\sum_{i=1}^{n}\left(y_{i}-\frac{1}{2}\right)\int\ind\left\{ u^{'}\t\geq-\frac{X_{i}^{'}\t}{b_{n}}\right\} K_{d}\left(u\right)du\nonumber \\
 & =\frac{1}{n}\sum_{i=1}^{n}\left(y_{i}-\frac{1}{2}\right)\int\ind\left\{ u_{1}\geq-\frac{X_{i}^{'}\t}{b_{n}}\right\} K_{d}\left(u\right)du\text{ by spherical symmetry of }K_{d}\nonumber \\
 & =\frac{1}{n}\sum_{i=1}^{n}\left(y_{i}-\frac{1}{2}\right)\int\ind\left\{ u_{1}\leq\frac{X_{i}^{'}\t}{b_{n}}\right\} K_{d}\left(u\right)du\text{ by evenness of }K_{d}\nonumber \\
 & =\frac{1}{n}\sum_{i=1}^{n}\left(y_{i}-\frac{1}{2}\right)\Lambda\left(\frac{X_{i}^{'}\t}{b_{n}}\right)\label{eq:HO_Horowitz}
\end{align}
with
\begin{equation}
\L\left(t\right):=\int\ind\left\{ u_{1}\leq t\right\} K_{d}\left(u\right)du.\label{eq:Kernel_Lambda}
\end{equation}

Clearly, \eqref{eq:HO_Horowitz} coincides with definitional formula
of Horowitz's SMS estimator. We now show via the following lemma that,
under Assumption \ref{assu:HighKernel}, the one-dimensional ``CDF-type''
function $\L\left(t\right)$ defined above satisfy the ``higher-order
kernel'' conditions in \citet*{horowitz1992smoothed}.
\begin{lem}
\label{lem:Kernel_1D}Under Assumption \ref{assu:HighKernel}, $\L$
defined in \ref{eq:Kernel_Lambda} satisfies the following conditions:
\begin{itemize}
\item (i) $\L$ is uniformly bounded, twice differentiable, has uniformly
bounded first and second derivatives, and vanishes outside a compact
set in $\R$.
\item (ii) $\lim_{t\to-\infty}\L\left(t\right)=0$ and $\lim_{t\to-\infty}\L\left(t\right)=1$.
\item (iii) Defining $\l\left(t\right):=\frac{d}{dt}\L\left(t\right)$,
\begin{align*}
\int_{-\infty}^{\infty}t^{j}\l\left(t\right)dt=0,\ \forall j\leq s-1,\quad\int_{-\infty}^{\infty}t^{s}\l\left(t\right)dt=R_{s}>0.
\end{align*}
\end{itemize}
\end{lem}
Hence, the results in \citet*{horowitz1992smoothed}, as well as generalizations
of Theorem \ref{thm:Thm_Bin_Rate}, apply. Specifically, the convergence
rate of $\hat{\t}$ would be given by
\[
\norm{\hat{\t}-\t_{0}}=O_{p}\left(\max\left\{ b_{n}^{s},\ \left(nb_{n}\right)^{-\frac{1}{2}},\ \left(n^{2}b_{n}^{d}\log n\right)^{\frac{1}{3}}\right\} \right),
\]
corresponding to an optimal rate of
\begin{align*}
\norm{\hat{\t}-\t_{0}} & \sim\begin{cases}
n^{-\frac{s}{2s+1}}, & \text{for }d<s+2,\\
n^{-\frac{2s}{3s+d}}\left(\log n\right)^{\frac{1}{3}}, & \text{for }s+2\leq d<3s,\\
n^{-\frac{s}{d}}\left(\log n\right)^{\frac{2s}{d}} & \text{for }d\geq3s.
\end{cases}
\end{align*}
Furthermore, the asymptotic normality of $\hat{\t}$ can be established
accordingly when $d<s+2$.
\begin{thm}
If $d<s+2$ and $b_{n}\sim n^{-\frac{1}{2s+1}}$, then
\begin{align*}
n^{\frac{s}{2s+1}}\left(I-\t_{0}\t_{0}^{'}\right)\left(\hat{\t}-\t_{0}\right) & \dto\cN\left(V^{-}A_{s},cV^{-}\O V^{-}\right)
\end{align*}
for some constant $c>0$.
\end{thm}

\section{\label{sec:MultiIndex}TSMS for Multi-Index Single-Crossing Models}

We now turn to the more general setting of multi-index single-crossing
models, where the TSMS estimator naturally arises while there are
no natural analogs of the MS and SMS estimators.

Let $\left(y_{i},{\bf X}_{i}\right)_{i=1}^{n}$ be a random sample
of data with ${\cal X}:=Supp\left({\bf X}_{i}\right)\subseteq\R^{J\times d}$
and the dimension of $y_{i}$ unrestricted. Let $h_{0}:{\cal X}\to\R$
be an unknown function that is directly identified from data. Usually
$h_{0}\left(x\right)$ is defined via a known functional of the conditional
distribution of $y_{i}$ given ${\bf X}_{i}=x$, e.g. $h_{0}\left(x\right)=\E\left[\rest{y_{i}}X_{i}=x\right]-\frac{1}{2}$
in the binary choice model above. Let $\t_{0}\in\T\subseteq\R^{d}$
be an unknown finite-dimensional parameter of interest, which is related
to $h_{0}$ via the following assumption.
\begin{assumption}[Multivariate Single-Crossing Conditions]
\label{assu:Assum_Mono} For any $x=\left(x_{1},...,x_{J}\right)^{'}\in\R^{J\times d}$,
\begin{align}
x_{j}^{'}\t_{0}>0\ \forall j=1,...,J\quad & \imp\quad h_{0}\left(x\right)>0,\nonumber \\
x_{j}^{'}\t_{0}=0\ \forall j=1,...,J\quad & \imp\quad h_{0}\left(x\right)=0,\label{eq:Mono_Zero}\\
x_{j}^{'}\t_{0}<0\ \forall j=1,...,J\quad & \imp\quad h_{0}\left(x\right)<0.\nonumber 
\end{align}
\end{assumption}
\noindent Again we normalize $\t_{0}\in\S^{d-1}$, as \ref{assu:Assum_Mono}
imposes no restriction on the scale of $\t_{0}$.

Based on Assumption \ref{assu:Assum_Mono}, we may define the following
population and sample criterion functions $Q$, $Q_{n}$ by
\begin{align}
Q\left(\t\right) & :=Pg_{\t,h_{0}},\label{eq:MMI_Q}\\
Q_{n}\left(\t\right) & :=\P_{n}g_{\t,\hat{h}},\label{eq:MMI_Qn}
\end{align}
where $\hat{h}$ is again some first-stage nonparametric estimator
of $h_{0}$, and
\begin{align*}
g_{\t,h}\left(x\right) & :=g_{+,\t,h}\left(x\right)+g_{-,\t,h}\left(x\right)\\
g_{+,\t,h}\left(x\right) & :=\left[h\left(x\right)\right]_{+}\l\left(x,\t_{0}\right)\\
g_{-,\t,h}\left(x\right) & :=\left[-h\left(x\right)\right]_{+}\l\left(-x,\t_{0}\right)\\
\l\left(x,\t\right) & :=-\prod_{j=1}^{J}\ind\left\{ {\bf X}_{ij}^{'}\t>0\right\} ,
\end{align*}
with
\[
\left[t\right]_{+}:=\max\left(t,0\right)
\]
denoting the positive part function. The TSMS estimator is again given
by
\[
\hat{\t}:=\arg\max_{\t\in\S^{d-1}}Q_{n}\left(\t\right).
\]
We can then extend our analysis of the asymptotic theory for the TSMS
estimator in the binary choice setting to the current multi-index
setting.

In the following, it would often be convenient to work with the vectorization
$\text{vec}\left({\bf X}_{i}\right)$ of the matrix random variable
${\bf X}_{i}$ in $\R^{J\times d}$.
\begin{assumption}[Regularity Conditions]
\label{assu:Assum_PID} ~
\begin{itemize}
\item[(i)]  ${\bf 0}\in\R^{Jd}$ is an interior point of $\text{vec}\left({\cal X}\right)$,
and ${\cal X}$ is a convex and compact subset of $\R^{J\times d}$.
\item[(ii)]  The probability density function $p\left(X\right)$ of ${\bf X}_{i}$
is uniformly bounded and also uniformly bounded away from zero on
${\cal X}$.
\item[(iii)]  $h_{0}\left(x\right)$ is twice continuously differentiable in $\text{vec}\left(x\right)\in\R^{Jd}$
with uniformly bounded first and second derivatives.
\item[(iv)]  $\Dif_{\text{vec}\left(x\right)}h_{0}\left(x\right)^{'}\left(\ind_{J}\otimes\t_{0}\right)>0$.
\end{itemize}
\end{assumption}
We first explain the intuition why and how Lemma \ref{lem:Term1}
generalizes to multi-index settings. At any given $x=\left(x_{1},...,x_{J}\right)$,
notice that
\[
g_{+,\t_{0},h_{0}}\left(x\right)=\left[h_{0}\left(x\right)\right]_{+}\prod_{j=1}^{J}\ind\left\{ x_{j}^{'}\t_{0}<0\right\} =0
\]
and hence, for $\t$ very close to $\t_{0}$, we have
\begin{align*}
\left|g_{+,\t,h_{0}}\left(x\right)-g_{+,\t_{0},h_{0}}\left(x\right)\right|= & \left[h_{0}\left(x\right)\right]_{+}\prod_{j=1}^{J}\ind\left\{ x_{j}^{'}\t<0\right\} 
\end{align*}
which is nonzero only if $h_{0}\left(x\right)>0$ and $x_{j}^{'}\t<0$
for all $j\in J$. For the event
\[
\prod_{j=1}^{J}\ind\left\{ x_{j}^{'}\t_{0}<0\right\} =0\quad\text{but }\prod_{j=1}^{J}\ind\left\{ x_{j}^{'}\t<0\right\} =1,
\]
to occur, generically one and only one\footnote{Here we only consider this generic case for notational simplicity.
See the formal proof of Lemma \ref{lem:Term1}' in the Appendix for
how we deal with more than one sign changes in the $J$ indexes. } of the $J$ inequalities switch sign from $\t_{0}$ to $\t$, in
which case there exists a unique $j^{*}$ such that
\[
x_{j}^{'}\t<0,\quad\forall j,
\]
but
\[
x_{j^{*}}^{'}\t_{0}>0\quad\text{and}\quad x_{k}^{'}\t_{0}<0,\quad\forall k\neq j^{*}.
\]
Hence, we have 
\begin{align*}
x_{j^{*}}^{'}\t_{0}>0>x_{j}^{'}\t & =x_{j^{*}}^{'}\t_{0}+x_{j}^{'}\left(\t-\t_{0}\right)>x_{j^{*}}^{'}\t_{0}-M\norm{\t-\t_{0}}
\end{align*}
and thus
\[
0<x_{j^{*}}^{'}\t_{0}<M\norm{\t-\t_{0}}.
\]
Now, let $\ol x_{j^{*}}:=x_{j^{*}}-M\norm{\t-\t_{0}}\t_{0}^{'}$ and
$\ol x_{k}:=x_{k}$ for all $k\neq j^{*}$, then we know
\[
\ol x_{j^{*}}^{'}\t_{0}=x_{j^{*}}^{'}\t_{0}-M\norm{\t-\t_{0}}<0,\quad\text{and}\quad\ol x_{k}^{'}\t_{0}=x^{'}\t_{0}<0
\]
and hence, by the single-crossing condition \eqref{eq:Mono_Zero}
\[
h_{0}\left(\ol x\right)<0.
\]
However, we also know that
\[
h_{0}\left(x\right)>0.
\]
Now, since $\ol x$ is close to $x$ by construction and $h_{0}$
is smooth in $x$, the above is only possible when $h_{0}\left(x\right)$
is close to $0$. Formally, we have
\begin{align*}
h_{0}\left(x\right)>0>h_{0}\left(\ol x\right) & =h_{0}\left(x\right)+\Dif h_{0}\left(\tilde{x}\right)\left(\ol x-x\right)\\
 & >h_{0}\left(x\right)-\sup_{\tilde{x}}\left|\Dif h_{0}\left(\tilde{x}\right)\right|\cd\norm{\ol x-x}\\
 & =h_{0}\left(x\right)-\sup_{\tilde{x}}\left|\Dif h_{0}\left(\tilde{x}\right)\right|\cd\norm{M\norm{\t-\t_{0}}\t_{0}^{'}}\\
 & =h_{0}\left(x\right)-CM\cd\norm{\t-\t_{0}}
\end{align*}
and thus
\[
0<h_{0}\left(x\right)<CM\cd\norm{\t-\t_{0}}=O\left(\norm{\t-\t_{0}}\right).
\]

This explains the key intuition why the smoothing effect of $h_{0}$
remains intact under the multi-index setup. In fact, Lemma \ref{lem:Term1}
and Lemma \ref{lem:Term2} generalize without any change to the multi-index
single-crossing model.

\noindent \textbf{Lemma 1'}\customlabel{lemTerm1prime}{$1'$}\emph{For
some constant $M_{1}>0$,
\[
P\sup_{\norm{\t-\t_{0}}\leq\d}\left|\GG_{n}\left(g_{\t,h_{0}}-g_{\t_{0},h_{0}}\right)\right|\leq M_{1}\d^{\frac{3}{2}}.
\]
}

\noindent 
\noindent \textbf{Lemma 2'}\customlabel{lemTerm2prime}{$2'$}\textbf{\emph{
}}\emph{For some constant $M_{2}>0$,}
\[
P\sup_{\t\in\T,h\in{\cal H}:\norm{\t-\t_{0}}\leq\d,\norm{h-h_{0}}_{\infty}\leq Ka_{n}}\left|\GG_{n}\left(g_{\t,h}-g_{\t_{0},h}-g_{\t,h_{0}}+g_{\t_{0},h_{0}}\right)\right|\leq M_{2}a_{n}\sqrt{\d}.
\]

\begin{lem}
\label{lem:LemTerm3Quad}$Pg_{\t,h}$ is twice continuously differentiable
in $\t$ with 
\[
\Dif_{\t}Pg_{\t_{0},h_{0}}={\bf 0},\quad\Dif_{\t\t}Pg_{\t_{0},h_{0}}=-V,
\]
for some positive semi-definite matrix $V$ of rank $d-1$.
\end{lem}
Then
\begin{align*}
P\left(g_{\t,\hat{h}}-g_{\t_{0},\hat{h}}\right) & =\left(\Dif_{\t}Pg_{\t_{0},\hat{h}}\right)^{'}\left(\t-\t_{0}\right)+\frac{1}{2}\left(\t-\t_{0}\right)^{'}\left(\Dif_{\t\t}Pg_{\t_{0},\hat{h}}\right)\left(\t-\t_{0}\right)+o\left(\norm{\t-\t_{0}}^{2}\right)\\
 & =\left(\Dif_{\t}Pg_{\t_{0},\hat{h}}\right)^{'}\left(\t-\t_{0}\right)-\frac{1}{2}\left(\t-\t_{0}\right)^{'}V\left(\t-\t_{0}\right)\\
 & \quad+\frac{1}{2}\left(\t-\t_{0}\right)^{'}\left(\Dif_{\t\t}Pg_{\t_{0},\hat{h}}-\Dif_{\t\t}Pg_{\t_{0},h_{0}}\right)\left(\t-\t_{0}\right)+o\left(\norm{\t-\t_{0}}^{2}\right)
\end{align*}

\begin{lem}[General Bound on the Rate of Convergence]
\label{thm:MMI_Rate} Under Assumptions \ref{assu:Assum_Mono}-\ref{assu:Assum_PID},
\[
\norm{\hat{\t}-\t_{0}}=O_{p}\left(a_{n}\right).
\]
\end{lem}
~

\noindent To obtain sharper bounds on the rate of convergence, we
need to analyze the term $P\left(g_{\t,\hat{h}}-g_{\t_{0},\hat{h}}\right)$
more closely.
\begin{thm}
\label{thm:MMI_3rates}Suppose Assumptions \ref{assu:Assum_Mono}-\ref{assu:Assum_PID}
hold and furthermore
\begin{align*}
P\left(g_{\t,\hat{h}}-g_{\t_{0},\hat{h}}\right) & =u_{n}A\left(\t-\t_{0}\right)+v_{n}W_{n}\left(\t-\t_{0}\right)-\left(\t-\t_{0}\right)^{'}V\left(\t-\t_{0}\right)+o_{p}\left(u_{n}\d+v_{n}\d+\d^{2}\right)
\end{align*}
with $A$ and $V$ being constant vector and matrix, $W_{n}=O_{p}\left(1\right)$,
and $u_{n},v_{n}=o\left(1\right)$. Then: 
\[
\norm{\hat{\t}-\t_{0}}=\max\left\{ n^{-\frac{1}{3}}a_{n}^{\frac{2}{3}},\ u_{n},\ v_{n}\right\} .
\]
\end{thm}

\section{Simulation}

In this section, we evaluate the finite-sample performance of our
TSMS estimator through a Monte Carlo Simulation. We derive our estimator
based on the criterion funtion \eqref{eq:Q0_MI}.

\subsection{Single-Index Setting}

We first consider the standard binary choice model
\[
y=\ind\left\{ X^{'}\t_{0}\geq\e\right\} ,
\]
which falls under the single-index setting. We set the dimension of
covariates $d=3$ and the true parameter $\theta_{0}=\left[\frac{1}{\sqrt{3}},\frac{1}{\sqrt{3}},\frac{1}{\sqrt{3}}\right]^{'}$.
Each of our covariates $X$ is drawn independently from a uniform
distribution on $[-5,5]$. We compare different first-stage estimators;
in particular, we estimate the true function $h_{0}$ using Gaussian
kernel, probit model, and OLS. When implementing Gaussian kernel,
we use a 5-fold cross-validation to tune the bandwidth parameter.
In addition, we include a benchmark case where the true $h_{0}$ is
used without first-stage estimation. After creating the first-step
estimator $\hat{h}$ of $h_{0}$, we construct our estimator $\hat{\theta}$
based on the adaptive-grid search algorithm developed in \citet{gao2019robust},
which searches for the optimizer of \eqref{eq:Q0_MI} on the unit
sphere. Moreover, we test the performance of all the estimators under
different error distributions for $\e$: (i) standard normal distribution
$\cN(0,1)$, and (ii) a de-medianed version of the log-normal distribution
$\log\cN(0,1)$. We also vary the sample size $n$ to investigate
the convergence rate of our estimator. 

\begin{table}
\centering{}\caption{\label{tab:EstErr_BCM}Estimation error in binary choice model}
\emph{}%
\begin{tabular}{cccccc}
\toprule 
 &  & \multicolumn{4}{c}{RMSE}\tabularnewline
\cmidrule{3-6} \cmidrule{4-6} \cmidrule{5-6} \cmidrule{6-6} 
Error & $n$ & True & Kernel & Probit & OLS\tabularnewline
\midrule 
\multirow{3}{*}{Gaussian} & 100 & 0.0336 & 0.1261 & 0.0871 & 0.1309\tabularnewline
\cmidrule{2-6} \cmidrule{3-6} \cmidrule{4-6} \cmidrule{5-6} \cmidrule{6-6} 
 & 500 & 0.0068 & 0.0538 & 0.0350 & 0.0544\tabularnewline
\cmidrule{2-6} \cmidrule{3-6} \cmidrule{4-6} \cmidrule{5-6} \cmidrule{6-6} 
 & 1000 & 0.0038 & 0.0299 & 0.0233 & 0.0363\tabularnewline
\midrule 
\multirow{3}{*}{Log-normal} & 100 & 0.0336 & 0.1470 & 0.1242 & 0.1416\tabularnewline
\cmidrule{2-6} \cmidrule{3-6} \cmidrule{4-6} \cmidrule{5-6} \cmidrule{6-6} 
 & 500 & 0.0069 & 0.0596 & 0.0587 & 0.0684\tabularnewline
\cmidrule{2-6} \cmidrule{3-6} \cmidrule{4-6} \cmidrule{5-6} \cmidrule{6-6} 
 & 1000 & 0.0038 & 0.0402 & 0.0433 & 0.0481\tabularnewline
\bottomrule
\end{tabular}
\end{table}

Table \ref{tab:EstErr_BCM} presents the root mean squared errors
(RMSE) of all the estimators given different error distributions and
sample sizes. Our TSMS estimator that implements Gaussian kernel in
the first stage converges fast and is robust against different error
distributions in our simulation. Under Gaussian noises, it is not
surprising to find that using probit model in the first stage leads
to the best performance. Indeed, the probit model matches the parametric
form of $h_{0}$ under Gaussian noises and hence achieves parametric
rate of convergence in the first stage. However, parametric methods
such as probit and OLS rely on correct specification, and are not
as robust as a nonparametric first stage. In particular, our Gaussian
kernel first stage outperforms both probit and OLS with a moderate
number of samples when the errors are drawn from a log-normal distribution.
It is also worth noting that, despite potential biases caused by misspecification,
using a parametric first stage such as probit and OLS can still produce
decent second-stage estimates. 

\subsection{Multi-index Model}

Next, we analyze our TSMS estimator in a multi-index setting by .
We set $J=2$, and consider the model
\[
y=\ind\left\{ X_{1}^{'}\theta_{0}>\epsilon_{1}\right\} \cd\ind\left\{ X_{2}^{'}\theta_{0}>\epsilon_{2}\right\} 
\]
where $X_{1}$ and $X_{2}$ are both three dimensional random variables,
while $\epsilon_{1i}$ and $\e_{i2}$ are independently drawn from
some chosen distributions. We take $h_{0}(x)=\E\left[y_{i}-\frac{1}{4}|X_{i}=x\right]$,
which can be verified to satisfy Assumption \ref{assu:Assum_Mono}
and also easily computed under the design that $\epsilon_{1}$ is
drawn independently from $\e_{2}$. Then, we run Gaussian kernel,
probit and OLS to estimate $h_{0}$ using the 6 covariates all together,
and then obtain the second-stage estimator of $\t_{0}$.

\begin{table}
\centering{}\caption{\label{tab:EstErr_MI}Estimation error in multi-index model}
\emph{}%
\begin{tabular}{cccccc}
\toprule 
 &  & \multicolumn{4}{c}{RMSE}\tabularnewline
\cmidrule{3-6} \cmidrule{4-6} \cmidrule{5-6} \cmidrule{6-6} 
Error & $n$ & True & Kernel & Probit & OLS\tabularnewline
\midrule 
\multirow{3}{*}{Gaussian} & 100 & 0.0874 & 0.2118 & 0.2392 & 0.2309\tabularnewline
\cmidrule{2-6} \cmidrule{3-6} \cmidrule{4-6} \cmidrule{5-6} \cmidrule{6-6} 
 & 500 & 0.0270 & 0.0875 & 0.1321 & 0.0982\tabularnewline
\cmidrule{2-6} \cmidrule{3-6} \cmidrule{4-6} \cmidrule{5-6} \cmidrule{6-6} 
 & 1000 & 0.0212 & 0.0622 & 0.0968 & 0.0727\tabularnewline
\midrule 
\multirow{3}{*}{Log-normal} & 100 & 0.0950 & 0.1990 & 0.2310 & 0.2422\tabularnewline
\cmidrule{2-6} \cmidrule{3-6} \cmidrule{4-6} \cmidrule{5-6} \cmidrule{6-6} 
 & 500 & 0.0350 & 0.0978 & 0.1216 & 0.1178\tabularnewline
\cmidrule{2-6} \cmidrule{3-6} \cmidrule{4-6} \cmidrule{5-6} \cmidrule{6-6} 
 & 1000 & 0.0288 & 0.0734 & 0.0907 & 0.0844\tabularnewline
\bottomrule
\end{tabular}
\end{table}

Table \ref{tab:EstErr_MI} lists the RMSE for all our estimators of
$\t_{0}$ with varying error distributions and sample sizes. Overall,
the TSMS estimator still performs well in finite sample under the
multi-index setting. Most notably in comparison with the single-index
setting (Table \ref{tab:EstErr_BCM}), our estimator with the kernel
first stage now outperforms the one based on the probit first stage
for any sample size and error distribution, since the probit model
is now misspecified under the multi-index setting. 

\section{Conclusion}

\noindent This paper considers the asymptotic theory of the TSMS estimator
that is applicable in semiparametric models that a general form of
monotonicity in one or several parametric indexes. We show that the
first-stage nonparametric estimator effectively serves as an imperfect
smoothing function on a non-smooth criterion function, leading to
the pivotality of the first-stage estimation error with respect to
the second-stage convergence rate and asymptotic distribution.

The current analysis is mostly focused on a kernel first-stage regression,
but it would be interesting and informative to replicate the analysis
with a sieve first stage, say, based on the general results obtained
in \citet{belloni2015some} and \citet*{chen2015optimal}. Moreover,
a full-fledged distribution theory and inferential procedure that
fully accommodates the dimension $d$, the smoothness $s$, and various
kernel/sieve first-stage estimators still require considerable work
to be developed.

\bibliographystyle{ecta}
\bibliography{TSMS}

\appendix

\section*{Appendix}

\section{\label{sec:App_Proof}Proofs}

\subsection{\label{subsec:pf_Entropy}Lemmas on Entropy Integrals}

Define ${\cal G}:=\left\{ g_{\t,h}-g_{\t_{0},h}:\ \t\in\T,h\in{\cal H}\right\} $,
which is uniformly bounded since ${\cal H}$ is uniformly bounded.
We first establish the finiteness of the following uniform entropy
integral.
\begin{lem}
\label{lem:J_uni_finite}$J:=\sup_{Q}\int_{0}^{1}\sqrt{\log\mathscr{N}\left(\e,\cG,L_{2}\left(Q\right)\right)}d\e<\infty.$
\end{lem}
\begin{proof}
The collection of indicators for half spaces $\ind\left\{ x_{j}^{'}\t\geq0\right\} $
across $\t\in\S^{d-1}$ is a VC-subgraph class of functions with VC
dimension $d+2$, so by VW Lemma 2.6.18,
\begin{align*}
 & \left\{ \prod_{j\in J}\ind\left\{ x^{'}\t\geq0\right\} -\prod_{j\in J}\ind\left\{ x^{'}\t_{0}\geq0\right\} :\t\in\T\right\} \\
= & \left\{ \bigwedge_{j\in J}\ind\left\{ x^{'}\t\geq0\right\} -\bigwedge_{j\in J}\ind\left\{ x^{'}\t_{0}\geq0\right\} :\t\in\T\right\} 
\end{align*}
is also VC-subgraph class, which thus have bounded uniform entropy
integrals. Moreover, since ${\cal H}\subseteq{\cal C}_{M}^{\left\lfloor d/2\right\rfloor +1}\left({\cal X}\right)$,
we know by VW Theorem 2.7.1 that $\log\mathscr{N}\left(\d,{\cal H},\norm{\cd}_{\infty}\right)\leq C\d^{-d/\left(\left\lfloor d\right\rfloor +1\right)}$
and thus also have bounded uniform entropy integrals 
\[
\int_{0}^{1}\sup_{Q}\sqrt{1+\log\mathscr{N}\left(\e,\cG_{2},L_{2}\left(Q\right)\right)}d\e<\infty.
\]
By \citet*{kosorok2007introduction} Theorem 9.15, we deduce $\cG$
also has uniformly bounded entropy integral.
\end{proof}
Alternatively, we could follow \citet*{chen2003estimation} and work
with the following bracketing integral.
\begin{lem}
\label{lem:J_br_finite}$J_{[]}:=\int_{0}^{1}\sqrt{1+\log\mathscr{N}_{[]}\left(\e,\cG,L_{2}\left(P\right)\right)}d\e<\infty.$
\end{lem}
\begin{proof}
Since ${\cal H}\subseteq{\cal C}_{M}^{\left\lfloor d/2\right\rfloor +1}\left({\cal X}\right)$,
we know by VW Theorem 2.7.1 that $\log\mathscr{N}\left(\d,{\cal H},\norm{\cd}_{\infty}\right)\leq C\d^{-d/\left(\left\lfloor d\right\rfloor +1\right)}$
so that $\int_{0}^{1}\sqrt{1+\log\mathscr{N}\left(\e,\cG_{2},L_{2}\left(P\right)\right)}d\e<\infty$.
Moreover, for any $\left(\t,h\right),\left(\tilde{\t},\tilde{h}\right)\in\T\times{\cal H}$,
we have
\begin{align*}
 & \left|\left(g_{\t,h}-g_{\t_{0},h}\right)-\left(g_{\tilde{\t},\tilde{h}}-g_{\t_{0},\tilde{h}}\right)\right|\\
\leq & \left|g_{\t,h}-g_{\tilde{\t},h}\right|+\left|\left(g_{\tilde{\t},h}-g_{\t_{0},h}\right)-\left(g_{\tilde{\t},\tilde{h}}-g_{\t_{0},\tilde{h}}\right)\right|\\
\leq & \left|h\left(x\right)\right|\sum_{j\in J}\ind\left\{ \left|x_{j}^{'}\t\right|\leq\norm{x_{j}}\norm{\tilde{\t}-\t}\right\} +\left|h\left(x\right)-\tilde{h}\left(x\right)\right|\cd1\\
\leq & M\sum_{j\in J}\ind\left\{ \left|x_{j}^{'}\t\right|\leq\norm{x_{j}}\norm{\tilde{\t}-\t}\right\} +J\norm{\tilde{h}-h}_{\infty}
\end{align*}
so that
\begin{align*}
 & P\left(\left(g_{\t,h}-g_{\t_{0},h}\right)-\left(g_{\tilde{\t},\tilde{h}}-g_{\t_{0},\tilde{h}}\right)\right)^{2}\\
\leq & P\left(\left(M^{2}+2M\norm{\tilde{h}-h}_{\infty}\right)\sum_{j\in J}\ind\left\{ \left|x_{j}^{'}\t\right|\leq\norm{x_{j}}\norm{\tilde{\t}-\t}\right\} +\norm{\tilde{h}-h}_{\infty}^{2}\right)\\
= & \left(M^{2}+2M\norm{\tilde{h}-h}_{\infty}\right)\sum_{j\in J}P\left\{ \left|x_{j}^{'}\t\right|\leq\norm{x_{j}}\norm{\tilde{\t}-\t}\right\} +\norm{\tilde{h}-h}_{\infty}^{2}\\
\leq & M^{'}\norm{\tilde{\t}-\t}+\norm{\tilde{h}-h}_{\infty}^{2}
\end{align*}
Hence, following the proof of Theorem 3 (with Conditions 3.2 and 3.3)
in \citet*{chen2003estimation}, for any $\T_{\e}$ that is an $\e$-cover
of $\T$ and ${\cal H}_{\e}$ that is an $\e$-cover of ${\cal H}$,
we deduce that $\T_{\e}\times{\cal H}_{\e}$ is a $\sqrt{M^{'}\e+\e^{2}}\leq\sqrt{M^{''}\e}$
bracket for $\left(\cG,L_{2}\left(P\right)\right)$, implying that
\[
\log\mathscr{N}_{[]}\left(\e,\cG,\norm{\cd}_{\infty}\right)\leq\log\mathscr{N}\left(\e^{2},{\cal \T},\norm{\cd}\right)+\log\mathscr{N}\left(\e^{2},{\cal H},\norm{\cd}_{\infty}\right)\leq2d\left(C-\log\e\right)+\e^{-\frac{2d}{\left\lfloor d\right\rfloor +1}}
\]
and hence 
\begin{align*}
J & :=\int_{0}^{1}\sqrt{1+\log\mathscr{N}_{[]}\left(\e,\cG_{2},L_{2}\left(P\right)\right)}d\e\leq\int_{0}^{1}\sqrt{2d\left(C-\log\e\right)+\e^{-\frac{2d}{\left\lfloor d\right\rfloor +1}}}d\e\leq C^{'}\int_{0}^{1}\e^{-\frac{d}{\left\lfloor d\right\rfloor +1}}d\e<\infty.
\end{align*}
\end{proof}

\subsection{Proof of Lemma \ref{lemTerm1prime}}
\begin{proof}
At any given $x=\left(x_{1},...,x_{J}\right)\in{\cal X}$, notice
that
\[
g_{+,\t_{0},h_{0}}\left(x\right)=\left[h_{0}\left(x\right)\right]_{+}\prod_{j=1}^{J}\ind\left\{ x_{j}^{'}\t_{0}\le0\right\} =0
\]
and thus
\begin{align*}
\left|g_{+,\t,h_{0}}\left(x\right)-g_{+,\t_{0},h_{0}}\left(x\right)\right|= & \left[h_{0}\left(x\right)\right]_{+}\prod_{j=1}^{J}\ind\left\{ x_{j}^{'}\t\leq0\right\} ,
\end{align*}
which is nonzero if and only if
\begin{equation}
h_{0}\left(x\right)>0\quad\text{and}\quad x_{j}^{'}\t\leq0\ \forall j\in J.\label{eq:MMI_dif_not_0}
\end{equation}
Let $x$ be such that \eqref{eq:MMI_dif_not_0} holds. Then the set
\begin{align*}
J_{+} & :=\left\{ j:1\leq j\leq J\text{ and }x_{j}^{'}\t_{0}>0\right\} 
\end{align*}
is nonempty, since \eqref{eq:MMI_dif_not_0} and $J_{+}=\es$ would
imply that $g_{+,\t_{0},h_{0}}\left(x\right)>0$, which is not possible.

Now, define
\[
\ol x_{j}:=\begin{cases}
x_{j}-M\norm{\t-\t_{0}}\t_{0}^{'}, & \forall j\in J_{+},\\
x_{j}, & \forall j\notin J_{+}.
\end{cases}
\]
Then by \eqref{eq:MMI_dif_not_0} and the definition of $J_{+},$
\[
\ol x_{j}^{'}\t_{0}=\begin{cases}
x_{j}^{'}\t_{0}-M\norm{\t-\t_{0}}\leq x_{j}^{'}\t_{0}+x^{'}\left(\t-\t_{0}\right)=x^{'}\t\leq0, & \text{if }j\in J_{+},\\
x_{j}^{'}\t_{0}\leq0, & \text{if }j\notin J_{+},
\end{cases}
\]
or equivalently,
\[
\ol x_{j}^{'}\t_{0}\leq0,\quad\forall j\in J,
\]
which, by the multi-index single-crossing condition \eqref{eq:Mono_MMI},
implies that
\[
h_{0}\left(\ol x\right)\leq0.
\]

Now we have
\begin{align*}
h_{0}\left(x\right)>0\geq h_{0}\left(\ol x\right) & =h_{0}\left(x\right)+h_{0}\left(\ol x\right)-h_{0}\left(x\right)\\
 & \geq h_{0}\left(x\right)-\left|\sup_{\tilde{x}}\Dif_{\text{vec}\left(\tilde{x}\right)}h_{0}\left(\tilde{x}\right)\right|\cd\norm{\text{vec}\left(\ol x\right)-\text{vec}\left(x\right)}\\
 & \geq h_{0}\left(x\right)-M\cd M\cd\norm{\t-\t_{0}}\cd\norm{\left(\sum_{j\in J_{+}}e_{j}\right)\otimes\ind_{d}},\\
 & \geq h_{0}\left(x\right)-\sqrt{\#\left(J_{+}\right)}M^{2}\norm{\t-\t_{0}}\\
 & \geq h_{0}\left(x\right)-\sqrt{J}M^{2}\norm{\t-\t_{0}}
\end{align*}
and thus
\[
0<h_{0}\left(x\right)<\sqrt{J}M^{2}\cd\norm{\t-\t_{0}}=O\left(\norm{\t-\t_{0}}\right).
\]
Hence, for $\norm{\t-\t_{0}}\leq\d$
\begin{align*}
\left|g_{+,\t,h_{0}}\left(x\right)-g_{+,\t_{0},h_{0}}\left(x\right)\right| & =\left[h_{0}\left(x\right)\right]_{+}\cd\prod_{j=1}^{J}\ind\left\{ x_{j}^{'}\t\leq0\right\} .\\
 & \leq C\norm{\t-\t_{0}}\cd\sum_{j=1}^{J}\ind\left\{ x_{j}^{'}\t\leq0<x_{j}^{'}\t_{0}\right\} \\
 & \leq C\d\sum_{j=1}^{J}\ind\left\{ \left|x^{'}\t_{0}\right|\leq\norm x\d\right\} .
\end{align*}
Similarly, the arguments above can be adapted to bound $\left|g_{-,\t,h_{0}}\left(x\right)-g_{-,\t_{0},h_{0}}\left(x\right)\right|$.

Define ${\cal \cG}_{1,\d}:=\left\{ g_{\t,h_{0}}-g_{\t_{0},h_{0}}:\ \norm{\t-\t_{0}}\leq\d\right\} $.
By the arguments above, ${\cal \cG}_{1,\d}$ has an envelope $G_{1,\d}$
givcen by
\[
\left|g_{\t,h_{0}}\left(x\right)-g_{\t_{0},h_{0}}\left(x\right)\right|\leq C\d\sum_{j=1}^{J}\ind\left\{ \left|x^{'}\t_{0}\right|\leq\norm x\d\right\} =:G_{1,\d}.
\]
Moreover,
\begin{align*}
PG_{1,\d}^{2} & =\E\left[C^{2}\d^{2}\sum_{j=1}^{J}\ind\left\{ \left|X_{ij}^{'}\t_{0}\right|\leq\norm{X_{ij}}\d\right\} \right]\\
 & =JC^{2}\d^{2}\P\left(\left|\frac{X_{ij}^{'}}{\norm{X_{ij}}}\t_{0}\right|\leq\d\right)\leq C^{2}\d^{3}.
\end{align*}
Now, since ${\cal \cG}_{1,\d}\subseteq\cG$, we have $\mathscr{N}\left(\e,\cG_{1,\d},L_{2}\left(P\right)\right)\leq\mathscr{N}\left(\e,\cG,L_{2}\left(P\right)\right)$
and by Lemma \ref{lem:J_uni_finite}
\[
J_{1,\d}:=\int_{0}^{1}\sqrt{1+\log\mathscr{N}\left(\e,\cG_{1,},L_{2}\left(P\right)\right)}d\e\leq J<\infty.
\]
Then, by VW Theorem 2.14.1, we have
\[
P\sup_{g\in{\cal \cG}_{1,\d}}\left|\GG_{n}\left(g\right)\right|\leq J_{1,\d}\sqrt{PG_{1,\d}^{2}}\leq J_{1}C\d^{\frac{3}{2}}=M_{1}\d^{\frac{3}{2}}.
\]
\end{proof}

\subsection{Proof of Lemma \ref{lemTerm2prime}}
\begin{proof}
Define ${\cal \cG}_{2,\d,n}:=\left\{ g_{\t,h}-g_{\t_{0},h}-g_{\t,h_{0}}+g_{\t_{0},h_{0}}:\ \norm{\t-\t_{0}}\leq\d,\norm{h-h_{0}}_{\infty}\leq Ka_{n}\right\} $.
Then we have
\begin{align*}
 & \left|g_{+,\t,h}-g_{+,\t_{0},h}-g_{+,\t,h_{0}}+g_{+,\t_{0},h_{0}}\right|\\
= & \left|\left[h\left(x\right)\right]_{+}-\left[h_{0}\left(x\right)\right]_{+}\right|\left|\prod_{j}\ind\left\{ x_{j}^{'}\t\leq0\right\} -\prod_{j}\ind\left\{ x_{j}^{'}\t_{0}\leq0\right\} \right|\\
\leq & Ka_{n}\sum_{j=1}^{J}\ind\left\{ \left|x_{j}^{'}\t_{0}\right|\leq\norm{x_{j}}\d\right\} 
\end{align*}
and similarly for $g_{-,\t,h}$. Hence, an envelope function $G_{2,\d,n}$
for ${\cal \cG}_{2,\d,n}$ is given by
\begin{align*}
 & \left|g_{\t,h}-g_{\t_{0},h}-g_{\t,h_{0}}+g_{\t_{0},h_{0}}\right|\\
\leq & =:Ka_{n}\sum_{j=1}^{J}\ind\left\{ \left|x_{j}^{'}\t_{0}\right|\leq\norm{x_{j}}\d\right\} =:G_{2,n,\d}
\end{align*}
with
\begin{align*}
PG_{2,n,\d}^{2} & =K^{2}a_{n}^{2}\sum_{j=1}^{J}\P\left(\left|\frac{X_{ij}^{'}}{\norm{X_{ij}}}\t_{0}\right|\leq\d\right)\leq Ca_{n}^{2}\d.
\end{align*}
Since ${\cal \cG}_{2,\d,n}\subseteq\cG-{\cal \cG}_{1,\d}:=\left\{ g-\tilde{g}:g\in\cG,\tilde{g}\in\cG_{1,\d}\right\} $,
by Lemma 9.14 of \citet{kosorok2007introduction}, ${\cal \cG}_{2,\d,n}$
must also have bounded uniform entropy integrals. Hence,
\[
J_{2}:=\int_{0}^{1}\sqrt{1+\log\mathscr{N}\left(\e,\cG_{2},L_{2}\left(P\right)\right)}d\e<\infty,
\]
and by VW Theorem 2.14.1, 
\[
P\sup_{g\in{\cal \cG}_{2,\d,n}}\norm{\GG_{n}\left(g\right)}\leq J_{2,\d}\sqrt{PG_{2,n,\d}^{2}}\leq J_{2}Ca_{n}\sqrt{\d}=Ma_{n}\sqrt{\d}.
\]
\end{proof}

\subsection{\label{subsec:pf_Term3i}Proof of Lemma \ref{lem:Term_3}}

We first cite the following result in \citet*{absil2013extrinsic}
about the extrinsic representation of the Riemannian (surface) gradients
and Hessians on $\S^{d-1}$ via standard gradients and Hessians in
the ambient space $\R^{d}$ of $\S^{d}$.
\begin{lem}[Riemannian (Surface) Gradient and Hessian]
\label{lem:SurfGrad}Let $\Psi:\R^{d}\to\R$ be a differentiable
function in the standard sense, and let $\psi:\S^{d-1}\to\R$ be the
restriction of $\Psi$ on $\S^{d-1}$:
\[
\psi\left(\t\right)=\Psi\left(\t\right),\quad\forall\t\in\S^{d-1},
\]
Let $\Dif_{\t},\Dif_{\t\t}$ denote the standard gradient and Hessian
in $\R^{d}$. Let $\Dif_{\t}^{{\cal S}},\Dif_{\t\t}^{{\cal S}}$ denotes
the Riemannian (surface) gradient and Hessian on $\S^{d-1}$. Then,
for any $\t_{0}\in\S^{d-1}$,
\begin{align*}
\Dif_{\t}^{{\cal S}}\psi\left(\t_{0}\right) & =\Dif_{\t}\Psi\left(\t_{0}\right)-\left\langle \t_{0},\Dif_{\t}\Psi\left(\t_{0}\right)\right\rangle \t_{0}^{'}=\Dif_{\t}\Psi\left(\t_{0}\right)\left(I_{d}-\t_{0}\t_{0}^{'}\right)
\end{align*}
\begin{align*}
\Dif_{\t\t}^{{\cal S}}\psi\left(\t_{0}\right) & =\left(I_{d}-\t_{0}\t_{0}^{'}\right)\Dif_{\t\t}\Psi\left(\t_{0}\right)\left(I_{d}-\t_{0}\t_{0}^{'}\right)-\Dif_{\t}\Psi\left(\t_{0}\right)\t_{0}\left(I_{d}-\t_{0}\t_{0}^{'}\right)
\end{align*}
with $\Dif_{\t}^{{\cal S}}psi\left(\t_{0}\right),\Dif_{\t}\Psi\left(\t_{0}\right)$
written as $1\times d$ row vectors\footnote{Hence $\Dif_{\t}\Psi\left(\t_{0}\right)\left(\t-\t_{0}\right)$ is
a scalar as $\t-\t_{0}$ is a column vector. To clarify, all vectors
are by default column vectors in this paper unless otherwise noted.}, $\Dif_{\t\t}^{{\cal S}}\psi\left(\t_{0}\right),\Dif_{\t\t}\Psi\left(\t_{0}\right)$
as $d\times d$ matrices, and $I_{d}$ denoting the $d\times d$ identity
matrix.
\end{lem}
\noindent We also state the following elementary results on change
of coordinates with respect to an orthonormal basis in $\R^{d}$,
which will be heavily exploited subsequently.
\begin{defn}[Change of Coordinates]
\label{def:BasisT} Let $\left\{ \t_{0},\tilde{e}_{2},..,\tilde{e}_{d}\right\} $
be an orthonormal basis in $\R^{d}$. Define $T_{\t_{0}}$ to be the
$d\times d$ basis transformation matrix
\[
T_{\t_{0}}:=\left(\t_{0},\tilde{e}_{2},..,\tilde{e}_{d}\right).
\]
so that $T_{\t_{0}}^{'}x=\left(\t_{0}^{'}x,\tilde{e}_{2}^{'}x,..,\tilde{e}_{d}^{'}x\right)$.
\end{defn}
\begin{lem}
\label{lem:ChangeBasis} (i) $T_{\t_{0}}^{'}=T_{\t_{0}}^{-1}$. (ii)
$\left|\text{det}\left(T_{\t_{0}}\right)\right|=1$, (iii) $u^{'}T_{\t_{0}}^{'}\t_{0}=u_{1}$
and 
\[
\left(I-\t_{0}\t_{0}^{'}\right)T_{\t_{0}}u\equiv\left(I-\t_{0}\t_{0}^{'}\right)T_{\t_{0}}\ol u_{-1},\ \forall u\in\R^{d}
\]
where $\ol u_{-1}:=\left(0,u_{-1}^{'}\right)^{'}\in\R^{d}$ and $u_{-1}:=\left(u_{2},...,u_{d}\right)^{'}\in\R^{d-1}$.
\end{lem}
\begin{proof}
(i)(ii) are elementary. (iii)(iv) follow from the observation that
$T_{\t_{0}}^{'}\t_{0}=\left(1,0,...,0\right)^{'}$ and
\begin{align*}
\left(I-\t_{0}\t_{0}^{'}\right)T_{\t_{0}} & =\left(\t_{0},\tilde{e}_{2},..,\tilde{e}_{d}\right)-\left(\t_{0},\tilde{e}_{2},..,\tilde{e}_{d}\right)\left(\begin{array}{c}
1\\
0\\
\vdots\\
0
\end{array}\right)=\left(0,\tilde{e}_{2},..,\tilde{e}_{d}\right).
\end{align*}
\end{proof}

\subsubsection*{Alternative Representation of $h_{0}\left(x\right)$}

Under the change of coordinate from $x$ to $u=T_{\t_{0}}^{'}x$,
the function $h_{0}\left(x\right)$ can be equivalently written as
a function of $u$ as
\[
h_{0u}\left(u\right):=h_{0}\left(T_{\t_{0}}u\right).
\]
Under this change of coordinate, several important properties of $h_{0}$
will be inherited by $h_{0u}$.
\begin{lem}
\label{lem:h0u} $h_{0u}$ has the following properties:
\begin{itemize}
\item i) $h_{0u}$ is twice differentiable with uniformly bounded derivatives.
\item ii) $h_{0u}\left(u_{1},u_{-1}\right)\lesseqqgtr0$ if and only if
$u_{1}\lesseqqgtr0$, for any $u_{-1}$.
\item iii) $\Dif_{u_{1}}h_{0u}\left(u_{1},u_{-1}\right)=\Dif_{x}h_{0}\left(x\right)^{'}\t_{0}$.
\end{itemize}
\end{lem}
\begin{proof}
i) and ii) are trivial. iii) follows from the chain rule:
\begin{align*}
\Dif_{u_{1}}h_{0u}\left(u_{1},u_{-1}\right) & =\Dif_{u}h_{0u}\left(u\right)^{'}e_{1}=\Dif_{u}h_{0}\left(T_{\t_{0}}u\right)^{'}e_{1}\\
 & =\Dif_{x}h_{0}\left(T_{\t_{0}}u\right)^{'}T_{\t_{0}}e_{1}=\Dif_{x}h_{0}\left(T_{\t_{0}}u\right)^{'}\t_{0}.
\end{align*}
\end{proof}
We emphasize the following intuitive property about $h_{0}$ and $h_{0u}$.
\begin{lem}
\label{lem:FDnot0}Under Assumption \ref{assu:Basic}(c)(d), for any
$x\in{\cal X}$ s.t. $x^{'}\t_{0}=0$, or equivalently for any $u$
s.t. $u_{1}=0$, we have
\[
\Dif_{u_{1}}h_{0u}\left(0,u_{-1}\right)=\Dif_{x}h_{0}\left(x\right)^{'}\t_{0}>0.
\]
\end{lem}
\begin{proof}
Since $h_{0}\left(x\right)=F\left(\rest{x^{'}\t_{0}}x\right)$, we
have
\[
\Dif_{x}h_{0}\left(x\right)=f\left(\rest{x^{'}\t_{0}}x\right)\t_{0}+\rest{\frac{\p}{\p x}F\left(\rest{\e}x\right)}_{\e=x^{'}\t_{0}}.
\]
Since $F\left(\rest 0x\right)\equiv\frac{1}{2}$ for any $x$, we
have
\[
\frac{\p}{\p x}F\left(\rest 0x\right)\equiv{\bf 0}.
\]
Hence, for any $x\in{\cal X}$ s.t. $x^{'}\t_{0}=0$, we have
\[
\Dif_{x}h_{0}\left(x\right)^{'}\t_{0}=f\left(\rest 0x\right)\t_{0}^{'}\t_{0}+\frac{\p}{\p x}F\left(\rest 0x\right)^{'}\t_{0}=f\left(\rest 0x\right)>0.
\]
\end{proof}

\subsubsection*{Proof of Lemma \ref{lem:Term_3}(i)}
\begin{proof}
Consider the following first-order Taylor expansion of $f_{n,\t}$
around $\t_{0}$:
\begin{align*}
\psi_{n,\t}\left(z\right)-\psi_{n,\t_{0}}\left(z\right)=\  & \left(y-\frac{1}{2}\right)\left[\Phi\left(\frac{x^{'}\t}{b_{n}}\right)-\Phi\left(\frac{x^{'}\t_{0}}{b_{n}}\right)\right]\\
=\  & \left(y-\frac{1}{2}\right)\text{\ensuremath{\Dif_{\t}^{{\cal S}}\Phi\left(\frac{\xi\left(x\right)}{b_{n}}\right)}}\left(\t-\t_{0}\right)\\
=\  & \left(y-\frac{1}{2}\right)\Dif_{\t}\Phi\left(\frac{\xi\left(x\right)}{b_{n}}\right)\left(I_{d}-\t_{0}\t_{0}^{'}\right)\left(\t-\t_{0}\right)\\
=\  & \left(y-\frac{1}{2}\right)\phi\left(\frac{\xi\left(x\right)}{b_{n}}\right)\frac{x^{'}}{b_{n}}\left(I-\t_{0}\t_{0}^{'}\right)\left(\t-\t_{0}\right)
\end{align*}
for some $\xi\left(x\right)$ that lies between $x^{'}\t$ and $\text{\ensuremath{x^{'}\t_{0}}}$.
Then the function space
\[
{\cal G}_{n,\d}^{\psi}:=\left\{ \psi_{n,\t}\left(z\right)-\psi_{n,\t_{0}}\left(z\right):\ \norm{\psi_{n,\t}\left(z\right)-\psi_{n,\t_{0}}\left(z\right)}\leq\d\right\} 
\]
has an envelope $\text{\ensuremath{\Psi_{n,\d}}}$ given b
\begin{align}
\left|\psi_{n,\t}\left(z\right)-\psi_{n,\t_{0}}\left(z\right)\right|=\  & \left|y-\frac{1}{2}\right|\left|\Phi\left(\frac{x^{'}\t}{b_{n}}\right)-\Phi\left(\frac{x^{'}\t_{0}}{b_{n}}\right)\right|\nonumber \\
=\  & \frac{1}{2b_{n}}\phi\left(\frac{\xi\left(x\right)}{b_{n}}\right)\left|x^{'}\left(I-\t_{0}\t_{0}^{'}\right)\left(\t-\t_{0}\right)\right|\nonumber \\
\leq\  & \frac{1}{2b_{n}}\phi\left(\frac{\xi\left(x\right)}{b_{n}}\right)\left|x^{'}\left(I-\t_{0}\t_{0}^{'}\right)\frac{\t-\t_{0}}{\norm{\t-\t_{0}}}\right|\d\nonumber \\
\leq\  & \frac{1}{2b_{n}}\ol{\phi}_{n,\d}\left(x^{'}\t_{0}\right)\norm{\left(I-\t_{0}\t_{0}^{'}\right)x}\d\label{eq:Bound_xi}\\
=:\  & \Psi_{n,\d}\nonumber 
\end{align}
where the function $\ol{\phi}_{n,\d}$ in \eqref{eq:Bound_xi} is
defined as
\begin{align}
\ol{\phi}_{n,\d}\left(x^{'}\t_{0}\right) & :=\max_{\e:\left|\e\right|\leq\d}\phi\left(\frac{x^{'}\t_{0}+\e}{b_{n}}\right)=\phi\left(0\right)\ind\left\{ \left|x^{'}\t_{0}\right|\leq\d\right\} +\phi\left(\frac{\left|x^{'}\t_{0}\right|-\d}{b_{n}}\right)\ind\left\{ \left|x^{'}\t_{0}\right|>\d\right\} \label{eq:phi_bar_nd}
\end{align}
given that $\phi\left(t\right)$ is decreasing in $\left|t\right|$.
This ensures the inequality in \eqref{eq:Bound_xi} by $\phi\left(\frac{\xi\left(x\right)}{b_{n}}\right)\leq\ol{\phi}_{n,\d}\left(x^{'}\t_{0}\right)$,
because $\xi\left(x\right)$ lies between $x^{'}\t_{0}$ and $x^{'}\t$,
while
\[
x^{'}\t\in\left[x^{'}\t_{0}-\norm x\d,x^{'}\t_{0}+\norm x\d\right]\subseteq\left[x^{'}\t_{0}-\d,x^{'}\t_{0}+\d\right],
\]
so that $\xi\left(x\right)\in\left[x^{'}\t_{0}-\d,x^{'}\t_{0}+\d\right]$.

Now, impose the change of coordinates to the basis $\left\{ \t_{0},\tilde{e}_{2},..,\tilde{e}_{d}\right\} $
as i\textbackslash R\textasciicircum\{Jd\}n Definition \ref{def:BasisT}
with $u:=T_{\t_{0}}^{'}x$ and thus $x=T_{\t_{0}}u$. Then, by Lemma
\ref{lem:ChangeBasis},
\begin{align*}
P\Psi_{n,\d}^{2} & =\frac{\d^{2}}{4b_{n}^{2}}\int\ol{\phi}_{n,\d}^{2}\left(x^{'}\t_{0}\right)x^{'}\left(I-\t_{0}\t_{0}^{'}\right)xp_{x}dx\\
 & =\frac{\d^{2}}{4b_{n}^{2}}\int\ol{\phi}_{n,\d}^{2}\left(u^{'}T_{\t_{0}}^{'}\t_{0}\right)u^{'}T_{\t_{0}}^{'}\left(I-\t_{0}\t_{0}^{'}\right)T_{\t_{0}}up_{x}dT_{\t_{0}}u\\
 & =\frac{\d^{2}}{4b_{n}^{2}}\int\ol{\phi}_{n,\d}^{2}\left(u_{1}\right)\ol u_{-1}^{'}T_{\t_{0}}^{'}\left(I-\t_{0}\t_{0}^{'}\right)T_{\t_{0}}\ol u_{-1}p_{x}du\\
 & =\frac{\d^{2}}{4b_{n}^{2}}\int\int\ol{\phi}_{n,\d}^{2}\left(u_{1}\right)du_{1}\ol u_{-1}^{'}T_{\t_{0}}^{'}\left(I-\t_{0}\t_{0}^{'}\right)T_{\t_{0}}\ol u_{-1}p_{x}du_{-1}
\end{align*}
while 
\begin{align*}
\int\ol{\phi}_{n,\d}^{2}\left(u_{1}\right)du_{1}= & \int\phi^{2}\left(0\right)\ind\left\{ \left|u_{1}\right|\leq\d\right\} du_{1}+\int\phi^{2}\left(\frac{\left|u_{1}\right|-\d}{b_{n}}\right)\ind\left\{ \left|u_{1}\right|>\d\right\} du_{1}\\
= & 2\phi^{2}\left(0\right)\int_{0}^{\d}du_{1}+2\int_{\d}^{1}\phi^{2}\left(\frac{u_{1}-\d}{b_{n}}\right)du_{1}\\
= & 2\phi^{2}\left(0\right)\d+2\int_{0}^{b_{n}^{-1}\left(1-\d\right)}\phi^{2}\left(\zeta_{1}\right)d\left(b_{n}\zeta_{1}+\d\right)\text{with }\zeta_{1}:=\frac{u_{1}-\d}{b_{n}}\\
\leq & 2\phi^{2}\left(0\right)\d+2b_{n}\int_{0}^{\infty}\phi^{2}\left(\zeta_{1}\right)d\zeta_{1}\\
\leq & C\left(\d+b_{n}\right)
\end{align*}
and $\int\ol u_{-1}^{'}T_{\t_{0}}^{'}\left(I-\t_{0}\t_{0}^{'}\right)T_{\t_{0}}\ol u_{-1}p_{x}du_{-1}\in\left(0,\infty\right)$.
Hence, 
\begin{align*}
P\Psi_{n,\d}^{2} & \leq\frac{\d^{2}}{4b_{n}^{2}}C\left(\d+b_{n}\right),
\end{align*}
and by VW Theorem 2.14.1, we have
\begin{align*}
P\sup_{\norm{\t-\t_{0}}\leq\d}\left|\GG_{n}\left(\psi_{n,\t}-\psi_{n,\t_{0}}\right)\right| & \leq J\sqrt{P\Psi_{n,\d}^{2}}\leq M_{1}\frac{\d}{b_{n}}\left(\d+b_{n}\right)^{\frac{1}{2}}.
\end{align*}
\end{proof}

\subsubsection*{Proof of Lemma \ref{lem:Term_3}(ii)}
\begin{proof}
First, consider the following second-order Taylor expansion of $\psi_{n,\t}-\psi_{n,\t_{0}}$:
\begin{align*}
 & \psi_{n,\t}\left(z\right)-\psi_{n,\t_{0}}\left(z\right)\\
=\  & \left(y-\frac{1}{2}\right)\left[\text{\ensuremath{\Dif_{\t}^{{\cal S}}\Phi\left(\frac{x^{'}\t_{0}}{b_{n}}\right)}}\left(\t-\t_{0}\right)+\frac{1}{2}\left(\t-\t_{0}\right)^{'}\Dif_{\t\t}^{{\cal S}}\Phi\left(\frac{\xi\left(x\right)}{b_{n}}\right)\left(\t-\t_{0}\right)\right]\\
=\  & \left(y-\frac{1}{2}\right)\Dif_{\t}\Phi\left(\frac{x^{'}\t_{0}}{b_{n}}\right)\left(I_{d}-\t_{0}\t_{0}^{'}\right)\left(\t-\t_{0}\right)\\
=\  & +\frac{1}{2}\left(y-\frac{1}{2}\right)\left(\t-\t_{0}\right)\left(I_{d}-\t_{0}\t_{0}^{'}\right)\Dif_{\t\t}\Phi\left(\frac{\xi\left(x\right)}{b_{n}}\right)\left(I_{d}-\t_{0}\t_{0}^{'}\right)\left(\t-\t_{0}\right)\\
 & -\frac{1}{2}\left(y-\frac{1}{2}\right)\Dif_{\t}\Phi\left(\frac{\xi\left(x\right)}{b_{n}}\right)\t_{0}\left(I_{d}-\t_{0}\t_{0}^{'}\right)\left(\t-\t_{0}\right)\\
=\  & \left(y-\frac{1}{2}\right)\phi\left(\frac{x^{'}\t_{0}}{b_{n}}\right)\frac{x^{'}}{b_{n}}\left(I-\t_{0}\t_{0}^{'}\right)\left(\t-\t_{0}\right)\\
 & +\frac{1}{2}\left(y-\frac{1}{2}\right)\left(\t-\t_{0}\right)^{'}\left(I_{d}-\t_{0}\t_{0}^{'}\right)\phi^{'}\left(\frac{\xi\left(x\right)}{b_{n}}\right)\cd\frac{xx^{'}}{b_{n}^{2}}\left(I_{d}-\t_{0}\t_{0}^{'}\right)\left(\t-\t_{0}\right)\\
 & -\frac{1}{2}\left(y-\frac{1}{2}\right)\phi\left(\frac{\xi\left(x\right)}{b_{n}}\right)\frac{x^{'}}{b_{n}}\t_{0}\left(\t-\t_{0}\right)^{'}\left(I_{d}-\t_{0}\t_{0}^{'}\right)\left(\t-\t_{0}\right)
\end{align*}
for some $\xi\left(x\right)$ between $x^{'}\t_{0}$ and $x^{'}\t$.
Then:
\begin{align}
 & P\left(\psi_{n,\t}\left(z\right)-\psi_{n,\t_{0}}\left(z\right)\right)\nonumber \\
= & \int\E\left[\rest{y_{i}-\frac{1}{2}}X_{i}=x\right]\left(\Phi\left(\frac{x^{'}\t}{b_{n}}\right)-\Phi\left(\frac{x^{'}\t_{0}}{b_{n}}\right)\right)p_{x}dx\nonumber \\
=\  & \left[\int h_{0}\left(x\right)\phi\left(\frac{x^{'}\t_{0}}{b_{n}}\right)\frac{x^{'}}{b_{n}}\left(I-\t_{0}\t_{0}^{'}\right)p_{x}dx\right]\left(\t-\t_{0}\right)\label{eq:A_1}\\
 & +\frac{1}{2}\left(\t-\t_{0}\right)^{'}\left[\int h_{0}\left(x\right)-\frac{1}{2}\phi^{'}\left(\frac{\xi\left(x\right)}{b_{n}}\right)\left(I_{d}-\t_{0}\t_{0}^{'}\right)\frac{xx^{'}}{b_{n}^{2}}\left(I_{d}-\t_{0}\t_{0}^{'}\right)p_{x}dx\right]\left(\t-\t_{0}\right)\label{eq:A_2}\\
 & -\frac{1}{2}\left[\int h_{0}\left(x\right)\phi\left(\frac{\xi\left(x\right)}{b_{n}}\right)\frac{x^{'}\t_{0}}{b_{n}}p_{x}d_{x}\right]\left(\t-\t_{0}\right)^{'}\left(I_{d}-\t_{0}\t_{0}^{'}\right)\left(\t-\t_{0}\right)\label{eq:A_3}\\
=: & A_{n,1}\left(\t-\t_{0}\right)+\left(\t-\t_{0}\right)^{'}A_{n,2}\left(\t-\t_{0}\right)+A_{n,3}\left(\t-\t_{0}\right)^{'}\left(I_{d}-\t_{0}\t_{0}^{'}\right)\left(\t-\t_{0}\right)\label{eq:A_ALL}
\end{align}
In the following we deal with $A_{n,1},A_{n,2},A_{n,3}$ separately.

First, for $A_{n,1}$, we consider the bracketed term in \eqref{eq:A_1}
and expand $F\left(t\right)$ around $t=0$:
\begin{align*}
A_{n,1}:=\  & \int h_{0}\left(x\right)\phi\left(\frac{x^{'}\t_{0}}{b_{n}}\right)\frac{x^{'}}{b_{n}}\left(I-\t_{0}\t_{0}^{'}\right)p_{x}dx\\
=\  & \frac{1}{b_{n}}\int h_{0}\left(T_{\t_{0}}u\right)\phi\left(\frac{u^{'}T_{\t_{0}}^{'}\t_{0}}{b_{n}}\right)u^{'}T_{\t_{0}}^{'}\left(I-\t_{0}\t_{0}^{'}\right)p_{x}du\\
=\  & \frac{1}{b_{n}}\int h_{0u}\left(u_{1},u_{-1}\right)\phi\left(\frac{u_{1}}{b_{n}}\right)\ol u_{-1}^{'}T_{\t_{0}}^{'}\left(I-\t_{0}\t_{0}^{'}\right)p_{x}du_{1}du_{-1}\\
=\  & \frac{1}{b_{n}}\int h_{0u}\left(b_{n}\zeta_{1},u_{-1}\right)\phi\left(\zeta_{1}\right)\ol u_{-1}^{'}T_{\t_{0}}^{'}\left(I-\t_{0}\t_{0}^{'}\right)p_{x}d\left(b_{n}\zeta_{1}\right)du_{-1}\text{with }\zeta_{1}:=\frac{u_{1}}{b_{n}}\\
=\  & \int\left[\int\Dif_{u_{1}}h_{0u}\left(0,u_{-1}\right)b_{n}\zeta_{1}+\Dif_{u_{1}}^{2}h_{0u}\left(b_{n}\tilde{\zeta}_{1},u_{-1}\right)\left(b_{n}\zeta_{1}\right)^{2}\right]\phi\left(\zeta_{1}\right)\ol u_{-1}^{'}T_{\t_{0}}^{'}\left(I-\t_{0}\t_{0}^{'}\right)p_{x}d\zeta_{1}du_{-1}\text{ for some }\tilde{\zeta}_{1}\text{ between }0\text{ and }\zeta_{1}\\
=\  & b_{n}\cd\int\int_{-b_{n}^{-1}}^{b_{n}^{-1}}\zeta_{1}\phi\left(\zeta_{1}\right)d\zeta_{1}\Dif_{u_{1}}h_{0u}\left(0,u_{-1}\right)\ol u_{-1}^{'}T_{\t_{0}}^{'}\left(I-\t_{0}\t_{0}^{'}\right)p_{x}du_{-1}\\
 & +b_{n}^{2}\cd\int\int\Dif_{u_{1}}^{2}h_{0u}\left(b_{n}\tilde{\zeta}_{1},u_{-1}\right)\zeta_{1}^{2}\phi\left(\zeta_{1}\right)d\zeta_{1}\cd\int\ol u_{-1}^{'}T_{\t_{0}}^{'}\left(I-\t_{0}\t_{0}^{'}\right)p_{x}du_{-1}\\
=\  & b_{n}^{2}\cd\int\int\Dif_{u_{1}}^{2}h_{0u}\left(b_{n}\tilde{\zeta}_{1},u_{-1}\right)\zeta_{1}^{2}\phi\left(\zeta_{1}\right)d\zeta_{1}\cd\int\ol u_{-1}^{'}T_{\t_{0}}^{'}\left(I-\t_{0}\t_{0}^{'}\right)p_{x}du_{-1}
\end{align*}
since $\int_{-t}^{t}\zeta_{1}\phi\left(\zeta_{1}\right)d\zeta_{1}=0$
for all $t\in\R$. Moreover, noting that $\Dif_{u_{1}}^{2}h_{0u}\left(b_{n}\tilde{\zeta}_{1},u_{-1}\right)\to\Dif_{u_{1}}^{2}h_{0u}\left(0,u_{-1}\right)$
as $n\to\infty$, by the dominated convergence theorem, we have
\begin{align*}
b_{n}^{-2}A_{n,1} & =\int\int\Dif_{u_{1}}^{2}h_{0u}\left(b_{n}\tilde{\zeta}_{1},u_{-1}\right)\zeta_{1}^{2}\phi\left(\zeta_{1}\right)d\zeta_{1}\ol u_{-1}^{'}T_{\t_{0}}^{'}\left(I-\t_{0}\t_{0}^{'}\right)p_{x}du_{-1}\\
 & \to\int_{-\infty}^{\infty}\zeta_{1}^{2}\phi\left(\zeta_{1}\right)d\zeta_{1}\cd\int_{u_{1}=0}\Dif_{u_{1}}^{2}h_{0u}\left(0,u_{-1}\right)\ol u_{-1}^{'}p_{x}du_{-1}\cd T_{\t_{0}}^{'}\left(I-\t_{0}\t_{0}^{'}\right)\\
 & =\int_{u_{1}=0}\Dif_{u_{1}}^{2}h_{0u}\left(0,u_{-1}\right)\ol u_{-1}^{'}\ol u_{-1}^{'}p_{x}du_{-1}\cd T_{\t_{0}}^{'}\left(I-\t_{0}\t_{0}^{'}\right)\\
 & =:A_{1}
\end{align*}
and hence
\begin{equation}
A_{n,1}=A_{1}b_{n}^{2}+o\left(b_{n}^{2}\right).\label{eq:A_n1}
\end{equation}

Second, consider $A_{n,2}$ corresponding to \eqref{eq:A_2}:
\begin{align*}
A_{n,2} & =\left(I-\t_{0}\t_{0}^{'}\right)\left[\int h_{0}\left(x\right)\phi^{'}\left(\frac{\xi\left(x\right)}{b_{n}}\right)\frac{xx^{'}}{b_{n}^{2}}p_{x}dx\right]\left(I-\t_{0}\t_{0}^{'}\right)\\
 & =\left(I-\t_{0}\t_{0}^{'}\right)\left[\int h_{0}\left(x\right)\phi^{'}\left(\frac{x^{'}\t_{0}}{b_{n}}\right)\frac{xx^{'}}{b_{n}^{2}}p_{x}dx\right]\left(I-\t_{0}\t_{0}^{'}\right)\\
 & \quad+\left(I-\t_{0}\t_{0}^{'}\right)\left[\int h_{0}\left(x\right)\phi^{'}\left(\frac{\xi\left(x\right)}{b_{n}}\right)-\phi^{'}\left(\frac{x^{'}\t_{0}}{b_{n}}\right)\cd\frac{xx^{'}}{b_{n}^{2}}p_{x}dx\right]\left(I-\t_{0}\t_{0}^{'}\right)\\
 & =:A_{n,2,1}+A_{n,2,2}
\end{align*}
where 
\begin{align*}
A_{n,2,1} & =\left(I-\t_{0}\t_{0}^{'}\right)\left[\int h_{0}\left(x\right)\phi^{'}\left(\frac{x^{'}\t_{0}}{b_{n}}\right)\frac{xx^{'}}{b_{n}^{2}}p_{x}dx\right]\left(I-\t_{0}\t_{0}^{'}\right)\\
 & =\left(I-\t_{0}\t_{0}^{'}\right)\left[\int h_{0}\left(x\right)\phi^{'}\left(\frac{u_{1}}{b_{n}}\right)\frac{T_{\t_{0}}\ol u_{-1}\ol u_{-1}^{'}T_{\t_{0}}^{'}}{b_{n}^{2}}p_{x}du_{1}du_{-1}\right]\left(I-\t_{0}\t_{0}^{'}\right)\\
 & =\left(I-\t_{0}\t_{0}^{'}\right)T_{\t_{0}}\left[\int\Dif_{u_{1}}h_{0u}\left(b_{n}\tilde{\zeta}_{1},u_{-1}\right)b_{n}\zeta_{1}\phi^{'}\left(\zeta_{1}\right)\frac{\ol u_{-1}\ol u_{-1}^{'}}{b_{n}^{2}}b_{n}d\zeta_{1}du_{-1}\right]T_{\t_{0}}^{'}\left(I-\t_{0}\t_{0}^{'}\right)\\
 & =\left(I-\t_{0}\t_{0}^{'}\right)T_{\t_{0}}\left[\int\Dif_{u_{1}}h_{0u}\left(b_{n}\tilde{\zeta}_{1},u_{-1}\right)\zeta_{1}\phi^{'}\left(\zeta_{1}\right)\ol u_{-1}\ol u_{-1}^{'}d\zeta_{1}dz_{-1}\right]T_{\t_{0}}^{'}\left(I-\t_{0}\t_{0}^{'}\right)\\
 & \to\left(I-\t_{0}\t_{0}^{'}\right)T_{\t_{0}}\cd\int\zeta_{1}\phi^{'}\left(\zeta_{1}\right)d\zeta_{1}\cd\int_{u_{1}=0}\Dif_{u_{1}}h_{0u}\left(0,u_{-1}\right)\ol u_{-1}\ol u_{-1}^{'}d\zeta_{1}dz_{-1}T_{\t_{0}}^{'}\left(I-\t_{0}\t_{0}^{'}\right)\\
 & =-\left(I-\t_{0}\t_{0}^{'}\right)T_{\t_{0}}\left(\int_{u_{1}=0}\Dif_{u_{1}}h_{0u}\left(0,u_{-1}\right)\ol u_{-1}\ol u_{-1}^{'}p_{x}du_{-1}\right)T_{\t_{0}}^{'}\left(I-\t_{0}\t_{0}^{'}\right)\\
 & =:-V
\end{align*}
since
\begin{align*}
\int\zeta_{1}\phi^{'}\left(\zeta_{1}\right)d\zeta_{1} & =\int\zeta_{1}\frac{1}{\sqrt{2\pi}}\left(-\zeta_{1}\right)e^{-\frac{1}{2}\zeta_{1}^{2}}d\zeta_{1}=-\int\zeta_{1}^{2}\phi\left(\zeta_{1}\right)d\zeta_{1}=-1.
\end{align*}
Now for any $\t\in\S^{d-1}$ in a neighborhood of $\t_{0}$, define
\begin{align}
v\left(\t\right) & :=\left(0,v\left(\t\right)_{-1}^{'}\right)^{'}:=T_{\t_{0}}^{'}\left(I-\t_{0}\t_{0}^{'}\right)\left(\t-\t_{0}\right)\nonumber \\
V_{u_{-1}} & :=\int_{u_{1}=0}\Dif_{u_{1}}h_{0u}\left(0,u_{-1}\right)u_{-1}u_{-1}^{'}p_{x}du_{-1}\in\R^{\left(d-1\right)\times\left(d-1\right)}\label{eq:V_tang}\\
V_{\ol u_{-1}} & :=\left(\begin{array}{cc}
0 & {\bf 0}^{'}\\
{\bf 0} & V_{u_{-1}}
\end{array}\right)\label{eq:V_tang_d}
\end{align}
so that 
\[
V=\left(I-\t_{0}\t_{0}^{'}\right)T_{\t_{0}}V_{\ol u_{-1}}T_{\t_{0}}^{'}\left(I-\t_{0}\t_{0}^{'}\right).
\]
Since $\Dif_{u_{1}}h_{0u}\left(0,u_{-1}\right)$ is strictly positive
for any $u_{-1}$
\begin{align*}
\left(\t-\t_{0}\right)^{'}V\left(\t-\t_{0}\right)= & v\left(\t\right)^{'}V_{\ol u_{-1}}v\left(\t\right)=v\left(\t\right)_{-1}^{'}V_{u_{-1}}v\left(\t\right)_{-1}\\
\geq & \l_{\min}\left(V_{u_{-1}}\right)\norm{v\left(\t\right)_{-1}}^{2}=\l_{\min}\left(V_{u_{-1}}\right)\norm{v\left(\t\right)}^{2}
\end{align*}
since $V_{u_{-1}}$ is positive definite and thus $\l_{\min}\left(V_{u_{-1}}\right)>0$.
Furthermore, notice that
\begin{align*}
\norm{v\left(\t\right)}^{2} & =\left(\t-\t_{0}\right)^{'}\left(I-\t_{0}\t_{0}^{'}\right)T_{\t_{0}}T_{\t_{0}}^{'}\left(I-\t_{0}\t_{0}^{'}\right)\left(\t-\t_{0}\right)\\
 & =\left(\t-\t_{0}\right)^{'}\left(I-\t_{0}\t_{0}^{'}\right)I\left(I-\t_{0}\t_{0}^{'}\right)\left(\t-\t_{0}\right)\\
 & =\norm{\left(I-\t_{0}\t_{0}^{'}\right)\left(\t-\t_{0}\right)}^{2}=\norm{\left(I-\t_{0}\t_{0}^{'}\right)\t}^{2}\\
 & =\left(1-\t_{0}^{'}\t\right)\left(1+\t_{0}^{'}\t\right)\\
 & =\norm{\t-\t_{0}}^{2}\left(1-\frac{1}{4}\norm{\t-\t_{0}}^{2}\right)\\
 & \geq\frac{3}{4}\norm{\t-\t_{0}}^{2}\quad\text{for }\norm{\t-\t_{0}}\leq1
\end{align*}
and hence, in a neighborhood of $\t_{0},$we have
\begin{equation}
\left(\t-\t_{0}\right)^{'}V\left(\t-\t_{0}\right)\geq\frac{3}{4}\l_{\min}\left(V_{u_{-1}}\right)\norm{\t-\t_{0}}^{2}=C\norm{\t-\t_{0}}^{2}.\label{eq:Quad_info}
\end{equation}
Now, we turn to $A_{n,2}$ and write $\d:=\norm{\t-\t_{0}}$, then
\begin{align*}
\left|A_{n,2,2}\right| & \leq\left(I-\t_{0}\t_{0}^{'}\right)\int\left|h_{0}\left(x\right)\right|\left|\phi^{'}\left(\frac{\xi\left(x\right)}{b_{n}}\right)-\phi^{'}\left(\frac{x^{'}\t_{0}}{b_{n}}\right)\right|\cd\frac{xx^{'}}{b_{n}^{2}}p_{x}dx\left(I-\t_{0}\t_{0}^{'}\right)\\
 & \leq\left(I-\t_{0}\t_{0}^{'}\right)\int\left|h_{0}\left(x\right)\right|\ol{\phi_{n,\d}^{''}}\left(x^{'}\t_{0}\right)\frac{\left|x^{'}\t-x^{'}\t_{0}\right|}{b_{n}}\cd\frac{xx^{'}}{b_{n}^{2}}p_{x}dx\left(I-\t_{0}\t_{0}^{'}\right)\\
 & \leq\left(I-\t_{0}\t_{0}^{'}\right)\int\left|h_{0}\left(x\right)\right|\ol{\phi_{n,\d}^{''}}\left(x^{'}\t_{0}\right)\frac{\d}{b_{n}}\cd\frac{xx^{'}}{b_{n}^{2}}p_{x}dx\left(I-\t_{0}\t_{0}^{'}\right)
\end{align*}
where
\begin{align*}
\ol{\phi_{n,\d}^{''}}\left(x^{'}\t_{0}\right) & :=\ind\left\{ \frac{\left|x^{'}\t_{0}\right|-\d}{b_{n}}\leq\sqrt{3}\right\} +\left|\phi^{''}\left(\frac{\left|x^{'}\t_{0}\right|-\d}{b_{n}}\right)\right|\ind\left\{ \frac{\left|x^{'}\t_{0}\right|-\d}{b_{n}}>\sqrt{3}\right\} 
\end{align*}
guarantees that $\left|\phi^{''}\left(t\right)\right|\leq\ol{\phi_{n,\d}^{''}}\left(x^{'}\t_{0}\right)$
for any
\[
t\in\left[\frac{x^{'}\t_{0}-\d}{b_{n}},\frac{x^{'}\t_{0}+\d}{b_{n}}\right]
\]
since $\phi^{''}\left(\left|t\right|\right)\leq1$ and $\phi^{''}\left(\left|t\right|\right)$
is decreasing in $\left|t\right|$ for $\left|t\right|\geq\sqrt{3}$.
Hence,
\[
\left|\phi^{'}\left(\frac{\xi\left(x\right)}{b_{n}}\right)-\phi^{'}\left(\frac{x^{'}\t_{0}}{b_{n}}\right)\right|=\left|\phi^{''}\left(\frac{\tilde{\xi}\left(x\right)}{b_{n}}\right)\right|\frac{\left|x^{'}\t-x^{'}\t_{0}\right|}{b_{n}}\leq\ol{\phi_{n,\d}^{''}}\left(x^{'}\t_{0}\right)\frac{\left|x^{'}\t-x^{'}\t_{0}\right|}{b_{n}}
\]
since $\tilde{\xi}\left(x\right)$ lies between $\xi\left(x\right)$
and $x^{'}\t_{0}$, while $\xi\left(x\right)\in\left[x^{'}\t_{0}-\d,x^{'}\t_{0}+\d\right]$.
Then,
\begin{align*}
\left|A_{n,2,2}\right| & \leq\left(I-\t_{0}\t_{0}^{'}\right)\int\left|h_{0}\left(x\right)\right|\left|\ol{\phi_{n,\d}^{''}}\left(x^{'}\t_{0}\right)\right|\frac{\d}{b_{n}}\cd\frac{xx^{'}}{b_{n}^{2}}p_{x}dx\left(I-\t_{0}\t_{0}^{'}\right)\\
 & =\left(I-\t_{0}\t_{0}^{'}\right)\int\left|h_{0}\left(x\right)\right|\left|\ol{\phi_{n,\d}^{''}}\left(x^{'}\t_{0}\right)\right|\frac{\d}{b_{n}}\cd\frac{xx^{'}}{b_{n}^{2}}p_{x}dx\left(I-\t_{0}\t_{0}^{'}\right)\\
 & =\frac{\d}{b_{n}^{3}}\left(I-\t_{0}\t_{0}^{'}\right)\int\left[\int\Dif_{u_{1}}h_{0u}\left(\tilde{u}_{1},u_{-1}\right)\left|u_{1}\right|\ol{\phi_{n,\d}^{''}}\left(u_{1}\right)du_{1}\right]T_{\t_{0}}\ol u_{-1}\ol u_{-1}^{'}T_{\t_{0}}^{'}p_{x}du_{-1}\left(I-\t_{0}\t_{0}^{'}\right)
\end{align*}
where
\begin{align*}
\  & \int\Dif_{u_{1}}h_{0u}\left(\tilde{u}_{1},u_{-1}\right)\left|u_{1}\right|\ol{\phi_{n,\d}^{''}}\left(u_{1}\right)du_{1}\\
=\  & \int\Dif_{u_{1}}h_{0u}\left(\tilde{u}_{1},u_{-1}\right)\ind\left\{ \frac{\left|u_{1}\right|-\d}{b_{n}}\leq\sqrt{3}\right\} \left|u_{1}\right|du_{1}\\
 & +\int\Dif_{u_{1}}h_{0u}\left(\tilde{u}_{1},u_{-1}\right)\left|\phi^{''}\left(\frac{\left|u_{1}\right|-\d}{b_{n}}\right)\right|\ind\left\{ \frac{\left|u_{1}\right|-\d}{b_{n}}>\sqrt{3}\right\} \left|u_{1}\right|du_{1}\\
=\  & 2\int_{0}^{\d+\sqrt{3}b_{n}}\Dif_{u_{1}}h_{0u}\left(\tilde{u}_{1},u_{-1}\right)u_{1}du_{1}+2\int_{\d+\sqrt{3}b_{n}}^{1}\Dif_{u_{1}}h_{0u}\left(\tilde{u}_{1},u_{-1}\right)\left|\phi^{''}\left(\frac{u_{1}-\d}{b_{n}}\right)\right|u_{1}du_{1}\\
\leq\  & M\left(\d+\sqrt{3}b_{n}\right)^{2}+2M\int_{\sqrt{3}}^{b_{n}^{-1}\left(1-\d\right)}\left|\phi^{''}\left(\zeta_{1}\right)\right|\left(b_{n}\zeta_{1}+\d\right)d\left(b_{n}\zeta_{1}+\d\right)\\
=\  & M\left(\d+\sqrt{3}b_{n}\right)^{2}+2Mb_{n}^{2}\int_{\sqrt{3}}^{\infty}\left|\phi^{''}\left(\zeta_{1}\right)\right|\zeta_{1}d\zeta_{1}+2b_{n}\d\int_{\sqrt{3}}^{\infty}\left|\phi^{''}\left(\zeta_{1}\right)\right|d\zeta_{1}\\
\leq\  & M^{'}\left(b_{n}^{2}+\d^{2}\right)
\end{align*}
and hence
\[
\left|A_{n,2,2}\right|\leq M^{'}\frac{\d}{b_{n}^{3}}\left(b_{n}^{2}+\d^{2}\right)=M^{'}b_{n}^{-1}\d\left(1+b_{n}^{-2}\d^{-2}\right).
\]
Combining $A_{n,2,1}$and $A_{n,2,2}$ we have
\begin{equation}
A_{n,2}=-A_{2}+o\left(1\right)+O\left(b_{n}^{-1}\d\left(1+b_{n}^{-2}\d^{-2}\right)\right)\label{eq:A_n2}
\end{equation}
We will show that $O\left(b_{n}^{-1}\d\left(1+b_{n}^{-2}\d^{-2}\right)\right)$
is irrelevant later.

Lastly, consider $A_{n,3}$ corresponding to \eqref{eq:A_3}:
\begin{align*}
A_{n,3} & =\frac{1}{2}\int h_{0}\left(x\right)\phi\left(\frac{\xi\left(x\right)}{b_{n}}\right)\frac{x^{'}\t_{0}}{b_{n}}p_{x}dx.\\
 & =\frac{1}{2}\int h_{0}\left(x\right)\phi\left(\frac{x^{'}\t_{0}}{b_{n}}\right)\frac{x^{'}\t_{0}}{b_{n}}p_{x}dx\\
 & \quad+\frac{1}{2}\int h_{0}\left(x\right)\left[\phi\left(\frac{\xi\left(x\right)}{b_{n}}\right)-\phi\left(\frac{x^{'}\t_{0}}{b_{n}}\right)\right]\frac{x^{'}\t_{0}}{b_{n}}p_{x}dx\\
 & =:A_{n,3,1}+A_{n,3,2}
\end{align*}
For $A_{n,3,1}$, we have
\begin{align*}
A_{n,3,1}= & \frac{1}{2}\int h_{0}\left(x\right)\phi\left(\frac{x^{'}\t_{0}}{b_{n}}\right)\frac{x^{'}\t_{0}}{b_{n}}p_{x}dx\\
= & \frac{1}{2}\int\Dif_{u_{1}}h_{0u}\left(\tilde{u}_{1},u_{-1}\right)u_{1}\phi\left(\frac{u_{1}}{b_{n}}\right)\frac{u_{1}}{b_{n}}p_{x}du_{1}du_{-1}\\
= & \frac{1}{2}\int\Dif_{u_{1}}h_{0u}\left(b_{n}\tilde{\zeta}_{1},u_{-1}\right)b_{n}\zeta_{1}\phi\left(\zeta_{1}\right)\zeta_{1}p_{x}b_{n}d\zeta_{1}du_{-1}
\end{align*}
so that 
\[
b_{n}^{-2}A_{n,3,1}\to\frac{1}{2}\int\Dif_{u_{1}}h_{0u}\left(0,u_{-1}\right)\zeta_{1}^{2}\phi\left(\zeta_{1}\right)d\zeta_{1}\int_{u_{1}=0}p_{x}du_{-1}:=A_{3}
\]
For $A_{n,3,2}$, writing $\d=\norm{\t-\t_{0}}$, we have
\begin{align*}
\left|A_{n,3,2}\right| & \leq\frac{1}{2}\int\left|h_{0}\left(x\right)\right|\left|\phi^{'}\left(\frac{\tilde{\xi}\left(x\right)}{b_{n}}\right)\right|\frac{\left|x^{'}\t-x^{'}\t_{0}\right|}{b_{n}}\frac{\left|x^{'}\t_{0}\right|}{b_{n}}p_{x}dx\\
 & \leq\frac{\d}{2b_{n}^{2}}\int\left|h_{0}\left(x\right)\right|\ol{\phi^{'}}_{n,\d}\left(x^{'}\t_{0}\right)\left|x^{'}\t_{0}\right|p_{x}dx
\end{align*}
with 
\begin{align*}
\ol{\phi^{'}}_{n,\d}\left(x^{'}\t_{0}\right): & =e^{-\frac{1}{2}}\ind\left\{ \frac{\left|x^{'}\t_{0}\right|-\d}{b_{n}}\leq1\right\} +\left|\phi^{'}\left(\frac{\left|x^{'}\t_{0}\right|-\d}{b_{n}}\right)\right|\ind\left\{ \frac{\left|x^{'}\t_{0}\right|-\d}{b_{n}}>1\right\} 
\end{align*}
since $\left|\phi^{'}\left(t\right)\right|\leq\phi^{'}\left(1\right)=e^{-\frac{1}{2}}$
and $\left|\phi^{'}\left(t\right)\right|$ is increasing in $\left|t\right|$
for $0<\left|t\right|<1$ and then decreasing in $\left|t\right|$
for $\left|t\right|>1$. Then,
\begin{align*}
\left|A_{n,3,2}\right| & \leq\frac{\d}{2b_{n}^{2}}\int\Dif_{u_{1}}h_{0u}\left(\tilde{u}_{1},u_{-1}\right)u_{1}^{2}\ol{\phi^{'}}_{n,\d}\left(u_{1}\right)du_{1}p_{x}du_{-1}\\
 & \leq\frac{\d}{2b_{n}}M\int\ind\left\{ \left|u_{1}\right|\leq b_{n}+\d\right\} u_{1}^{2}du_{1}p_{x}du_{-1}\\
 & \quad+\frac{\d}{2b_{n}}M\int\phi^{'}\left(\frac{u_{1}-\d}{b_{n}}\right)\ind\left\{ \left|u_{1}\right|>b_{n}+\d\right\} u_{1}^{2}du_{1}p_{x}du_{-1}\\
 & =\frac{\d}{b_{n}}M\int\int_{0}^{b_{n}+\d}u_{1}^{2}du_{1}p_{x}du_{-1}+\frac{\d}{b_{n}}M\int\int_{b_{n}+\d}^{1}\left|\phi^{'}\left(\frac{u_{1}-\d}{b_{n}}\right)\right|u_{1}^{2}du_{1}p_{x}du_{-1}\\
 & \leq\frac{\d}{b_{n}}M\left(b_{n}+\d\right)^{3}\int p_{x}du_{-1}+\frac{\d}{b_{n}}b_{n}^{3}M\int\int_{1}^{b_{n}^{-1}\left(1-\d\right)}\left|\phi^{'}\left(\zeta_{1}\right)\right|\left(b_{n}\zeta_{1}+\d\right)^{2}d\zeta_{1}p_{x}du_{-1}\\
 & =M^{'}\left(b_{n}+\d\right)^{3}+\d M\int\int_{1}^{\infty}\left|\phi^{'}\left(\zeta_{1}\right)\right|\left(b_{n}^{2}\zeta_{1}^{2}+2b_{n}\d\zeta_{1}+\d^{2}\right)d\zeta_{1}p_{x}du_{-1}\\
 & \leq M^{''}\left[\left(b_{n}+\d\right)^{3}+\d\left(b_{n}+\d\right)^{2}\right]\\
 & =M^{'''}\left(b_{n}+\d\right)^{3}
\end{align*}
Combing $A_{n,3,1}$ and $A_{n,3,2}$ we have
\begin{equation}
A_{n,3}=A_{n,3,1}+A_{n,3,2}=A_{3}b_{n}^{2}+o\left(b_{n}^{2}\right)+O\left(\left(b_{n}+\d\right)^{3}\right).\label{eq:A_n3}
\end{equation}

Plugging the results in \eqref{eq:A_n1}\eqref{eq:A_n2}\eqref{eq:A_n3}
about $A_{n,1},A_{n,2},A_{n,3}$ into \eqref{eq:A_ALL}, we deduce,
with $\d:=\norm{\t-\t_{0}}$,
\begin{align}
 & P\left(\psi_{n,\t}\left(z\right)-\psi_{n,\t_{0}}\left(z\right)\right)\nonumber \\
=\  & A_{n,1}\left(\t-\t_{0}\right)+\left(\t-\t_{0}\right)^{'}A_{n,2}\left(\t-\t_{0}\right)+A_{n,3}\left(\t-\t_{0}\right)^{'}\left(I_{d}-\t_{0}\t_{0}^{'}\right)\left(\t-\t_{0}\right)\\
=\  & b_{n}^{2}A_{1}\left(\t-\t_{0}\right)+o\left(\d b_{n}^{2}\right)\nonumber \\
 & -\left(\t-\t_{0}\right)^{'}V\left(\t-\t_{0}\right)+o\left(\d^{2}\right)+O\left(b_{n}^{-1}\d^{3}\left(1+b_{n}^{-2}\d^{-2}\right)\right)\nonumber \\
 & +A_{3}b_{n}^{2}\d^{2}+o\left(b_{n}^{2}\d^{2}\right)+O\left(\d^{2}\left(b_{n}+\d\right)^{3}\right)\label{eq:P_psi_dif}\\
=\  & -\left(\t-\t_{0}\right)^{'}V\left(\t-\t_{0}\right)+b_{n}^{2}A_{1}\left(\t-\t_{0}\right)+o\left(\d^{2}\right)+o\left(b_{n}^{2}\d\right)+O\left(b_{n}^{-1}\d^{3}\left(1+b_{n}^{-2}\d^{-2}\right)\right)\nonumber 
\end{align}
\end{proof}

\subsection{\label{subsec:pf_Thm_Rate}Proof of Theorem \ref{thm:Thm_Bin_Rate}}
\begin{proof}
For consistency, we observe that
\[
\sup_{\t\in\T}\sup_{h\in{\cal H}}\left|\P_{n}g_{\t,h}-Pg_{\t,h}\right|=o_{p}\left(1\right).
\]
since $\cG$ is Gilvenko-Cantelli given Lemma \ref{lem:J_uni_finite}.
Moreover, 
\begin{align*}
\sup_{\t\in\T}\sup_{\norm{h-h_{0}}_{\infty}\leq\e}\left|Pg_{\t,h}-Pg_{\t,h_{0}}\right| & \leq P\left(\left|h-h_{0}\right|\right)\leq\e\to0\quad\text{as }\d\to0.
\end{align*}
As $\norm{\hat{h}-h_{0}}_{\infty}=o_{p}\left(1\right)$ and $\hat{h}\in{\cal H}$
with probability approaching 1 by Assumption \ref{assu:FirstStageRate},
we conclude by Theorem 1 of \citet[DvK thereafter]{delsol2020semiparametric}
that $\norm{\hat{\t}-\t_{0}}=o_{p}\left(1\right).$

For the rate of convergence, we apply Theorem 2 of DvK by verifying
their Conditions B1-B4.

B1 directly follows from the consistency of $\hat{\t}$ and the assumption
that$\norm{\hat{h}-h_{0}}_{\infty}=O_{p}\left(a_{n}\right)$.

For their Condition B2, observe that
\[
\GG_{n}\left(g_{\t,h}-g_{\t_{0},h}\right)=\GG_{n}\left(g_{\t,h_{0}}-g_{\t_{0},h_{0}}\right)+\GG_{n}\left(g_{\t,h}-g_{\t_{0},h}-g_{\t,h_{0}}+g_{\t_{0},h_{0}}\right)
\]
and thus, by \eqref{lem:Term1}and \eqref{lem:Term2}, 
\[
P\sup_{\norm{\t-\t_{0}}\leq\d,\norm{h-h_{0}}_{\infty}\leq Ka_{n}}\left|\GG_{n}\left(g_{\t,h}-g_{\t_{0},h}\right)\right|\leq M_{1}\d^{\frac{3}{2}}+M_{2}a_{n}\sqrt{\d}.
\]
so that $\Phi_{n}\left(\d\right)=\d^{\frac{3}{2}}+a_{n}\sqrt{\d}$
in the notation of DvK.

By Lemma \eqref{lem:Term_3}(i), for any $M<\infty$, we have
\begin{align*}
 & \P\left(\GG_{n}\left(\psi_{n,\t}-\psi_{n,\t_{0}}\right)>Mb_{n}^{-1}\left(b_{n}+\norm{\t-\t_{0}}\right)^{\frac{1}{2}}\norm{\t-\t_{0}}\right)\\
\leq\  & \P\left(\sup_{\norm{\t-\t_{0}}\leq\d,\norm{h-h_{0}}_{\infty}\leq Ka_{n}}\left|\GG_{n}\left(\psi_{n,\t}-\psi_{n,\t_{0}}\right)\right|>Mb_{n}^{-1}\left(b_{n}+\norm{\t-\t_{0}}\right)^{\frac{1}{2}}\norm{\t-\t_{0}}\right)\\
\leq\  & \frac{P\sup_{\norm{\t-\t_{0}}\leq\d,\norm{h-h_{0}}_{\infty}\leq Ka_{n}}\left|\GG_{n}\left(\psi_{n,\t}-\psi_{n,\t_{0}}\right)\right|}{Mb_{n}^{-1}\left(b_{n}+\norm{\t-\t_{0}}\right)^{\frac{1}{2}}\norm{\t-\t_{0}}}\quad\text{ by Markov Inequality},\\
\leq\  & \frac{M_{3}b_{n}^{-1}\left(b_{n}+\norm{\t-\t_{0}}\right)^{\frac{1}{2}}\norm{\t-\t_{0}}}{Mb_{n}^{-1}\left(b_{n}+\norm{\t-\t_{0}}\right)^{\frac{1}{2}}\norm{\t-\t_{0}}}=\frac{M_{3}}{M}\to0\quad\text{as }M\to\infty.
\end{align*}
Hence, combining with \eqref{lem:Term_3}(ii), we have
\begin{align}
P\left(g_{\t,\hat{h}}-g_{\t_{0},\hat{h}}\right) & =\frac{1}{\sqrt{n}}\GG_{n}\left(\psi_{n,\t}-\psi_{n,\t_{0}}\right)+P\left(\psi_{n,\t}-\psi_{n,\t_{0}}\right),\nonumber \\
 & \leq R_{n}\frac{1}{\sqrt{n}}b_{n}^{-1}\left(b_{n}+\norm{\t-\t_{0}}\right)^{\frac{1}{2}}\norm{\t-\t_{0}}-C\norm{\t-\t_{0}}^{2}+M_{4}b_{n}^{2}\norm{\t-\t_{0}}\nonumber \\
 & \quad+M_{5}b_{n}^{-1}\norm{\t-\t_{0}}^{3}\left(1+b_{n}^{-2}\norm{\t-\t_{0}}^{-2}\right)\label{eq:App_p_dif}
\end{align}
with $R_{n}=O_{p}\left(1\right)$.

Letting $\norm{\hat{\t}-\t_{0}}:=O_{p}\left(\d_{n}\right)$, we seek
to find the smallest $\d_{n}$ that verifies Condition B3 and B4 in
DvK\footnote{$\d_{n}=r_{n}^{-1}$ in DvK's notation.}. First, we set
the bandwidth $b_{n}$ to be such that
\[
\frac{1}{\sqrt{nb_{n}}}=b_{n}^{2}\quad\iff b_{n}=n^{-\frac{1}{5}},
\]
which exactly corresponds to the optimal choice of bandwidth in \citet*{horowitz1992smoothed}.
This ensures that the second and the third terms in \eqref{eq:App_p_dif}
are of the same order of magnitude
\[
\frac{1}{\sqrt{n}}b_{n}^{-1}\d_{n}\left(\d_{n}+b_{n}\right)^{\frac{1}{2}}\sim b_{n}^{2}\d
\]
provided that $\d_{n}=o\left(b_{n}\right)$. Setting $\d_{n}\sim n^{-2/5}=o\left(b_{n}\right)$,
we see that
\[
b_{n}^{2}\sim\frac{1}{\sqrt{n}}b_{n}^{-1}\left(\d_{n}+b_{n}\right)^{\frac{1}{2}}\sim n^{-\frac{2}{5}}=O\left(\d_{n}\right),
\]
and moreover $b_{n}^{-1}\d_{n}^{3}\left(1+b_{n}^{-2}\d_{n}^{-2}\right)=o\left(1\right)\d_{n}^{2}.$
Hence, Condition B3 of DvK is verified. Lastly, for Condition B4,
we see that
\begin{align*}
\frac{1}{\d_{n}^{2}}\Phi_{n}\left(\d_{n}\right) & =\frac{1}{\d_{n}^{2}}\left(\d_{n}^{\frac{3}{2}}+a_{n}\sqrt{\d_{n}}\right)=\left(\d_{n}^{-\frac{1}{2}}+a_{n}\d_{n}^{-\frac{3}{2}}\right)\sim n^{\frac{1}{5}}+a_{n}n^{\frac{3}{5}},
\end{align*}
which is $O\left(\sqrt{n}\right)$ provided that $a_{n}=O\left(n^{-1/10}\right)$.
Since $a_{n}=\left(nb_{n}^{d}/\log n\right)^{-\frac{1}{2}}+b_{n}^{2}$
for the Nadaraya-Watson estimator, with $b_{n}\sim n^{-\frac{1}{5}}$
we have
\[
a_{n}=n^{-\frac{1}{2}+\frac{d}{10}}\sqrt{\log n}=O_{p}\left(n^{-\frac{1}{10}}\right)\quad\iff\quad d<4.
\]

Hence, for $d<4$, the impact of the first-stage estimation through
$a_{n}$ is negligible with $b_{n}\sim n^{-\frac{1}{5}}$, and thus
\[
\norm{\hat{\t}-\t_{0}}=O_{p}\left(n^{-2/5}\right).
\]

For $d\geq4$, the $n^{-2/5}$-rate is unattainable due to the higher
dimensionality ($d$) of the first-stage kernel regression. Optimally,
we set $b_{n}$ so as to minimize
\begin{equation}
\max\left\{ n^{-\frac{1}{3}}\left(nb_{n}^{d}/\log n\right)^{-\frac{1}{2}\cd\frac{2}{3}},\ b_{n}^{2},\ \left(nb_{n}\right)^{-\frac{1}{2}}\right\} ,\label{eq:bin_3_rates}
\end{equation}
which is solved by setting $b_{n}^{2}\sim n^{-\frac{1}{3}}\left(nb_{n}^{d}/\log n\right)^{-\frac{1}{2}\cd\frac{2}{3}}$
(up to the $\log n$ factor) with
\[
b_{n}\sim n^{-\frac{2}{d+6}}
\]
giving an optimal rate of convergence at 
\[
\d_{n}=n^{-\frac{4}{d+6}}\left(\log n\right)^{\frac{1}{3}},
\]
provided that the first-stage estimator $\hat{h}$ is still consistent
with $a_{n}=\left(nb_{n}^{d}/\log n\right)^{-1/2}\to0$, or 
\[
b_{n}\sim n^{-\frac{2}{d+6}}>>n^{-\frac{1}{d}},
\]
which is possible if $d<6$.

For $d\geq6$, $b_{n}^{2}$ becomes the dominant term in \eqref{eq:bin_3_rates},
which should be minimized subject to the constraint $a_{n}=\left(nb_{n}^{d}/\log n\right)^{-1/2}\to0$.
This can be roughly achieved by setting, say, $b_{n}\sim\left(n^{-1}\log^{2}n\right)^{\frac{1}{d}}$,
in which case $a_{n}=1/\log n\to0$ and 
\[
\norm{\hat{\t}-\t_{0}}=O_{p}\left(b_{n}^{2}\right)=n^{-\frac{2}{d}}\left(\log n\right)^{\frac{4}{d}}.
\]
\end{proof}

\subsection{Proof of Theorem \ref{thm:Thm_Bin_Dist}(i)}
\begin{proof}
For $d<4$, define $\mathbb{M}_{n}\left(\t\right):=\P_{n}g_{\t,\hat{h}}$
and $\mathbb{M}\left(\t\right):=-\left(\t-\t_{0}\right)^{'}V\left(\t-\t_{0}\right)$
so that
\begin{align*}
 & \d_{n}^{-1}\left[\left(\mathbb{M}_{n}\left(\tilde{\t}_{n}\right)-\mathbb{M}\left(\tilde{\t}_{n}\right)\right)-\left(\mathbb{M}_{n}\left(\t_{0}\right)-\mathbb{M}\left(\t_{0}\right)\right)\right]\\
= & \frac{1}{\sqrt{n}\d_{n}}\GG_{n}\left(g_{\tilde{\t}_{n},\hat{h}}-g_{\t_{0},\hat{h}}\right)+\frac{1}{\d_{n}}\left[P\left(g_{\tilde{\t}_{n},\hat{h}}-g_{\t_{0},\hat{h}}\right)-\mathbb{M}\left(\t\right)\right]\\
=: & B_{n,1}+B_{n,2}
\end{align*}
for any $\tilde{\t}_{n}$ s.t. $\norm{\tilde{\t}_{n}-\t_{0}}=O_{p}\left(\d_{n}\right)=O_{p}\left(n^{-2/5}\right)$.
With the optimal choice of bandwidth $b_{n}^{-1/5}$, we know $a_{n}=n^{-\frac{1}{2}+\frac{d}{10}}\sqrt{\log n}=o\left(n^{-\frac{1}{10}}\right)$
and thus by Lemma \ref{lem:Term1} and \ref{lem:Term2}, we have
\begin{align*}
 & P\sup_{\norm{\hat{h}-h_{0}}\leq Ka_{n}}\frac{1}{\sqrt{n}\d_{n}}\left|\GG_{n}\left(g_{\tilde{\t}_{n},\hat{h}}-g_{\t_{0},\hat{h}}\right)\right|\\
\leq\  & M\frac{1}{\sqrt{n}\d_{n}}\left(\d_{n}\sqrt{\d_{n}}+a_{n}\sqrt{\d_{n}}\right)=O\left(n^{-\frac{1}{2}}\d_{n}+n^{-\frac{1}{2}}a_{n}\d_{n}^{-\frac{1}{2}}\right)\\
=\  & o\left(\d_{n}\right)+o\left(n^{-\frac{1}{2}}n^{-\frac{1}{10}}\left(n^{-\frac{2}{5}}\right)^{-\frac{3}{2}}\right)\d_{n}=o\left(\d_{n}\right)+o\left(1\right)\d_{n}=o\left(\d_{n}\right)
\end{align*}
Hence, 
\[
B_{n,1}=o_{p}\left(\d_{n}\right).
\]

Now, recall that
\begin{align*}
B_{n,2} & =\frac{1}{\d_{n}}\left[P\left(g_{\tilde{\t}_{n},\hat{h}}-g_{\t_{0},\hat{h}}\right)-\mathbb{M}\left(\t\right)\right]\\
 & =\frac{1}{\sqrt{n}\d_{n}}\GG_{n}\left(\psi_{n,\tilde{\t}_{n}}-\psi_{n,\t_{0}}\right)+\frac{1}{\d_{n}}\left[P\left(\psi_{n,\tilde{\t}_{n}}-\psi_{n,\t_{0}}\right)-\mathbb{M}\left(\t\right)\right]\\
 & =:B_{n,2,1}+B_{n,2,2}
\end{align*}

First, we analyze $B_{n,2,1}$:
\begin{align*}
B_{n,2,1}=\  & \frac{1}{\sqrt{n}\d_{n}}\GG_{n}\left(\psi_{n,\tilde{\t}_{n}}-\psi_{n,\t_{0}}\right)\\
=\  & \frac{1}{n\d_{n}}\sum_{i=1}^{n}\left(\psi_{n,\tilde{\t}_{n}}\left(Z_{i}\right)-\psi_{n,\t_{0}}\left(Z_{i}\right)-P\left(\psi_{n,\tilde{\t}_{n}}-\psi_{n,\t_{0}}\right)\right)\\
=\  & \frac{1}{n\d_{n}}\sum_{i=1}^{n}\left[\left(y_{i}-\frac{1}{2}\right)\phi\left(\frac{X_{i}^{'}\t_{0}}{b_{n}}\right)\frac{X_{i}^{'}}{b_{n}}\left(I-\t_{0}\t_{0}^{'}\right)-A_{n,1}\right]\left(I-\t_{0}\t_{0}^{'}\right)\left(\tilde{\t}_{n}-\t_{0}\right)+R_{n,\t}\\
=\  & Z_{n}^{'}\left(I-\t_{0}\t_{0}^{'}\right)\left(\tilde{\t}_{n}-\t_{0}\right)+R_{n,\t}
\end{align*}
with 
\begin{align*}
Z_{n}^{'} & :=\frac{1}{n\d_{n}}\sum_{i=1}^{n}\left[\left(y_{i}-\frac{1}{2}\right)\phi\left(\frac{X_{i}^{'}\t_{0}}{b_{n}}\right)\frac{X_{i}^{'}}{b_{n}}\left(I-\t_{0}\t_{0}^{'}\right)-A_{n,1}\right]\\
 & =\frac{1}{n\d_{n}}\sum_{i=1}^{n}\left[\left(y_{i}-\frac{1}{2}\right)\phi\left(\frac{X_{i}^{'}\t_{0}}{b_{n}}\right)\frac{X_{i}^{'}}{b_{n}}\left(I-\t_{0}\t_{0}^{'}\right)-A_{n,1}\right]
\end{align*}
and 
\begin{align*}
R_{n,\t}:= & \left(\tilde{\t}_{n}-\t_{0}\right)^{'}\frac{1}{n\d_{n}}\sum_{i=1}^{n}\left[\frac{1}{2}\left(y_{i}-\frac{1}{2}\right)\left(I_{d}-\t_{0}\t_{0}^{'}\right)\phi^{'}\left(\frac{\xi\left(X_{i}\right)}{b_{n}}\right)\cd\frac{X_{i}X_{i}^{'}}{b_{n}^{2}}\left(I_{d}-\t_{0}\t_{0}^{'}\right)-A_{n,2}\right]\left(\tilde{\t}_{n}-\t_{0}\right)\\
 & -\frac{1}{n\d_{n}}\sum_{i=1}^{n}\left[\frac{1}{2}\left(y_{i}-\frac{1}{2}\right)\phi\left(\frac{\xi\left(X_{i}\right)}{b_{n}}\right)\frac{X_{i}^{'}}{b_{n}}\t_{0}-A_{n,3}\right]\cd\left(\tilde{\t}_{n}-\t_{0}\right)^{'}\left(I_{d}-\t_{0}\t_{0}^{'}\right)\left(\tilde{\t}_{n}-\t_{0}\right)
\end{align*}
Now, since $\E\left[Z_{n}\right]={\bf 0}$ and 
\begin{align*}
\E\left[Z_{n}Z_{n}^{'}\right] & =\frac{1}{n\d_{n}^{2}}\int\phi^{2}\left(\frac{x^{'}\t_{0}}{b_{n}}\right)\left(I-\t_{0}\t_{0}^{'}\right)\frac{xx^{'}}{b_{n}^{2}}\left(I-\t_{0}\t_{0}^{'}\right)p_{x}dx\\
 & =\frac{1}{nb_{n}^{2}\d_{n}^{2}}\int\phi^{2}\left(\frac{x^{'}\t_{0}}{b_{n}}\right)\left(I-\t_{0}\t_{0}^{'}\right)xx^{'}\left(I-\t_{0}\t_{0}^{'}\right)p_{x}dx\\
 & =\frac{1}{nb_{n}\d_{n}^{2}}\int\phi^{2}\left(\zeta_{1}\right)\left(I-\t_{0}\t_{0}^{'}\right)T_{\t_{0}}\ol u_{-1}\ol u_{-1}^{'}T_{\t_{0}}^{'}\left(I-\t_{0}\t_{0}^{'}\right)p_{x}d\zeta du_{-1}\\
 & =\int\phi^{2}\left(\zeta_{1}\right)d\zeta_{1}\left(I-\t_{0}\t_{0}^{'}\right)T_{\t_{0}}\ol u_{-1}\ol u_{-1}^{'}T_{\t_{0}}^{'}\left(I-\t_{0}\t_{0}^{'}\right)p_{x}du_{-1}\\
 & =O\left(1\right)
\end{align*}
so $Z_{n}=O_{p}\left(1\right)$. Furthermore, the Lindberg condition
can be verified as
\begin{align*}
 & \frac{1}{n\d_{n}^{2}}\int\phi^{2}\left(\frac{x^{'}\t_{0}}{b_{n}}\right)\left(I-\t_{0}\t_{0}^{'}\right)\frac{xx^{'}}{b_{n}^{2}}\left(I-\t_{0}\t_{0}^{'}\right)\cd\ind\left\{ \frac{1}{n^{2}\d_{n}^{2}b_{n}^{2}}\phi^{2}\left(\frac{x^{'}\t_{0}}{b_{n}}\right)x^{'}\left(I-\t_{0}\t_{0}^{'}\right)x\geq\e^{2}\right\} p_{x}dx\\
\leq\  & \frac{1}{nb_{n}\d_{n}^{2}}\int\phi^{2}\left(\zeta_{1}\right)\left(I-\t_{0}\t_{0}^{'}\right)T_{\t_{0}}\ol u_{-1}\ol u_{-1}^{'}T_{\t_{0}}^{'}\left(I-\t_{0}\t_{0}^{'}\right)\cd\ind\left\{ \frac{1}{n\d_{n}b_{n}}\phi\left(\zeta_{1}\right)\geq\e\right\} p_{x}d\zeta_{1}du_{-1}\\
=\  & \int\phi^{2}\left(\zeta_{1}\right)\left(I-\t_{0}\t_{0}^{'}\right)T_{\t_{0}}\ol u_{-1}\ol u_{-1}^{'}T_{\t_{0}}^{'}\left(I-\t_{0}\t_{0}^{'}\right)\cd\ind\left\{ \d_{n}\phi\left(\zeta_{1}\right)\geq\e\right\} p_{x}d\zeta_{1}du_{-1}\\
\to\  & {\bf 0}
\end{align*}
for every $\e>0$ as $n\to\infty$. Hence, by the triangular-array
CLT, we have
\begin{equation}
Z_{n}\dto\cN\left(0,\Sigma\right),\label{eq:Z_normal}
\end{equation}
where
\begin{align}
\Sigma & :=\left(I-\t_{0}\t_{0}^{'}\right)T_{\t_{0}}\left[\frac{1}{2\sqrt{\pi}}\int_{u_{1}=0}\ol u_{-1}\ol u_{-1}^{'}p_{x}du_{-1}\right]T_{\t_{0}}^{'}\left(I-\t_{0}\t_{0}^{'}\right).\nonumber \\
 & =\left(I-\t_{0}\t_{0}^{'}\right)T_{\t_{0}}\left[\frac{1}{2\sqrt{\pi}}\O_{\ol u_{-1}}\right]T_{\t_{0}}^{'}\left(I-\t_{0}\t_{0}^{'}\right)\label{eq:Sigma}
\end{align}
where
\[
\O_{\ol u_{-1}}:=\int_{u_{1}=0}\ol u_{-1}\ol u_{-1}^{'}p_{x}du_{-1}.
\]
Similarly, we can deduce
\[
\norm{R_{n,\t}}=O_{p}\left(\frac{1}{\sqrt{n\d_{n}^{2}b_{n}^{3}}}\right)\norm{\tilde{\t}_{n}-\t_{0}}^{2}=o_{p}\left(\frac{1}{\d_{n}}\norm{\tilde{\t}_{n}-\t_{0}}^{2}\right).
\]
Hence
\[
B_{n,2,1}=Z_{n}^{'}\left(\tilde{\t}_{n}-\t_{0}\right)+o_{p}\left(\frac{1}{\d_{n}}\norm{\tilde{\t}_{n}-\t_{0}}^{2}\right).
\]
Now, by \eqref{eq:P_psi_dif} and the observation that $A_{1}=A_{1}\left(I-\t_{0}\t_{0}^{'}\right)$,
\begin{align*}
P\left(\psi_{n,\tilde{\t}_{n}}\left(z\right)-\psi_{n,\t_{0}}\left(z\right)\right) & =b_{n}^{2}A_{1}\left(I-\t_{0}\t_{0}^{'}\right)\left(\tilde{\t}_{n}-\t_{0}\right)-\left(\tilde{\t}_{n}-\t_{0}\right)^{'}V\left(\tilde{\t}_{n}-\t_{0}\right)+o\left(b_{n}^{2}\norm{\tilde{\t}_{n}-\t_{0}}\right)
\end{align*}
and hence

\begin{align*}
B_{n,2,2}=\  & \frac{1}{\d_{n}}\left[P\left(\psi_{n,\tilde{\t}_{n}}-\psi_{n,\t_{0}}\right)-\mathbb{M}\left(\t\right)\right]=\frac{1}{\d_{n}}\left[b_{n}^{2}A_{1}\left(\tilde{\t}_{n}-\t_{0}\right)+o\left(b_{n}^{2}\right)\right]\\
=\  & A_{1}\left(I-\t_{0}\t_{0}^{'}\right)\left(\tilde{\t}_{n}-\t_{0}\right)+o\left(\norm{\tilde{\t}_{n}-\t_{0}}\right)
\end{align*}
Combining $B_{n,1}$, $B_{n,2,1}$ and $B_{n,2,2}$ we have
\begin{align*}
 & \d_{n}^{-1}\left[\left(\mathbb{M}_{n}\left(\tilde{\t}_{n}\right)-\mathbb{M}\left(\tilde{\t}_{n}\right)\right)-\left(\mathbb{M}_{n}\left(\t_{0}\right)-\mathbb{M}\left(\t_{0}\right)\right)\right]\\
=\  & o_{p}\left(\d_{n}\right)+Z_{n}^{'}\left(\tilde{\t}_{n}-\t_{0}\right)+o_{p}\left(\frac{1}{\d_{n}}\norm{\tilde{\t}_{n}-\t_{0}}^{2}\right)+A_{1}\left(\tilde{\t}_{n}-\t_{0}\right)+o\left(\norm{\tilde{\t}_{n}-\t_{0}}\right)\\
=\  & \left(Z_{n}^{'}+A_{1}\right)\left(I-\t_{0}\t_{0}^{'}\right)\left(\tilde{\t}_{n}-\t_{0}\right)+o_{p}\left(\norm{\tilde{\t}_{n}-\t_{0}}+\frac{1}{\d_{n}}\norm{\tilde{\t}_{n}-\t_{0}}^{2}+\d_{n}\right)\\
=\  & \left(Z_{n}^{'}+A_{1}\right)T_{\t_{0}}T_{\t_{0}}^{'}\left(I-\t_{0}\t_{0}^{'}\right)\left(\tilde{\t}_{n}-\t_{0}\right)+o_{p}\left(\norm{\tilde{\t}_{n}-\t_{0}}+\frac{1}{\d_{n}}\norm{\tilde{\t}_{n}-\t_{0}}^{2}+\d_{n}\right)
\end{align*}
All conditions in VW Theorem 3.2.16 are now satisfied with $V_{u_{-1}}\in\R^{\left(d-1\right)\times\left(d-1\right)}$
being nonsingular and invertible, where $V_{u_{-1}}$ is defined in
\eqref{eq:V_tang} with the projection onto the tangent space of $\S^{d-1}$
via $\left(I-\t_{0}\t_{0}^{'}\right)$ and the change of coordinates
via $T_{\t_{0}}^{'}$. Specifically, noting that 
\begin{align*}
\Sigma & =\left(I-\t_{0}\t_{0}^{'}\right)T_{\t_{0}}\left[\frac{1}{2\sqrt{\pi}}\O_{\ol u_{-1}}\right]T_{\t_{0}}^{'}\left(I-\t_{0}\t_{0}^{'}\right)\\
V & =\left(I-\t_{0}\t_{0}^{'}\right)T_{\t_{0}}V_{\ol u_{-1}}T_{\t_{0}}^{'}\left(I-\t_{0}\t_{0}^{'}\right)=T_{\t_{0}}V_{\ol u_{-1}}T_{\t_{0}}^{'}
\end{align*}
and writing $A_{\ol u_{-1}}\equiv\left(0,A_{u_{-1}}\right):=f^{'}\left(0\right)\cd\int_{u_{1}=0}\ol u_{-1}p_{x}du_{-1}$
so that
\[
A_{1}=T_{\t_{0}}A_{\ol u_{-1}}
\]
we have
\begin{align*}
V^{-}\Sigma V^{-} & =\frac{1}{2\sqrt{\pi}}T_{\t_{0}}\left(\begin{array}{cc}
0 & {\bf 0}\\
{\bf 0} & V_{u_{-1}}^{-1}\O_{\ol u_{-1}}V_{u_{-1}}^{-1}
\end{array}\right)T_{\t_{0}}^{'}=\frac{1}{2\sqrt{\pi}}T_{\t_{0}}V_{\ol u_{-1}}^{-}\O_{\ol u_{-1}}V_{\ol u_{-1}}^{-}T_{\t_{0}}^{'}
\end{align*}
and
\begin{align*}
V^{-}A_{1} & =T_{\t_{0}}\left(\begin{array}{c}
0\\
V_{u_{-1}}^{-1}A_{u_{-1}}
\end{array}\right)=T_{\t_{0}}V_{\ol u_{-1}}^{-}A_{\ol u_{-1}}
\end{align*}
Hence, by VW Theorem 3.2.16, we have
\begin{align*}
\d_{n}^{-1}T_{\t_{0}}^{'}\left(I-\t_{0}\t_{0}^{'}\right)\left(\hat{\t}-\t_{0}\right) & =V_{\ol u_{-1}}^{-1}\left(T_{\t_{0}}^{'}Z_{n}+A_{\ol u_{-1}}\right)+o_{p}\left(1\right)\\
 & \dto\cN\left(\left(\begin{array}{c}
0\\
V_{u_{-1}}^{-1}A_{u_{-1}}
\end{array}\right),\ \left(\begin{array}{cc}
0 & {\bf 0}^{'}\\
{\bf 0} & \frac{1}{2\sqrt{\pi}}V_{u_{-1}}^{-1}\O_{\ol u_{-1}}V_{u_{-1}}^{-1}
\end{array}\right)\right)
\end{align*}
and
\[
\d_{n}^{-1}\left(I-\t_{0}\t_{0}^{'}\right)\left(\hat{\t}-\t_{0}\right)\dto\cN\left(T_{\t_{0}}V_{\ol u_{-1}}^{-}A_{\ol u_{-1}},\ \frac{1}{2\sqrt{\pi}}T_{\t_{0}}V_{\ol u_{-1}}^{-}\O_{\ol u_{-1}}V_{\ol u_{-1}}^{-}T_{\t_{0}}^{'}\right).
\]
\end{proof}

\subsection{Proof of Theorem \ref{thm:Thm_Bin_Dist}(ii)}
\begin{proof}
For $4\leq d<6$, we set $b_{n}\sim n^{-\frac{2}{d+6}}$ so that $\d_{n}=n^{-\frac{4}{d+6}}\left(\log n\right)^{\frac{1}{3}}$
and $a_{n}=n^{-\frac{6-d}{2\left(d+6\right)}}\sqrt{\log n}$. In particular,
\begin{equation}
\d_{n}\sim\left(n^{2}b_{n}^{d}/\log n\right)^{-\frac{1}{3}}\sim n^{-\frac{1}{3}}a_{n}^{\frac{2}{3}}.\label{eq:dn_n_an}
\end{equation}
Now, consider the scaled process indexed by any $s$ in the tangent
space of $\S^{d-1}$ at $\t_{0}$:
\begin{align}
 & \frac{1}{\sqrt{n}\d_{n}^{2}}\GG_{n}\left(g_{\t_{0}+s\d_{n},\hat{h}}-g_{\t_{0},\hat{h}}\right)\nonumber \\
=\  & \frac{1}{\sqrt{n}\d_{n}^{2}}\GG_{n}\left(g_{\t_{0}+s\d_{n},\hat{h}}-g_{\t_{0},\hat{h}}-g_{+s\d_{n},h_{0}}+g_{\t_{0},h_{0}}\right)\nonumber \\
 & +\frac{1}{\sqrt{n}\d_{n}^{2}}\GG_{n}\left(g_{\t_{0}+s\d_{n},h_{0}}-g_{\t_{0},h_{0}}\right)+\frac{1}{\d_{n}^{2}}P\left(g_{\t_{0}+s\d_{n},\hat{h}}-g_{\t_{0},\hat{h}}\right)\label{eq:D_decompose}\\
=\  & D_{n,1}+D_{n,2}+D_{n,3}\nonumber 
\end{align}

For $D_{n,1}$, we verify VW Condition 2.11.21 to apply their Theorem
2.11.23. Define
\begin{align*}
\g_{n,s} & :=n^{-\frac{1}{2}}\d_{n}^{-2}\left(g_{\t_{0}+s\d_{n},\hat{h}}-g_{\t_{0},\hat{h}}-g_{\t_{0}+s\d_{n},h_{0}}+g_{\t_{0},h_{0}}\right)\\
{\cal G}_{2,n} & :=\left\{ \g_{n,s}:\ s^{'}\t_{0}=0,\ s\in\R^{d}\right\} 
\end{align*}
Similarly to the proof of Lemma \ref{lem:Term2}, we can show that
${\cal G}_{2,n}$ has an envelope function 
\[
G_{2,n}\left(x\right)=Kn^{-\frac{1}{2}}\d_{n}^{-2}a_{n}\ind\left\{ \left|x^{'}\t_{0}\right|\leq\norm x\d_{n}\right\} 
\]
with, by \eqref{eq:dn_n_an},
\begin{equation}
PG_{2,n}^{2}\leq Cn^{-1}\d_{n}^{-4}a_{n}^{2}\d_{n}=C\left(n^{-\frac{1}{3}}a_{n}^{\frac{2}{3}}\d_{n}^{-1}\right)^{3}=O\left(1\right).\label{eq:P_G2}
\end{equation}
Furthermore, since $\sqrt{n}\d_{n}\to\infty$,
\begin{align}
P\left[G_{2,n}^{2}\ind\left\{ G_{2,n}>\e\sqrt{n}\right\} \right] & \leq P\left[Kn^{-1}\d_{n}^{-4}a_{n}^{2}\ind\left\{ \left|x^{'}\t_{0}\right|\leq\norm x\d_{n}\right\} \ind\left\{ n^{-1}\d_{n}^{-4}a_{n}^{2}\geq\e\sqrt{n}\right\} \right]\nonumber \\
 & \leq Cn^{-1}a_{n}^{2}\d_{n}^{-3}\ind\left\{ n^{-1}a_{n}^{2}\d_{n}^{-3}\cd\d_{n}^{-1}\geq\e\sqrt{n}\right\} \leq C^{'}\ind\left\{ C^{'}\geq\e\sqrt{n}\d_{n}\right\} \nonumber \\
 & \to0\quad\text{as }n\to\infty\quad\text{for every }\e>0\label{eq:PG_Linid}
\end{align}
In addition, for any $s,t$,
\begin{align*}
\left|\g_{n,s}-\g_{n,t}\right| & =n^{-\frac{1}{2}}\d_{n}^{-2}\left|g_{\t_{0}+s\d_{n},h}-g_{\t_{0}+t\d_{n},h}-g_{\t_{0}+s\d_{n},h_{0}}+g_{\t_{0}+t\d_{n},h_{0}}\right|\\
 & =n^{-\frac{1}{2}}\d_{n}^{-2}\left|\hat{h}\left(x\right)-h_{0}\left(x\right)\right|\cd\left|\ind\left\{ x^{'}\left(\t_{0}+s\d_{n}\right)\geq0\right\} -\ind\left\{ x^{'}\left(\t_{0}+t\d_{n}\right)\geq0\right\} \right|\\
 & \leq Kn^{-\frac{1}{2}}\d_{n}^{-2}a_{n}\cd\left(\ind\left\{ \left|x^{'}\t_{0}+\frac{1}{2}\d_{n}x^{'}\left(s+t\right)\right|\leq\frac{1}{2}\d_{n}\left|x^{'}\left(s-t\right)\right|\right\} \right)
\end{align*}
and thus, for any $\e_{n}\to0$, we have
\begin{align}
\sup_{\norm{s-t}\leq\e_{n}} & P\left(\g_{n,s}-\g_{n,t}\right)^{2}\leq Kn^{-1}a_{n}^{2}\d_{n}^{-4}\cd C\d_{n}\e_{n}=C^{'}\e_{n}\to0.\label{eq:P_f_st}
\end{align}
VW Condition 2.11.21 is thus verified by \eqref{eq:P_G2}\eqref{eq:PG_Linid}
and \eqref{eq:P_f_st}. Lastly, since
\begin{align*}
\sqrt{\log\mathscr{N}_{[]}\left(\e\norm{G_{2,n}}_{L_{2}\left(P\right)},{\cal G}_{2,n},L_{2}\left(P\right)\right)} & \leq M\left(\e\norm{G_{2,n}}_{L_{2}\left(P\right)}\right)^{-\frac{d}{\left\lfloor d\right\rfloor +1}}\\
 & =\left(\frac{1}{n^{-1}\d_{n}^{-4}a_{n}^{2}\d_{n}}\right)^{-\frac{d}{\left\lfloor d\right\rfloor +1}}\e^{-\frac{d}{\left\lfloor d\right\rfloor +1}}\leq C\e^{-\frac{d}{\left\lfloor d\right\rfloor +1}}.
\end{align*}
and thus
\begin{align*}
\int_{0}^{\e_{n}}\sqrt{\log\mathscr{N}_{[]}\left(\e\norm{G_{2,n}}_{L_{2}\left(P\right)},{\cal G}_{2,n},L_{2}\left(P\right)\right)}d\e & \leq C\e_{n}^{\frac{d}{\left\lfloor d\right\rfloor +1}}\to0.
\end{align*}
By VW Theorem 2.11.23, the sequence
\[
\left\{ \GG_{n}\g_{n,s}:s^{'}\t_{0}=0,\ s\in\R^{d}\right\} 
\]
is asymptotically tight in $l^{\infty}\left(\R^{d}\cap\t_{0}^{\indep}\right)$
and converges in distribution to a Gaussian process $G$ with the
covariance function
\[
H\left(s,t\right):=\lim_{n\to\infty}\left(P\g_{n,s}\g_{n,t}-P\g_{n,s}P\g_{n,t}\right).
\]

Next, we show that $D_{n,2}$ is asymptotically negligible, since
by Lemma \eqref{lem:Term1}
\begin{align*}
D_{n,2}:=\frac{1}{\sqrt{n}\d_{n}^{2}}\GG_{n}\left(g_{\t_{0}+s\d_{n},h_{0}}-g_{\t_{0},h_{0}}\right) & =O_{p}\left(\frac{1}{\sqrt{n}\d_{n}^{2}}\d_{n}^{\frac{3}{2}}\right)=O_{p}\left(\sqrt{\frac{\d_{n}}{n}}\right)=o_{p}\left(1\right)
\end{align*}

Finally, for $D_{n,3}$ we show that, based on Lemma \eqref{lem:Term_3},
\begin{align*}
D_{n,3}=\  & \frac{1}{\d_{n}^{2}}P\left(g_{\t_{0}+s\d_{n},\hat{h}}-g_{\t_{0},\hat{h}}\right)\\
=\  & \frac{1}{\sqrt{n}\d_{n}^{2}}\GG_{n}\left(\psi_{n,\t_{0}+s\d_{n}}-\psi_{n,\t_{0}}\right)+\frac{1}{\d_{n}^{2}}P\left(\psi_{n,\t_{0}+s\d_{n}}-\psi_{n,\t_{0}}\right)\\
=\  & \frac{1}{\sqrt{n}\d_{n}^{2}}O_{p}\left(b_{n}^{-\frac{1}{2}}\d_{n}\right)+\frac{1}{\d_{n}^{2}}\left(-s^{'}Vs\cd\d_{n}^{2}+b_{n}^{2}\d_{n}\cd A_{1}^{'}s+o\left(\d_{n}^{2}\right)+o\left(b_{n}^{2}\d_{n}\right)\right)\\
=\  & -s^{'}Vs+A_{1}^{'}s+O_{p}\left(\frac{1}{\sqrt{nb_{n}}\d_{n}}\right)+o\left(b_{n}^{2}\d_{n}^{-1}\right)+o\left(1\right)\\
=\  & -s^{'}Vs+A_{1}^{'}s+o_{p}\left(1\right)
\end{align*}
since $\left(nb_{n}\right)^{-\frac{1}{2}}=n^{-\frac{d+4}{2\left(d+6\right)}}=o\left(\d_{n}\right)=o\left(n^{-\frac{4}{d+6}}\left(\log n\right)^{\frac{1}{3}}\right)$.

Combining $D_{n,1},D_{n,2}$ and $D_{n,3}$, we conclude that 
\[
\frac{1}{\sqrt{n}\d_{n}^{2}}\GG_{n}\left(g_{\t_{0}+s\d_{n},\hat{h}}-g_{\t_{0},\hat{h}}\right)\dto G\left(s\right)+A_{1}^{'}s-s^{'}Vs
\]
and thus by the argmax continuous mapping theorem (VW Theorem 3.2.2),
we have
\[
\d_{n}^{-1}\left(I-\t_{0}\t_{0}^{'}\right)\left(\hat{\t}-\t_{0}\right)\dto\arg\max_{s:s^{'}\t_{0}=0}G\left(s\right)+A_{1}^{'}s-s^{'}Vs.
\]
\end{proof}

\subsection{Proof of Theorem \ref{thm:Thm_Bin_Dist}(iii)}
\begin{proof}
For $d\geq6$ with $b_{n}\sim n^{-\frac{2}{d+6}}$, we note that 
\[
\d_{n}:=\norm{\hat{\t}-\t_{0}}=O_{p}\left(n^{-\frac{4}{d+6}}\left(\log n\right)^{\frac{1}{3}}\right)=O\left(b_{n}^{2}\right)
\]
and moreover
\[
\d_{n}\sim b_{n}^{2}>>\left(nb_{n}\right)^{-\frac{1}{2}},\quad\d_{n}\sim b_{n}^{2}>>n^{-\frac{1}{3}}a_{n}^{\frac{2}{3}}.
\]
The rest of the proof can be obtained by an easy adaption of the proof
for Theorem \ref{thm:Thm_Bin_Dist}(ii) above. Specifically, we observe
that for $D_{n,1}$, the inequality \eqref{eq:P_G2} becomes
\[
PG_{2,n}^{2}\leq Cn^{-1}\d_{n}^{-4}a_{n}^{2}\d_{n}=C\left(n^{-\frac{1}{3}}a_{n}^{\frac{2}{3}}\d_{n}^{-1}\right)^{3}=o\left(1\right).
\]
Hence,
\[
\d_{n}^{-1}\left(I-\t_{0}\t_{0}^{'}\right)\left(\hat{\t}-\t_{0}\right)\dto\arg\max_{s:s^{'}\t_{0}=0}A_{1}^{'}s-s^{'}Vs=A_{1}.
\]
\end{proof}

\subsection{Proof of Lemma \ref{thm:MMI_Rate}}
\begin{proof}
The proofs of Lemma \ref{lem:Term1} and Lemma \ref{lem:Term2} are
essentially unchanged. For the term $P\left(g_{\t,\hat{h}}-g_{\t_{0},\hat{h}}\right)$,
we note that
\begin{align*}
P\left(g_{\t,\hat{h}}-g_{\t_{0},\hat{h}}\right) & =P\left(g_{\t,\hat{h}}-g_{\t_{0},\hat{h}}-g_{\t,h_{0}}+g_{\t_{0},h_{0}}\right)+P\left(g_{\t,h_{0}}-g_{\t_{0},h_{0}}\right)
\end{align*}
where
\begin{align*}
 & P\left|g_{\t,\hat{h}}-g_{\t_{0},\hat{h}}-g_{\t,h_{0}}+g_{\t_{0},h_{0}}\right|\\
\leq\  & P\left|\g\left(\hat{h}\left(X\right)\right)-\g\left(h_{0}\left(X\right)\right)\right|\left|\prod_{j}\ind\left\{ X_{j}^{'}\t\geq0\right\} -\prod_{j}\ind\left\{ X_{j}^{'}\t_{0}\geq0\right\} \right|\\
\leq\  & MP\left|\hat{h}\left(X\right)-h_{0}\left(X\right)\right|\left|\ind\left\{ X_{j\left(X\right)}^{'}\t\geq0\right\} -\ind\left\{ X_{j\left(X\right)}^{'}\t_{0}\geq0\right\} \right|\\
 & \quad\quad\text{for some \ensuremath{j\left(X\right)} with probability 1 for \ensuremath{\t\ }sufficiently close to \ensuremath{\t_{0}}}\\
\leq\  & Ca_{n}\norm{\t-\t_{0}}
\end{align*}
and 
\begin{align*}
P\left(g_{\t,h_{0}}-g_{\t_{0},h_{0}}\right) & =-\left(\t-\t_{0}\right)^{'}V\left(\t-\t_{0}\right)+o\left(\norm{\t-\t_{0}}^{2}\right).
\end{align*}
Hence,
\begin{align*}
P\left(g_{\t,\hat{h}}-g_{\t_{0},\hat{h}}\right) & =-\left(\t-\t_{0}\right)^{'}V\left(\t-\t_{0}\right)+O\left(a_{n}\d\right)+o\left(\norm{\t-\t_{0}}^{2}\right).
\end{align*}

Combining this with Lemma \ref{lem:Term1} and Lemma \ref{lem:Term2},
we conclude that Conditions B1-B4 in DvK can be verified with the
smallest $\d_{n}$ such that
\[
\d_{n}=\max\left\{ n^{-1},n^{-\frac{1}{3}}a_{n}^{\frac{2}{3}},\ a_{n}\right\} =a_{n}.
\]
\end{proof}

\subsection{Proof of Lemma }
\begin{proof}
(i) and (ii) are immediate. For (iii), notice that 
\begin{align*}
\l\left(t\right) & =\frac{d}{dt}\int_{-\infty}^{t}\int_{-\infty}^{\infty}K_{d}\left(u\right)du_{1}du_{-1}=\int_{-\infty}^{\infty}K_{d}\left(t,u_{-1}\right)du_{-1}.
\end{align*}
Hence, 
\begin{align*}
\int_{-\infty}^{\infty}t^{j}\l\left(t\right)dt & =\int_{-\infty}^{\infty}\int_{-\infty}^{\infty}t^{j}K_{d}\left(t,u_{-1}\right)du_{-1}dt=\int_{-\infty}^{\infty}u_{1}^{j}K_{d}\left(u\right)du=0,
\end{align*}
and
\begin{align*}
\int_{-\infty}^{\infty}t^{s}\l\left(t\right)dt & =\int_{-\infty}^{\infty}u_{1}^{s}K_{d}\left(u\right)du=R_{s}>0.
\end{align*}
\end{proof}

\subsection{Proof of Theorem \ref{thm:MMI_3rates}}
\begin{proof}
Following the proof of Lemma \ref{thm:MMI_Rate}, we see now
\begin{align*}
P\left(g_{\t,\hat{h}}-g_{\t_{0},\hat{h}}\right) & =-\left(\t-\t_{0}\right)^{'}V\left(\t-\t_{0}\right)+O\left(u_{n}+v_{n}\right)\d+o\left(\d^{2}\right)
\end{align*}
so that
\[
\d_{n}=\max\left\{ n^{-1},n^{-\frac{1}{3}}a_{n}^{\frac{2}{3}},\ u_{n},v_{n}\right\} =\max\left\{ n^{-\frac{1}{3}}a_{n}^{\frac{2}{3}},\ u_{n},v_{n}\right\} .
\]
\end{proof}

\section{Online Appendix}

\subsection{Proof of Corollary \ref{cor:h0_dif_bound}}
\begin{proof}
Viewing $F\left(\rest{\e}x\right)$ as a function of $\left(\e,x\right)$,
we write $\frac{\p}{\p\e}F\left(\e|x\right)$ and $\frac{\p}{\p x}F\left(\e|x\right)$
as derivatives w.r.t. its two arguments. Since $h_{0}\left(x\right)=F\left(\rest{x^{'}\t_{0}}x\right)$,
we have
\begin{align*}
\left|\frac{\p}{\p x_{j}}h_{0}\left(x\right)\right| & =\left|f\left(\rest{x^{'}\t_{0}}x\right)\t_{0,j}+\rest{\frac{\p}{\p x_{j}}F\left(\rest{\e}x\right)}_{\e=x^{'}\t_{0}}\right|\\
 & \leq\left|f\left(\rest{x^{'}\t_{0}}x\right)\right|\cd\left|\t_{0,j}\right|+\left|\frac{\p}{\p x_{j}}F\left(\rest{x^{'}\t_{0}}x\right)\right|\leq M\cd1+M=2M,
\end{align*}
and 
\begin{align*}
\left|\frac{\p^{2}}{\p x_{k}\p x_{j}}h_{0}\left(x\right)\right| & =\left|\frac{\p}{\p\e}f\left(\rest{x^{'}\t_{0}}x\right)\t_{0,j}\t_{0,k}+\frac{\p}{\p x_{k}}f\left(\rest{x^{'}\t_{0}}x\right)\t_{0,j}+\frac{\p^{2}}{\p x_{k}\p x_{j}}F\left(\rest{x^{'}\t_{0}}x\right)\right|\\
 & \leq M\cd1\cd1+M\cd1+M=2M.
\end{align*}
\end{proof}

\subsection{\label{subsec:Pf_Term1}Proof of Lemma \ref{lem:Term1}}
\begin{proof}
Define ${\cal \cG}_{1,\d}:=\left\{ g_{\t,h_{0}}-g_{\t_{0},h_{0}}:\ \norm{\t-\t_{0}}\leq\d\right\} $,
which has an envelope $G_{1,\d}$:
\begin{align*}
\left|g_{\t,h_{0}}-g_{\t_{0},h_{0}}\right| & =\left|h_{0}\left(x\right)\right|\left|\ind\left\{ x^{'}\t\geq0\right\} -\ind\left\{ x^{'}\t_{0}\geq0\right\} \right|\\
 & =\left|h_{0}\left(x\right)\right|\left(\ind\left\{ x^{'}\t\geq0>x^{'}\t_{0}\right\} +\ind\left\{ x^{'}\t_{0}\geq0>x^{'}\t\right\} \right)\\
 & =\left|h_{0}\left(x\right)\right|\left(\ind\left\{ x^{'}\t_{0}+x^{'}\left(\t-\t_{0}\right)\geq0>x^{'}\t_{0}\right\} +\ind\left\{ x^{'}\t_{0}\geq0>x^{'}\t_{0}+x^{'}\left(\t-\t_{0}\right)\right\} \right)\\
 & \leq\left|h_{0}\left(x\right)\right|\left(\ind\left\{ x^{'}\t_{0}+\norm x\norm{\t-\t_{0}}\geq0>x^{'}\t_{0}\right\} +\ind\left\{ x^{'}\t_{0}\geq0>x^{'}\t_{0}-\norm x\norm{\t-\t_{0}}\right\} \right)\\
 & \leq\left|h_{0}\left(x\right)\right|\left(\ind\left\{ 0>x^{'}\t_{0}\geq-\norm x\norm{\t-\t_{0}}\right\} +\ind\left\{ \norm x\norm{\t-\t_{0}}>x^{'}\t_{0}\geq0\right\} \right)\\
 & =\left|h_{0}\left(x\right)\right|\ind\left\{ \left|x^{'}\t_{0}\right|\leq\norm x\norm{\t-\t_{0}}\right\} 
\end{align*}
Whenever $\left|x^{'}\t_{0}\right|\leq\norm x\norm{\t-\t_{0}}<\norm{\t-\t_{0}}$,
we have
\[
0\in\left[x^{'}\t_{0}-\norm{\t-\t_{0}},x^{'}\t_{0}+\norm{\t-\t_{0}}\right]=\left[\left(x-\norm{\t-\t_{0}}\t_{0}\right)^{'}\t_{0},\left(x+\norm{\t-\t_{0}}\t_{0}\right)^{'}\t_{0}\right].
\]
which implies that
\begin{equation}
h_{0}\left(x-\norm{\t-\t_{0}}\t_{0}\right)\leq0\leq h_{0}\left(x+\norm{\t-\t_{0}}\t_{0}\right).\label{eq:0_btw_h0}
\end{equation}
By Lemma \ref{cor:h0_dif_bound}, 
\begin{align*}
h_{0}\left(x+\norm{\t-\t_{0}}\t_{0}\right) & \leq h_{0}\left(x\right)+\sup_{x}\norm{\Dif_{x}h_{0}\left(x\right)}\cd\norm{\t-\t_{0}}\leq h_{0}\left(x\right)+M\norm{\t-\t_{0}},\\
h_{0}\left(x+\norm{\t-\t_{0}}\t_{0}\right) & \geq h_{0}\left(x\right)-\sup_{x}\norm{\Dif_{x}h_{0}\left(x\right)}\cd\norm{\t-\t_{0}}\geq h_{0}\left(x\right)-M\norm{\t-\t_{0}},
\end{align*}
and thus \eqref{eq:0_btw_h0} implies that
\[
h_{0}\left(x\right)-M\norm{\t-\t_{0}}\leq0\leq h_{0}\left(x\right)+M\norm{\t-\t_{0}},
\]
which further implies that
\[
\left|h_{0}\left(x\right)\right|\leq M\norm{\t-\t_{0}}.
\]
Hence, 
\begin{align*}
\left|g_{\t,h_{0}}-g_{\t_{0},h_{0}}\right| & \leq\left|h_{0}\left(x\right)\right|\ind\left\{ \left|x^{'}\t_{0}\right|\leq\norm x\norm{\t-\t_{0}}\right\} \\
 & \leq M\norm{\t-\t_{0}}\ind\left\{ \left|x^{'}\t_{0}\right|\leq\norm x\norm{\t-\t_{0}}\right\} \\
 & \leq C\d\ind\left\{ \left|x^{'}\t_{0}\right|\leq\norm x\d\right\} =:G_{1,\d}.
\end{align*}
Now, since $X_{i}/\norm{X_{i}}$ is uniformly distributed on $\S^{d-1}$,
\begin{align*}
PG_{1,\d}^{2} & =\E\left[C^{2}\d^{2}\ind\left\{ \left|X_{i}^{'}\t_{0}\right|\leq\norm{X_{i}}\d\right\} \right]\\
 & =C^{2}\d^{2}\P\left(\left|\frac{X_{i}^{'}}{\norm{X_{i}}}\t_{0}\right|\leq\d\right)\\
 & \leq C^{2}\d^{3}
\end{align*}
Now, since ${\cal \cG}_{1,\d}\subseteq\cG$, we have $\mathscr{N}\left(\e,\cG_{1,\d},L_{2}\left(P\right)\right)\leq\mathscr{N}\left(\e,\cG,L_{2}\left(P\right)\right)$
and by Lemma \ref{lem:J_uni_finite}
\[
J_{1,\d}:=\int_{0}^{1}\sqrt{1+\log\mathscr{N}\left(\e,\cG_{1,},L_{2}\left(P\right)\right)}d\e\leq J<\infty.
\]
Then, by VW Theorem 2.14.1, we have
\[
P\sup_{g\in{\cal \cG}_{1,\d}}\left|\GG_{n}\left(g\right)\right|\leq J_{1,\d}\sqrt{PG_{1,\d}^{2}}\leq J_{1}C\d\sqrt{\d}=M_{1}\d\sqrt{\d}.
\]
\end{proof}

\subsection{\label{subsec:Pf_Term2}Proof of Lemma \ref{lem:Term2}}
\begin{proof}
Define ${\cal \cG}_{2,\d,n}:=\left\{ g_{\t,h}-g_{\t_{0},h}-g_{\t,h_{0}}+g_{\t_{0},h_{0}}:\ \norm{\t-\t_{0}}\leq\d,\norm{h-h_{0}}_{\infty}\leq Ka_{n}\right\} $,
which has an envelope function $G_{2,\d,n}$ given by
\begin{align*}
 & \left|g_{\t,h}-g_{\t_{0},h}-g_{\t,h_{0}}+g_{\t_{0},h_{0}}\right|\\
= & \left|h\left(x\right)-h_{0}\left(x\right)\right|\left|\ind\left\{ x^{'}\t\geq0\right\} -\ind\left\{ x^{'}\t_{0}\geq0\right\} \right|\\
\leq & \left|h\left(x\right)-h_{0}\left(x\right)\right|\ind\left\{ \left|x^{'}\t_{0}\right|\leq\norm x\norm{\t-\t_{0}}\right\} \\
\leq & Ka_{n}\ind\left\{ \left|x^{'}\t_{0}\right|\leq\norm x\d\right\} \\
=: & G_{2,n,\d}
\end{align*}
with
\begin{align*}
PG_{2,n,\d}^{2} & =K^{2}a_{n}^{2}\P\left(\left|\frac{X_{i}^{'}}{\norm{X_{i}}}\t_{0}\right|\leq\d\right)\leq Ca_{n}^{2}\d.
\end{align*}
Since ${\cal \cG}_{2,\d,n}\subseteq\cG-{\cal \cG}_{1,\d}:=\left\{ g-\tilde{g}:g\in\cG,\tilde{g}\in\cG_{1,\d}\right\} $,
by Lemma 9.14 of \citet{kosorok2007introduction}, ${\cal \cG}_{2,\d,n}$
must also have bounded uniform entropy integrals. Hence,
\[
J_{2}:=\int_{0}^{1}\sqrt{1+\log\mathscr{N}\left(\e,\cG_{2},L_{2}\left(P\right)\right)}d\e<\infty,
\]
and by VW Theorem 2.14.1, 
\[
P\sup_{g\in{\cal \cG}_{2,\d,n}}\norm{\GG_{n}\left(g\right)}\leq J_{2,\d}\sqrt{PG_{2,n,\d}^{2}}\leq J_{2}Ca_{n}\sqrt{\d}=Ma_{n}\sqrt{\d}.
\]
\end{proof}

\end{document}